\documentclass[runningheads,orivec, envcountsect, envcountsame]{llncs}
\usepackage[T1]{fontenc}

\usepackage{ifthen}
\newboolean{debug}
\setboolean{debug}{false}
\newboolean{conference}
\setboolean{conference}{false}
\usepackage{orcidlink}
\usepackage{xparse}

\usepackage{CJKutf8}

\usepackage{fancyvrb} %

\usepackage{colonequals}
\usepackage{amsmath}
\usepackage{amsfonts,amssymb}
\usepackage{mathtools}
\usepackage{mathrsfs} %
\usepackage{mathcommand}

\usepackage{stmaryrd}
\usepackage{scalerel}

\usepackage{hyperref}
\usepackage{cleveref}
\usepackage{crossreftools}
\usepackage{amsthm}
\newtheorem*{proposition*}{Proposition}
\newtheorem*{theorem*}{Theorem}
\newtheorem*{corollary*}{Corollary}

\usepackage{cite}

\crefname{theorem}{Thm.}{Thms.}
\crefname{proposition}{Prop.}{Props.}
\crefname{lemma}{Lem.}{Lems.}
\crefname{corollary}{Cor.}{Cors.}
\crefname{definition}{Def.}{Defs.}
\crefname{section}{\S}{\S}
\crefname{appendix}{\S}{\S}
\crefname{subsubsubappendix}{\S}{\S}
\crefname{figure}{Fig.}{Figs.}

\usepackage{algorithm, algpseudocode}

\usepackage[svgnames]{xcolor}
\definecolor[named]{ACMBlue}{cmyk}{1,0.1,0,0.1}
\definecolor[named]{ACMYellow}{cmyk}{0,0.16,1,0}
\definecolor[named]{ACMOrange}{cmyk}{0,0.42,1,0.01}
\definecolor[named]{ACMRed}{cmyk}{0,0.90,0.86,0}
\definecolor[named]{ACMLightBlue}{cmyk}{0.49,0.01,0,0}
\definecolor[named]{ACMGreen}{cmyk}{0.20,0,1,0.19}
\definecolor[named]{ACMPurple}{cmyk}{0.55,1,0,0.15}
\definecolor[named]{ACMDarkBlue}{cmyk}{1,0.58,0,0.21}
\hypersetup{colorlinks=true,breaklinks=true,
    linkcolor=ACMPurple,
    citecolor=ACMPurple,
    urlcolor=ACMDarkBlue,
    filecolor=ACMDarkBlue}

\usepackage{diagbox}
\usepackage[normalem]{ulem} %

\MakeRobust{\ref}%

\usepackage{tikz}
\usetikzlibrary{calc, trees, bending, backgrounds, shapes, shapes.geometric, fit, arrows, arrows.meta, positioning, graphs, quotes}
\pgfdeclarelayer{bg}    %
\pgfsetlayers{bg,main}  %

\ifthenelse{\boolean{debug}}{%
\usepackage[xcolor, cleveref, notion, quotation, composition]{knowledge}
\knowledgestyle {notion}{color=green!70!black!100}
\knowledgestyle {intro}{color=blue}
\knowledgestyle {kl unknown}{color=red}
\knowledgestyle {kl unknown cont}{color=red}
}{
\usepackage[xcolor, cleveref, notion, quotation, paper]{knowledge}
}

\knowledgeconfigureenvironment{theorem,lemma,proof}{}
\knowledgeconfigure{diagnose line=true, diagnose bar=true}
\knowledgeconfigure{protect quotation={graph,figure,tikzpicture,tikzcd}} %

\usepackage{ebproof}
\ebproofnewstyle{small}{
    right label template = \scriptsize\inserttext,
}

\usepackage{stackengine}

\makeatletter
\let\orgdescriptionlabel\descriptionlabel
\renewcommand*{\descriptionlabel}[1]{%
    \let\orglabel\label
    \let\label\@gobble
    \phantomsection
    \edef\@currentlabel{#1}%
    \let\label\orglabel
    \orgdescriptionlabel{#1}%
}
\makeatother

\usepackage{centernot}

\usepackage{etoolbox}
\makeatletter
\let\original@footnote\footnote
\newcommand{\align@footnote}[1]{%
    \ifmeasuring@
        \chardef\@tempfn=\value{footnote}%
        \footnotemark
        \setcounter{footnote}{\@tempfn}%
    \else
        \iffirstchoice@
        \original@footnote{#1}%
        \fi
    \fi}
\pretocmd{\start@align}{\let\footnote\align@footnote}{}{}
\makeatother

\usepackage{makecell}

\usepackage[most, hooks]{tcolorbox}

\usepackage{xnewcommand} %

\newcounter{mylabelcounter}
\makeatletter
\newcommand{\labeltext}[2]{%
#1\refstepcounter{mylabelcounter}%
\immediate\write\@auxout{%
  \string\newlabel{#2}{{1}{\thepage}{{\unexpanded{#1}}}{mylabelcounter.\number\value{mylabelcounter}}{}}%
}%
}
\makeatother

\newcommand{\proofcasebase}[1]{\colorbox{black!10}{\vphantom{()}\hspace{.2em}\textsf{#1}\hspace{.2em}}~}
\newcommand{\proofcase}[1]{\noindent\setlength{\fboxsep}{0pt}\proofcasebase{#1}}
\newcommand{\subproofcase}[1]{\noindent\setlength{\fboxsep}{0pt}\quad\proofcasebase{#1}}
\newcommand{\subsubproofcase}[1]{\noindent\setlength{\fboxsep}{0pt}\qquad\proofcasebase{#1}}

\DeclareSymbolFont{extraup}{U}{zavm}{m}{n}
\DeclareMathSymbol{\varspadesuit}{\mathalpha}{extraup}{85}
\DeclareMathSymbol{\varheartsuit}{\mathalpha}{extraup}{86}
\DeclareMathSymbol{\vardiamondsuit}{\mathalpha}{extraup}{87}
\DeclareMathSymbol{\varclubsuit}{\mathalpha}{extraup}{88}

\DeclareMathSymbol{\shortminus}{\mathalpha}{operators}{`-}

\NewEnviron{inlinealign*}{\\[-2.8\baselineskip]\begin{align*}\BODY\end{align*}\\[-2.8\baselineskip]}

\DeclareRobustCommand{\lipicsEnd}{%
	\leavevmode\unskip\penalty9999 \hbox{}\nobreak\hfill
	\quad\hbox{$\lrcorner$}%
}
\NewEnviron{claimproof}{%
\begin{proof}\BODY\end{proof}}
\DeclareRobustCommand{\claimqedhere}{\qedhere}

\NewEnviron{hidden}{\BODY} %
\usepackage[backgroundcolor=orange!20, textsize=tiny]{todonotes}
\definecolor{yoshikieditcolor}{RGB}{200,230,200}

\NewEnviron{sideyoshiki}{\todo[caption={},backgroundcolor=yoshikieditcolor, size=\tiny]{YN: \BODY}}
\NewEnviron{yoshiki}{\todo[inline,caption={},backgroundcolor=yoshikieditcolor, size=\footnotesize]{YN: \BODY}} %

\usepackage{color}

\tikzstyle{mynode} = [inner sep = 1.5pt, fill= gray!20,  font = \scriptsize, draw, circle]
\tikzstyle{mysmallnode} = [inner sep = 1.pt, fill= gray!20, font = \scriptsize, draw, circle]

\tikzset{earrow/.style={>={{[flex] Latex[length=.08cm, width=2.5pt]}}}}
\tikzset{homoarrow/.style={earrow, line width = .8pt, densely dotted, color = olive, opacity=0.8}}

\DeclarePairedDelimiter\set{\{}{\}}
\DeclarePairedDelimiter\tuple{\langle}{\rangle}
\knowledgenewrobustcmd{\range}[2]{\cmdkl{[}#1..#2\cmdkl{]}}
\knowledgenewrobustcmd{\rangeone}[1]{\cmdkl{[}#1\cmdkl{]}}

\knowledgenewrobustcmd{\card}{\#}

\knowledgenewrobustcmd{\pfun}{\mathbin{\cmdkl{\rightharpoonup}}}
\knowledgenewrobustcmd{\fdom}{\cmdkl{\operatorname{dom}}}

\knowledgenewrobustcmd{\dcup}{\mathbin{\cmdkl{\sqcup}}}
\newcommand{\const}[1]{\mathsf{#1}}
\newcommand{\bl}{\cdot}
\newcommand{\defeq}{\coloneq}

\knowledgenewrobustcmd{\nat}{\cmdkl{\mathbb{N}}}
\knowledgenewrobustcmd{\Z}{\cmdkl{\mathbb{Z}}}

\knowledgenewrobustcmd{\pset}{\cmdkl{\wp}}

\knowledgenewrobustcmd{\univ}[1]{\cmdkl{|}#1\cmdkl{|}}

\NewDocumentCommand\la{O{1}}{%
    \ifcase#1 undefined
    \or L
    \or K
    \else undefined \fi
}
\NewDocumentCommand\rel{O{1}}{%
    \ifcase#1 undefined
    \or R
    \or S
    \else undefined \fi
}
\knowledgenewrobustcmd{\diagonal}{\cmdkl{\triangle}}

\knowledgenewrobustcmd{\eps}{\cmdkl{\varepsilon}}
\NewDocumentCommand\word{O{1}}{%
    \ifcase#1 undefined
    \or w
    \or v
    \else undefined \fi
}
\knowledgenewrobustcmd{\len}[1]{\cmdkl{\|}#1\cmdkl{\|}}

\DeclarePairedDelimiter\dia{\langle}{\rangle}
\DeclarePairedDelimiter\bo{[}{]}
\newcommand{\truec}{\mathtt{T}}
\newcommand{\falsec}{\mathtt{F}}

\newcommand{\union}{\mathbin{+}}
\newcommand{\compo}{\mathbin{;}}

\newcommand{\emp}{\const{0}}
\newcommand{\id}{\const{1}}

\DeclareMathOperator*{\bigcompo}{\scalerel*{\compo}{\sum}}
\newcommand{\com}[1]{\overline{#1}}
\newcommand{\dom}{\mathtt{\mathop{d}}}
\newcommand{\adom}{\mathtt{\mathop{a}}}

\newcommand{\capid}{\cap_{\id}}
\newcommand{\capcomid}{\cap_{\com{\id}}}

\newcommand{\vsig}{\mathbb{A}}
\knowledgenewrobustcmd{\exprvsig}{\cmdkl{\vsig}}

\NewDocumentCommand\afml{O{1}}{%
    \ifcase#1 undefined
    \or p
    \or q
    \else undefined \fi
}
\NewDocumentCommand\fml{O{1}}{%
    \ifcase#1 undefined
    \or \varphi
    \or \psi
    \or \rho
    \else undefined \fi
}
\NewDocumentCommand\fmlset{O{1}}{%
    \ifcase#1 undefined
    \or \Gamma
    \or \Delta
    \or \Lambda
    \else undefined \fi
}
\newrobustcmd{\fmlclass}{\Phi}
\knowledgenewrobustcmd{\fmlclassgen}[1]{\fmlclass^{#1}}

\newcommand{\psig}{\mathbb{P}}
\knowledgenewrobustcmd{\exprpsig}{\cmdkl{\psig}}

\NewDocumentCommand\term{O{1}}{%
    \ifcase#1 undefined
    \or t
    \or s
    \or u
    \else undefined \fi
}
\NewDocumentCommand\aterm{O{1}}{%
    \ifcase#1 undefined
    \or a
    \or b
    \else undefined \fi
}
\newrobustcmd{\termclass}{\Pi}
\knowledgenewrobustcmd{\termclassgen}[1]{\termclass^{#1}}
\knowledgenewrobustcmd{\termsubst}[1]{\cmdkl{[}#1\cmdkl{]}}

\NewDocumentCommand\expr{O{1}}{%
    \ifcase#1 undefined
    \or E
    \or F
    \or G
    \else undefined \fi
}
\newrobustcmd{\exprclass}{\mathcal{E}}
\knowledgenewrobustcmd{\exprsubst}[1]{\cmdkl{[}#1\cmdkl{]}}

\newcommand{\sig}{\mathrm{S}}
\knowledgenewrobustcmd{\sigREwLA}{\cmdkl{\sig}_{\text{"\REwLA"}}}

\NewDocumentCommand\kframe{O{1}}{%
    \ifcase#1 undefined
    \or \mathfrak{F}
    \else undefined \fi
}
\NewDocumentCommand\struc{O{1}}{%
    \ifcase#1 undefined
    \or \mathfrak{A}
    \or \mathfrak{B}
    \else undefined \fi
}
\knowledgenewrobustcmd{\strucsubst}[1]{\cmdkl{[}#1\cmdkl{]}}

\knowledgenewrobustcmd{\strucuniv}{\cmdkl{U}}

\knowledgenewrobustcmd{\wordstruc}{\cmdkl{\struc}}

\knowledgenewrobustcmd{\generatedsubstruc}[2]{#1\cmdkl{[}#2\cmdkl{..]}}

\NewDocumentCommand\vertex{O{1}}{%
    \ifcase#1 undefined
    \or c
    \or d
    \or e
    \else undefined \fi
}

\NewDocumentCommand\alg{O{1}}{%
    \ifcase#1 undefined
    \or \mathcal{A}
    \or \mathcal{B}
    \else undefined \fi
}

\NewDocumentCommand\algclass{O{1}}{%
    \ifcase#1 undefined
    \or \mathcal{C}
    \else undefined \fi
}
\NewDocumentCommand\strucclass{O{1}}{%
    \ifcase#1 undefined
    \or \mathcal{C}
    \else undefined \fi
}

\NewDocumentCommand\valclass{O{1}}{%
    \ifcase#1 undefined
    \or \mathcal{V}
    \else undefined \fi
}

\knowledgenewrobustcmd{\sem}[2]{\cmdkl{\llbracket}#1\cmdkl{\rrbracket}^{#2}} %

\knowledgenewrobustcmd{\GREL}[1]{\cmdkl{\mathsf{GREL}}^{#1}}
\knowledgenewrobustcmd{\GRELpreorder}[1]{\cmdkl{\mathsf{GREL}}_{\cmdkl{\lesssim}}^{#1}}
\knowledgenewrobustcmd{\GRELpartialorder}[1]{\cmdkl{\mathsf{GREL}}_{\cmdkl{\le}}^{#1}}
\knowledgenewrobustcmd{\GRELfinpartialorder}[1]{\cmdkl{\mathsf{GREL}}_{\cmdkl{\le_{\mathrm{fin}}}}^{#1}}
\knowledgenewrobustcmd{\GRELspartialorder}[1]{\cmdkl{\mathsf{GREL}}_{\cmdkl{<}}^{#1}}
\knowledgenewrobustcmd{\GRELfinlin}[1]{\cmdkl{\mathsf{GREL}}_{\cmdkl{\le_{\mathrm{fin\shortminus lin}}}}^{#1}}
\knowledgenewrobustcmd{\GRELsfinlin}[1]{\cmdkl{\mathsf{GREL}}_{\cmdkl{<_{\mathrm{fin\shortminus lin}}}}^{#1}}

\knowledgenewrobustcmd{\GRELstfinlin}[1]{\cmdkl{\mathsf{GREL}}_{\cmdkl{\le_{\mathrm{fin\shortminus lin}}^{\mathrm{st}}}}^{#1}}
\knowledgenewrobustcmd{\GRELstsfinlin}[1]{\cmdkl{\mathsf{GREL}}_{\cmdkl{<_{\mathrm{fin\shortminus lin}}^{\mathrm{st}}}}^{#1}}

\knowledgenewrobustcmd{\REL}[1]{\cmdkl{\mathsf{REL}}^{#1}}

\knowledgenewrobustcmd{\EqT}{\cmdkl{\mathrm{EqT}}}
\knowledgenewrobustcmd{\Th}{\cmdkl{\mathrm{Th}}}

\knowledgenewrobustcmd{\wlang}{\cmdkl{\mathcal{L}}}
\knowledgenewrobustcmd{\mlang}{\cmdkl{\mathcal{M}}}

\knowledgenewrobustcmd{\swlang}{\cmdkl{\mathcal{L}_{\mathrm{s}}}}
\knowledgenewrobustcmd{\smlang}{\cmdkl{\mathcal{M}_{\mathrm{s}}}}

\knowledgenewrobustcmd{\modelsfml}{\mathrel{\cmdkl{\models}}}

\newrobustcmd{\Hilbertstyle}{\mathcal{H}}

\knowledgenewrobustcmd{\termclassREwLA}[1]{\cmdkl{\termclass}_{\cmdkl{\REwLA}}^{#1}}
\knowledgenewrobustcmd{\semREwLA}[2]{\cmdkl{\llbracket}#1\cmdkl{\rrbracket}^{#2}}
\knowledgenewrobustcmd{\termclassREwLAp}[1]{\cmdkl{\termclass}_{\cmdkl{\REwLAp}}^{#1}}
\knowledgenewrobustcmd{\semREwLAp}[2]{\cmdkl{\llbracket}#1\cmdkl{\rrbracket}^{#2}}
\knowledgenewrobustcmd{\sigPDL}{\cmdkl{\sig}_{\cmdkl{\PDL}}}
\knowledgenewrobustcmd{\termclassPDL}[1]{\cmdkl{\termclass}_{\cmdkl{\PDL}}^{#1}}
\knowledgenewrobustcmd{\fmlclassPDL}[1]{\cmdkl{\fmlclass}_{\cmdkl{\PDL}}^{#1}}
\knowledgenewrobustcmd{\exprclassPDL}[1]{\cmdkl{\exprclass}_{\cmdkl{\PDL}}^{#1}}

\knowledgenewrobustcmd{\HilbertstylePDL}{\Hilbertstyle^{\PDL}}
\knowledgenewrobustcmd{\vdashPDL}{\mathrel{\cmdkl{\vdash}}_{\cmdkl{\HilbertstylePDL}}}
\knowledgenewrobustcmd{\PDLmodels}{\mathrel{\cmdkl{\models}}} %
\knowledgenewrobustcmd{\semPDL}[2]{\cmdkl{\llbracket}#1\cmdkl{\rrbracket}^{#2}}
\knowledgenewrobustcmd{\HilbertstyleHKTPDL}{\Hilbertstyle^{\HKTPDL}}
\knowledgenewrobustcmd{\vdashHKTPDL}{\mathrel{\cmdkl{\vdash}}_{\cmdkl{\HilbertstyleHKTPDL}}}

\knowledgenewrobustcmd{\sigPDLREwLAp}{\cmdkl{\sig}_{\cmdkl{\PDLREwLAp}}}
\knowledgenewrobustcmd{\termclassPDLREwLAp}[1]{\cmdkl{\termclass}_{\cmdkl{\PDLREwLAp}}^{#1}}
\knowledgenewrobustcmd{\fmlclassPDLREwLAp}[1]{\cmdkl{\fmlclass}_{\cmdkl{\PDLREwLAp}}^{#1}}
\knowledgenewrobustcmd{\exprclassPDLREwLAp}[1]{\cmdkl{\exprclass}_{\cmdkl{\PDLREwLAp}}^{#1}}

\knowledgenewrobustcmd{\semPDLREwLAp}[2]{\cmdkl{\llbracket}#1\cmdkl{\rrbracket}^{#2}}
\knowledgenewrobustcmd{\HilbertstylePDLREwLAp}{\Hilbertstyle^{\PDLREwLAp}_{\le_{\mathrm{fin\shortminus lin}}}}
\knowledgenewrobustcmd{\vdashPDLREwLAp}{\mathrel{\cmdkl{\vdash}}_{\cmdkl{\HilbertstylePDLREwLAp}}}

\knowledgenewrobustcmd\trPDLREwLAptoREwLAp{\cmdkl{\operatorname{Tr}}}

\knowledgenewrobustcmd{\fromifreePDL}[1]{#1\exprsubst{\fromifreePDLsubst}}
\knowledgenewrobustcmd{\fromifreePDLsubst}{\cmdkl{\Theta_0}}
\knowledgenewrobustcmd{\HilbertstylestPDLREwLAp}{\Hilbertstyle^{\PDLREwLAp}_{\le_{\mathrm{fin\shortminus lin}}^{\mathrm{st}}}}
\knowledgenewrobustcmd{\vdashstPDLREwLAp}{\mathrel{\cmdkl{\vdash}}_{\cmdkl{\HilbertstylestPDLREwLAp}}}

\knowledgenewrobustcmd{\termclassPDLREwLA}[1]{\cmdkl{\termclass}_{\cmdkl{\PDLREwLA}}^{#1}}
\knowledgenewrobustcmd{\fmlclassPDLREwLA}[1]{\cmdkl{\fmlclass}_{\cmdkl{\PDLREwLA}}^{#1}}

\knowledgenewrobustcmd{\tonorm}[1]{#1^{\cmdkl{\varheartsuit}}}
\knowledgenewrobustcmd{\tonormone}[1]{#1^{\cmdkl{\varheartsuit_1}}}
\knowledgenewrobustcmd{\tonormtwo}[1]{#1^{\cmdkl{\varheartsuit_2}}}

\knowledgenewrobustcmd{\PDLtonorm}[1]{#1^{\cmdkl{\heartsuit}}}
\knowledgenewrobustcmd{\PDLtonormone}[1]{#1^{\cmdkl{\heartsuit_1}}}
\knowledgenewrobustcmd{\PDLtonormtwo}[1]{#1^{\cmdkl{\heartsuit_2}}}

\knowledgenewrobustcmd{\sigifreePDL}{\cmdkl{\sig}_{\ifreePDL}}
\knowledgenewrobustcmd{\termclassifreePDL}[1]{\cmdkl{\termclass}_{\cmdkl{\ifreePDL}}^{#1}}
\knowledgenewrobustcmd{\fmlclassifreePDL}[1]{\cmdkl{\fmlclass}_{\cmdkl{\ifreePDL}}^{#1}}
\knowledgenewrobustcmd{\exprclassifreePDL}[1]{\cmdkl{\exprclass}_{\cmdkl{\ifreePDL}}^{#1}}

\knowledgenewrobustcmd{\semifreePDL}[2]{\cmdkl{\llbracket}#1\cmdkl{\rrbracket}^{#2}}

\knowledgenewrobustcmd{\HilbertstyleifreePDLsfinlin}{\Hilbertstyle^{\ifreePDL}_{<_{\mathrm{fin\shortminus lin}}}}
\knowledgenewrobustcmd{\vdashifreePDLsfinlin}{\mathrel{\cmdkl{\vdash}}_{\cmdkl{\HilbertstyleifreePDLsfinlin}}}
\knowledgenewrobustcmd{\nvdashifreePDLsfinlin}{\mathrel{\cmdkl{\nvdash}}_{\cmdkl{\HilbertstyleifreePDLsfinlin}}}

\knowledgenewrobustcmd{\HilbertstylestifreePDL}{\Hilbertstyle^{\ifreePDL}_{<_{\mathrm{fin\shortminus lin}}^{\mathrm{st}}}}
\knowledgenewrobustcmd{\vdashstifreePDL}{\mathrel{\cmdkl{\vdash}}_{\cmdkl{\HilbertstylestifreePDL}}}
\knowledgenewrobustcmd{\nvdashstifreePDL}{\mathrel{\cmdkl{\nvdash}}_{\cmdkl{\HilbertstylestifreePDL}}}

\knowledgenewrobustcmd{\HilbertstyleifreePDL}{\Hilbertstyle^{\ifreePDL}}
\knowledgenewrobustcmd{\vdashifreePDL}{\mathrel{\cmdkl{\vdash}}_{\cmdkl{\HilbertstyleifreePDL}}}
\knowledgenewrobustcmd{\nvdashifreePDL}{\mathrel{\cmdkl{\nvdash}}_{\cmdkl{\HilbertstyleifreePDL}}}

\knowledgenewrobustcmd{\cl}{\cmdkl{\operatorname{cl}}}
\knowledgenewrobustcmd{\tcl}{\cmdkl{\operatorname{tcl}}}
\knowledgenewrobustcmd{\clone}{\cmdkl{\operatorname{cl}'}}
\knowledgenewrobustcmd{\cltwo}{\cmdkl{\tilde{\operatorname{cl}}}}
\knowledgenewrobustcmd{\clbo}{\cmdkl{\tilde{\operatorname{cl}}^{\square}}}
\knowledgenewrobustcmd{\at}{\cmdkl{\operatorname{at}}}
\NewDocumentCommand\atom{O{1}}{%
    \ifcase#1 undefined
    \or \alpha
    \or \beta
    \or \gamma
    \else undefined \fi
}
\NewDocumentCommand\atomset{O{1}}{%
    \ifcase#1 undefined
    \or \mathscr{A}
    \or \mathscr{B}
    \else undefined \fi
}
\newcommand{\atomtof}[1]{\widehat{#1}}
\newcommand{\atomsettof}[1]{\bigvee #1}

\knowledgenewrobustcmd{\canonicalmodel}{\cmdkl{\struc}}
\knowledgenewrobustcmd{\canonicalmodelwithoutLob}{\cmdkl{\struc}}
\knowledgenewrobustcmd{\canonicalmodelst}{\cmdkl{\struc}}

\knowledgenewrobustcmd{\GRELstsfintree}{\cmdkl{\mathsf{GREL}}_{\cmdkl{<_{\mathrm{fin\shortminus tree}}^{\mathrm{st}}}}}
\knowledgenewrobustcmd{\GRELstfintree}[1]{\cmdkl{\mathsf{GREL}}_{\cmdkl{\le_{\mathrm{fin\shortminus tree}}^{\mathrm{st}}}}^{#1}}
\knowledgenewrobustcmd{\GRELstfinbintree}[1]{\cmdkl{\mathsf{GREL}}_{\cmdkl{\le_{\mathrm{fin\shortminus bintree}}^{\mathrm{st}}}}^{#1}}

\knowledgenewrobustcmd{\GRELfinlinprime}[1]{\cmdkl{\mathsf{GREL}}_{\cmdkl{\le_{\mathrm{fin\shortminus lin}}'}}^{#1}}
\knowledgenewrobustcmd{\GRELpartialorderprime}[1]{\cmdkl{\mathsf{GREL}}_{\cmdkl{\le'}}^{#1}}

\knowledgenewrobustcmd{\tonormprime}[1]{#1^{\cmdkl{\varheartsuit'}}}

\NewDocumentCommand\automaton{O{1}}{%
    \ifcase#1
        undefined
    \or \mathscr{A}
    \or \mathscr{B}
    \else undefined
    \fi
}
\NewDocumentCommand\trace{O{1}}{%
    \ifcase#1
        undefined
    \or \tau
    \else undefined
    \fi
}
\knowledgenewrobustcmd{\series}{\mathbin{\cmdkl{\diamond}}}

\knowledgenewrobustcmd{\tonormbin}[1]{#1^{\cmdkl{\vardiamondsuit}}}
\knowledgenewrobustcmd{\tonormuni}[1]{#1^{\cmdkl{\clubsuit}}}
\knowledgenewrobustcmd{\tonormonetermvariable}[1]{#1^{\cmdkl{\spadesuit}}}

\knowledgenewrobustcmd{\clex}{\cmdkl{\tilde{\operatorname{cl}}}}
\knowledgenewrobustcmd{\clexbo}{\cmdkl{\tilde{\operatorname{cl}}^{\square}}}
\knowledgenewrobustcmd{\clexex}{\cmdkl{\tilde{\operatorname{cl}}'}}

\knowledgenewrobustcmd{\Iverson}[1]{\cmdkl{[}#1\cmdkl{]}}

\knowledgenewrobustcmd{\wlangASA}{\cmdkl{\mathcal{L}}}

\knowledgenewrobustcmd{\treelangATA}{\cmdkl{\mathcal{L}}}
\knowledgenewrobustcmd{\runlangATA}{\cmdkl{\mathcal{R}}}

\knowledgenewrobustcmd{\termendmarker}{\cmdkl{\term[3]_{\mathdollar}}}
\knowledgenewrobustcmd{\termendmarkertwo}{\cmdkl{\term[3]_{\mathdollar}}}

\knowledge{text={i.e.}, italic}
  | ie

\knowledge{text={I.e.}, italic}
  | Ie

\knowledge{text={s.t.}, italic}
  | st

\knowledge{text={e.g.}, italic}
  | eg

\knowledge{text={vs.}, italic}
  | vs

\knowledge{text={w.r.t.}, italic}
  | wrt

\knowledge{text={a.k.a.}, italic}
  | aka

\knowledge{text={w.l.o.g.}, italic}
  | wlog

\knowledge{text={W.l.o.g.}, italic}
  | Wlog

\knowledge{text={cf.}, italic}
  | cf

\knowledge{text={Cf.}, italic}
  | Cf

\knowledge{text={iff}, italic}
  | iff

\knowledge{text={r.e.}, italic}
  | r.e.
  | re

\knowledge{text={resp.}, italic}
  | resp
\knowledge{text={\ensuremath{\text{\textsc{ExpTime}}}}}
 | EXPTIME

\knowledge{text={\ensuremath{\text{\textsc{PSpace}}}}}
 | PSPACE

\knowledge{text={\ensuremath{\text{\textsc{2NExpTime}}}}}
 | 2NEXPTIME

 \knowledge{text={\ensuremath{\text{\textsc{ExpSpace}}}}}
 | EXPSPACE

\knowledge{notion}
 | empty set

\knowledge{notion}
 | power set

 \knowledge{notion}
 | cardinality

\knowledge{notion}
 | union

\knowledge{notion}
 | disjoint union

\knowledge{notion}
 | Iverson bracket

\knowledge{ignore}
| binary relations
| binary relation

\knowledge{notion}
| identity relation

\knowledge{notion}
| empty relation

\knowledge{notion}
| composition
| (relational) composition
| relational composition

\knowledge{notion}
| iteration

\knowledge{notion}
| reflexive transitive closure
| RTC

\knowledge{notion}
| transitive closure
| TC

\knowledge{ignore}
| $n$-th iteration

\knowledge{notion}
| relational model
| relational models

\knowledge{notion}
 | transitive relation

\knowledge{notion}
| linear order

\knowledge{notion}
| strict linear order

\knowledge{notion}
| finite linear orders
| finite linear order

\knowledge{notion}
| preorder
| preorders

\knowledge{notion}
| partial order

\knowledge{notion}
| strict partial order

\knowledge{notion}
| finite \emph{strict} linear orders
| finite strict linear orders
| finite strict linear order

\knowledge{notion}
| character
| characters

\knowledge{notion}
| empty string

\knowledge{notion}
| string
| strings

\knowledge{notion}
| concatenation

\knowledge{notion}
| length

\knowledge{notion}
| tree language
| tree languages
| language@tree
| languages@tree

\knowledge{notion}
| string language
| string languages
| language@string
| languages@string

\knowledge{notion}
| language equivalence
| language equivalence@string
| equivalence problem

\knowledge{notion}
| substitution-closed language equivalence
| substitution-closed language equivalence@string
| substitution-closed equivalence problem

\knowledge{notion}
| emptiness problem
| language emptiness

\knowledge{notion}
| NFA
| NFAs

\knowledge{notion}
| alternating tree automata
| alternating tree automaton
| ATA
| ATAs

\knowledge{notion}
| states

\knowledge{notion}
| positive boolean formula

\knowledge{notion}
| accepting

\knowledge{notion}
| tree

\knowledge{notion}
| run@ATA

\knowledge{notion}
| language@ATA

\knowledge{notion}
| alternating string automaton
| alternating string automata
| ASA
| ASAs

\knowledge{notion}
| run@ASA

\knowledge{notion}
| language@ASA

\knowledge{notion}
| equation
| equations

\knowledge{notion}
| inequation
| inequations

\knowledge{notion}
| $\sig$-algebra
| $\sig$-algebras

\knowledge{notion}
| equational theory
| equational theories

\knowledge{notion}
 | substitutions
 | substitution

\knowledge{notion}
 | substitution-closed

\knowledge{notion}
 | fresh

\knowledge{notion}
 | formula variables
 | formula variable

\knowledge{notion}
| formula@general
| formulas@general
| (quantifier-free) formulas@general

\knowledge{notion}
 | subformula

\knowledge{notion}
| term variable
| term variables

\knowledge{notion}
| term@general
| terms@general

\knowledge{notion}
| subterm

\knowledge{notion}
| expression@general
| expressions@general

\knowledge{notion}
 | semantics@general
 | relational semantics@general

\knowledge{notion}
| generalized relational structures
| generalized (relational) structures
| generalized structures
| generalized relational structure
| generalized (relational) structure
| generalized structure

\knowledge{notion}
| relational structures
| (relational) structures
| structures
| relational structure
| (relational) structure
| structure

\knowledge{notion}
| string structure
| string structures

\knowledge{notion}
 | frame
 | frames

\knowledge{notion}
 | universe

\knowledge{notion}
 | universal relation

\knowledge{notion}
| valuation
| valuations

\knowledge{notion}
 | satisfiable

\knowledge{notion}
 | valid

\knowledge{notion}
 | theory
 | theories

\knowledge{notion}
| validity problem

\knowledge{notion}
| finite validity problem

\knowledge{notion}
| expressive power

\knowledge{notion}
 | substitution-instances
 | substitution-instance

\knowledge{notion}
 | axioms
 | axiom

\knowledge{notion}
 | rule
 | rules

\knowledge{notion}
 | admissible

\knowledge{notion}
 | derivable

\knowledge{notion}
 | derivation path

\knowledge{notion}
 | consistent
 | consistency
 | inconsistent

\knowledge{notion}
 | atom
 | atoms

\knowledge{notion}
 | canonical model

\knowledge{notion}
 | regexes
 | regex

\knowledge{notion}
| Kleene star

\knowledge{notion}
 | Kleene plus

\knowledge{notion}
 | Lookahead
 | lookahead

\knowledge{notion}
| negative lookahead

\knowledge{notion}
| positive lookahead

\knowledge{notion}
| antidomain

\knowledge{notion}
| domain

\knowledge{notion}
 | lookbehind

\knowledge{notion}
 | leftmost anchor

 \knowledge{notion}
 | rightmost anchor

\knowledge{notion}
 | backreferences

\knowledge{notion}
 | match-language
 | match-language@string

\knowledge{notion}
 | match-language equivalence
 | match-language equivalence@string

 \knowledge{notion}
 | substitution-closed match-language equivalence
 | substitution-closed match-language equivalence@string

 \knowledge{notion}
 | substitution-closed equivalences
 | substitution-closed equivalence
 | Substitution-closed equivalences

\knowledge{notion}
 | match-languages with substitutions
 | match-language with substitutions

\knowledge{notion}
 | languages with substitutions
 | language with substitutions

\newcommand{\RE}{\ensuremath{\mathrm{RE}}}
\knowledge{notion}
| regular expressions
| regular expression
| \RE
| Regular expressions

\knowledge{notion}
| term@RE
| terms@RE

\newcommand{\REwLA}{\ensuremath{\mathrm{REwLA}}}
\knowledge{notion}
| regular expressions with lookahead
| \REwLA
| Regular expressions with lookahead

\knowledge{notion}
| term@REwLA
| terms@REwLA

\knowledge{notion}
 | semantics@REwLA

\newcommand{\REwLAp}{\ensuremath{\mathrm{REwLA+}}}
\knowledge{notion}
| extended regular expressions with lookahead
| \REwLAp
| Extended regular expressions with lookahead

\knowledge{notion}
| term@REwLAp
| terms@REwLAp

\knowledge{notion}
 | semantics@REwLAp
\knowledge{notion}
 | box operator
 | box

\knowledge{notion}
 | diamond 
 | diamond operator

\knowledge{notion}
 | rich test
 | test
 | rich tests

\knowledge{notion}
 | tree unwinding

\knowledge{notion}
 | bisimulation

\knowledge{notion}
 | bisimilar

 \knowledge{notion}
 | generated substructure

\knowledge{notion}
 | Kleene algebra
 | KA

\newcommand{\PDL}{\ensuremath{\mathrm{PDL}}}
\knowledge{notion}
| propositional dynamic logic
| Propositional Dynamic Logic
| \PDL

\knowledge{notion}
| formula@PDL
| formulas@PDL

\knowledge{notion}
| term@PDL
| terms@PDL

\knowledge{notion}
| expressions@PDL
| expression@PDL

\knowledge{notion}
| semantics@PDL

\newcommand{\DPDL}{\ensuremath{\mathrm{DPDL}}}
\knowledge{notion}
| DPDL
| deterministic PDL
| deterministic propositional dynamic logic
| \DPDL

\newcommand{\HKTPDL}{\ensuremath{\mathrm{PDL}'}}
\knowledge{notion}
| propositional dynamic logic@HKTPDL
| \HKTPDL

\knowledge{notion}
| formula@HKTPDL
| formulas@HKTPDL

\knowledge{notion}
| term@HKTPDL
| terms@HKTPDL

\knowledge{notion}
| expression@HKTPDL
| expressions@HKTPDL

\knowledge{notion}
| semantics@HKTPDL

\newcommand{\PDLREwLAp}{\ensuremath{\mathrm{PDL}_{\mathrm{{REwLA+}}}}}
\knowledge{notion}
| \PDLREwLAp
| PDL of regular expressions with lookahead and restrictions to the identity and the complement of the identity

\knowledge{notion}
| term@PDLREwLAp
| terms@PDLREwLAp

\knowledge{notion}
| formula@PDLREwLAp
| formulas@PDLREwLAp

\knowledge{notion}
| expression@PDLREwLAp
| expressions@PDLREwLAp

\knowledge{notion}
| semantics@PDLREwLAp

\newcommand{\PDLREwLA}{\ensuremath{\mathrm{PDL}_{\mathrm{{REwLA}}}}}
\knowledge{notion}
| \PDLREwLA
| PDL of regular expressions with lookahead

\knowledge{notion}
| term@PDLREwLA
| terms@PDLREwLA

\knowledge{notion}
| formula@PDLREwLA
| formulas@PDLREwLA

\knowledge{notion}
| expression@PDLREwLA
| expressions@PDLREwLA

\knowledge{notion}
| semantics@PDLREwLA

\newcommand{\ifreePDL}{\ensuremath{\mathrm{PDL}^{-}}}
\knowledge{notion}
| \ifreePDL
| identity-free \PDL
 | identity-free PDL

\knowledge{notion}
| formula@PDL-
| formulas@PDL-
| identity-free \PDL formulas
| identity-free \PDL formula
| Identity-free \PDL formulas

\knowledge{notion}
| term@PDL-
| terms@PDL-

\knowledge{notion}
| expression@PDL-
| expressions@PDL-

\knowledge{notion}
| semantics@PDL-

\knowledge{notion}
| L{\"o}b's axiom

\knowledge{notion}
| L{\"o}b's rule

\knowledge{notion}
 | homomorphism

\knowledge{notion}
 | isomorphic
 | isomorphism

\knowledge{notion}
 | edges

\knowledge{notion}
 | vertices

\knowledge{notion}
 | consistent@PDLwithoutLob

\knowledge{notion}
 | atom@PDLwithoutLob

\knowledge{notion}
 | canonical model@PDLwithoutLob

\knowledge{notion}
 | consistent@st

\knowledge{notion}
 | atom@st

\knowledge{notion}
 | canonical model@st

 \knowledge{notion}
 | weakly decomposed normal form

 \knowledge{notion}
 | decomposed normal form

 \knowledge{notion}
 | FL-closure@PDL-

\knowledge{notion}
 | FL-closure@PDLREwLAp

\knowledge{notion}
 | closure

\begin{document}
\VerbatimFootnotes
\title{A Complete Propositional Dynamic Logic for Regular Expressions with Lookahead
}
\titlerunning{A Complete PDL for REwLA}
\author{
Yoshiki Nakamura\inst{1,2}\orcidlink{0000-0003-4106-0408}
}
\authorrunning{
Y. Nakamura
}
\institute{
Chiba University, Japan \and
Institute of Science Tokyo, Japan\\
\email{nakamura.yoshiki.ny@gmail.com}
}
\maketitle              %
\begin{abstract}
We consider (logical) reasoning for \emph{regular expressions with lookahead} (REwLA).
In this paper,
we give an axiomatic characterization for both the (match-)language equivalence and 
the largest substitution-closed equivalence that is sound for the (match-)language equivalence.
To achieve this, we introduce a variant of propositional dynamic logic (PDL) on finite linear orders,
extended with two operators:
the restriction to the identity relation and the restriction to its complement.
Our main contribution is a sound and complete Hilbert-style finite axiomatization for the logic, which captures the equivalences of REwLA.
Using the extended operators, the completeness is established via a reduction into an \emph{identity-free variant of PDL} on finite strict linear orders.
Moreover, the extended PDL has the same computational complexity as REwLA.

\end{abstract}

\section{Introduction}\label{section: introduction}
\AP
While classical ""regular expressions"" (\reintro*\kl{\RE}) \cite{kleeneRepresentationEventsNerve1951} are built from 
constants and the operators:
"concatenation" ($\compo$), "union" ($\union$), and "Kleene star" ($\bl^{*}$),
various extensions are implemented in real-world ""regexes"" (see, "eg", \cite{friedlMasteringRegularExpressions2006, thepcre2developersPerlcompatibleRegularExpressions}).
To optimize "regexes" ("eg", "wrt", its length or the complexity of matching algorithms),
we are interested in transforming a given "expression@regular expression" into an equivalent "one@regular expression".
In classical \kl{\RE}, there is a finite (and quasi-equational) algebraic axiomatization \cite{kozenCompletenessTheoremKleene1991}, known as "Kleene algebra", which is sound and complete for "language equivalence".
Another framework is \emph{"propositional dynamic logic"} (\emph{"\PDL"}) of \kl{\RE} with rich tests \cite{fischerPropositionalDynamicLogic1979,harelDynamicLogic2000}.
"\PDL" also enjoys a sound and complete axiomatization \cite{gabbayAxiomatizationsLogicsPrograms1977,parikhCompletenessPropositionalDynamic1978} and embeds the "language equivalence".
A natural (naive) question is whether such sound and complete axiomatic systems can be extended with the operators employed in "regexes".

As a first step, we consider \emph{"regular expressions with lookahead"} ("\REwLA") \cite{morihataTranslationRegularExpression2012,miyazakiDerivativesRegularExpressions2019}.
""Lookahead"" allows us to assert that a certain pattern is satisfied in the future of the current position.
For instance,
the expression \verb#((?!ab)(a|b))*# expresses the set $\set{\mathtt{b}^n \mathtt{a}^{m} \mid n, m \ge 0}$,
where the \emph{"negative lookahead"} \verb#(?!ab)# asserts that the next two "characters" are not \verb#ab# and the symbol ``\verb#|#'' expresses the "union" ($\union$).
Unlike \kl{\RE},
the "language equivalence" of "\REwLA" is \emph{not} closed under "substitutions",
so there is no sound and complete set of ("substitution-closed") axiom schemas.
For instance, although \verb#((?!ab)(a|b))*# and \verb#b*a*# have the same "language@@string" as above,
substituting \verb#b# with \verb#a# yields \verb#((?!aa)(a|a))*# and \verb#a*a*#,
which define the sets $\set{\eps, \mathtt{a}}$ and $\set{\mathtt{a}^{n} \mid n \ge 0}$, breaking the "language equivalence".\footnote{Another instance is \Verb#(?=a)b# $=$ $\emptyset$, where the \emph{"positive lookahead"} \Verb#?=# asserts that the next "character" is \Verb#a# and the symbol $\emptyset$ expresses the empty "language@string",
"cf", \Verb#(?=a)a# $\neq$ $\emptyset$.
}
For that reason, we also study the \emph{largest "substitution-closed" equivalence} that is sound for the "language equivalence", which has a sound and complete set of axiom schemas (\Cref{theorem: PDL REwLA completeness}) via a slight encoding (\Cref{proposition: largest substitution-closed language equivalence}).
Such equivalences are also useful from the perspective of reusability.
This equivalence is alternatively characterized
by the "relational semantics@@general" on "finite linear orders",
based on the semantics of slices of "strings" \cite{mamourasEfficientMatchingRegular2024},
where the "valuations" are filled (\Cref{proposition: largest substitution-closed}).

\subsubsection*{Contributions}
The main contribution of this paper is to present
an axiomatic characterization for "\REwLA" "wrt" both
(i) the "substitution-closed" equivalence (\Cref{theorem: PDL REwLA completeness}), and
(ii) the standard "language equivalence" (\Cref{theorem: completeness match-language equivalence}).
While several algebraic equational properties have been investigated ("eg", \cite[Definition 2.2 ``Kleene algebra with lookahead'']{miyazakiDerivativesRegularExpressions2019}\cite[Lemma 11]{mamourasEfficientMatchingRegular2024}), no complete axiomatizations have yet been presented, to our knowledge ("cf", "eg", \cite[p.~92:10]{mamourasEfficientMatchingRegular2024}).

In this paper, we present a finite axiomatization for an extended "\PDL" on "finite linear orders", which embeds the equivalences above for "\REwLA" (\Cref{section: PDLREwLAp}).
More precisely,
we introduce "\PDLREwLAp":
"\PDL" with
the restriction to the identity relation ($\bl^{\capid}$) and
the restriction to the complement of the identity relation ($\bl^{\capcomid}$).
Using these operators,
we can decompose the "relational semantics@@general" into the identity-part and the identity-free-part (\Cref{section: reduction to identity-free}).
We then can give a reduction
from the completeness theorem of "\PDLREwLAp" on "finite linear orders"
to the completeness theorem of an \emph{identity-free} variant of "\PDL" (denoted by "\ifreePDL") on "finite \emph{strict} linear orders@finite strict linear orders" (\Cref{section: PDL- completeness}).
While axiomatizations for fragments/variants of "\PDL" on finite trees \cite{krachtSyntacticCodesGrammar1995,krachtInessentialFeatures1997,afanasievPDLOrderedTrees2005} (or equivalently, on "finite strict linear orders" via "bisimulation") were presented \cite{blackburnLinguisticsLogicFinite1994,blackburnProofSystemFinite1996,palmPropositionalTenseLogic1999}, no axiomatization for full "\PDL" have yet been presented, to our knowledge.
A key of their axiomatizations is to employ \emph{"L{\"o}b's axiom"}.
Also for "\ifreePDL", we can employ "L{\"o}b's axiom", thanks to the absence of the identity-part.
Our approach eliminating the identity-part is inspired by Brunet's reduction \cite{brunetCompleteAxiomatisationFragment2020},
which shows the completeness of 
the equational theory of (reversible) Kleene lattices interpreted as algebras of languages from that of identity-free Kleene lattices \cite{doumaneCompletenessIdentityfreeKleene2018}.
Here, in our reduction, the operator $\bl^{\capcomid}$ is introduced for employing "L{\"o}b's axiom" forcibly (\Cref{figure: PDLREwLAp axioms}).

Moreover, the extension above does not increase the complexity,
in that
the "theory" of "\PDLREwLAp" on "finite linear orders" ($\GRELfinlin{}$) and
the embedded "substitution-closed equivalence" problems of "\REwLA"
are "EXPTIME"-complete (\Cref{theorem: complexity PDLREwLAp substitution-closed}),
and
the "theory" of "\PDLREwLAp" on a subclass of "finite linear orders" ($\GRELstfinlin{}$) and
the embedded standard "(match-)language equivalence@match-language equivalence" problems of "\REwLA" are "PSPACE"-complete (\Cref{theorem: complexity PDLREwLAp standard}),
respectively (\Cref{section: complexity}).

\subsubsection*{Organization}
In \Cref{section: preliminaries}, we give basic definitions of "\REwLA" and "\PDL".
In \Cref{section: PDLREwLAp}, we introduce "\PDLREwLAp" and its Hilbert-style axiomatization $\HilbertstylePDLREwLAp$ on "finite linear orders".
In \Cref{section: PDL-,section: reduction to identity-free,section: PDL- completeness}, we prove the completeness theorem of $\HilbertstylePDLREwLAp$.
After introducing "\ifreePDL" and its Hilbert-style axiomatization $\HilbertstyleifreePDLsfinlin$ on "finite strict linear orders" in \Cref{section: PDL-},
we provide the reduction from the completeness theorem of $\HilbertstyleifreePDLsfinlin$ to that of $\HilbertstylePDLREwLAp$ in \Cref{section: reduction to identity-free}.
In \Cref{section: PDL- completeness}, we prove the completeness theorem of $\HilbertstyleifreePDLsfinlin$.
In \Cref{section: completeness match-language equivalence}, we also give an axiomatic characterization for the standard ("match@match-language equivalence"-)"language equivalence".
In \Cref{section: complexity}, we consider the computational complexity.
In \Cref{section: conclusion}, we conclude this paper with future work.

\nointro\pfun
\nointro\fdom
\nointro{power set}
\nointro{cardinality}

\nointro{preorder}
\nointro{partial order}
\nointro{strict partial order}
\nointro{linear order}
\nointro{strict linear order}
\nointro{finite linear order}
\nointro{finite strict linear order}

\nointro{vertices}
\nointro{edges}
\nointro{tree}
\nointro{isomorphic}

\nointro{fresh}
\nointro{substitution}
\nointro{axioms}
\nointro{rule}
\nointro{derivable}
\nointro{admissible}
\nointro{substitution-instances}

\nointro{Kleene algebra}

\nointro{closure}
\section{Preliminaries}\label{section: preliminaries} 
\begin{scope}\knowledgeimport{general}
\AP
We write $\intro*\nat$ for the set of non-negative integers.
For $l, r \in \nat$,
we write $\intro*\range{l}{r}$ for the set $\set{i \in \nat \mid l \le i \le r}$.
For a set $X$, we write $\intro*\pset(X)$ for the "power set" of $X$
and write $\intro*\card(X)$ for the "cardinality" of $X$.
We often use $\intro*\dcup$ to denote that the set "union" $\cup$ is disjoint.

\AP
For a set $X$, we write $X^*$ for the set of all "strings" over $X$.
We write $\intro*\eps$ for the "empty string".
For a "string" $\word = a_1 \dots a_n$, we write $\intro*\len{\word}$ for the ""length"" $n$ of $\word$.

\AP
Given two disjoint sets $\vsig$ (for "term variables") and $\psig$ (for "formula variables"),
we use $\aterm[1], \aterm[2], \dotsc \in \vsig$ to denote ""term variables"" and use $\afml[1], \afml[2], \dotsc \in \psig$ to denote ""formula variables"".
We will use $\term[1],\term[2],\term[3],\dotsc$ to denote ""terms"",
use $\fml[1],\fml[2],\fml[3],\dotsc$ to denote ""formulas"", and
use $\expr[1],\expr[2],\expr[3],\dotsc$ to denote ""expressions"", "ie", "terms" or "formulas".
\AP
An ""equation"" $\term[1] = \term[2]$ is a pair of "terms".
We denote by $\term[1] \le \term[2]$ the "equation" $\term[1] \union \term[2] = \term[2]$.
For an "expression" $\expr$, we write $\intro*\exprvsig(\expr)$ for the set of "term variables" occurring in $\expr$ and $\intro*\exprpsig(\expr)$ for the set of "formula variables" occurring in $\expr$, respectively.

\AP
A ""frame"" $\kframe$ is a tuple $\tuple{\univ{\kframe}, \strucuniv^{\kframe}}$,
where
its ""universe"" $\intro*\univ{\kframe}$ is a non-empty set and
its ""universal relation"" $\intro*\strucuniv^{\kframe} \subseteq \univ{\kframe}^2$ is a binary relation.
Given two disjoint sets $\vsig$ and $\psig$,
a ""generalized structure"" $\struc$ on a "frame" $\kframe$ is a tuple $\tuple{\univ{\struc}, \strucuniv^{\struc}, \set{\aterm^{\struc}}_{\aterm \in \vsig}, \set{\afml^{\struc}}_{\afml \in \psig}}$,
where
$\tuple{\univ{\struc}, \strucuniv^{\struc}} = \kframe$,
$\aterm^{\struc} \subseteq \strucuniv^{\struc}$ is a binary relation for each $\aterm \in \vsig$, and
$\afml^{\struc} \subseteq \univ{\struc}$ is a unary relation for each $\afml \in \psig$.
We say that $\struc$ is a ""structure"" if $\strucuniv^{\struc} = \univ{\struc}^2$.
We write
$\intro*\GREL{}$ ("resp", $\intro*\REL{}$)
for the class of all "generalized structures" ("resp", "structures").
\AP
We also write
$\intro*\GRELpreorder{}$ ("resp", $\intro*\GRELfinlin{}$, $\intro*\GRELsfinlin{}$) for
all $\struc \in \GREL{}$ "st" $\strucuniv^{\struc}$ is a "preorder" ("resp", "finite linear order", "finite strict linear order").

\AP
Given an $\struc \in \GREL{}$,
the ""semantics"" $\intro*\sem{\bl}{\struc} \colon \termclassgen{\vsig,\psig} \to \pset(\strucuniv^{\struc}) \dcup \fmlclassgen{\vsig,\psig} \to \pset(\univ{\struc})$ is partially%
\footnote{Note that $\sem{\bl}{\struc}$ may be ill-defined, when $\sem{\term}{\struc} \not\subseteq \strucuniv^{\struc}$ for some $\term$.} defined as the unique ""homomorphism"" (as the two sort algebra) extending the ""valuation"" $(\lambda \aterm. \aterm^{\struc}) \dcup (\lambda \afml. \afml^{\struc}) \colon 
(\vsig \to \pset(\strucuniv^{\struc})) \dcup (\psig \to \pset(\univ{\struc}))$.
Here, $\termclassgen{\vsig,\psig}$ ("resp", $\fmlclassgen{\vsig,\psig}$) denotes the class of all "terms" ("resp", "formulas").

\AP
Let $\strucclass \subseteq \GREL{}$ be a class "st" $\sem{\bl}{\struc}$ is well-defined for all $\struc \in \strucclass$.
We say that an "equation" $\term[1] = \term[2]$ ("resp", "formula" $\fml$) is ""valid"" on $\strucclass$ if $\sem{\term[1]}{\struc} = \sem{\term[2]}{\struc}$ ("resp", $\sem{\fml}{\struc} = \univ{\struc}$) for all $\struc \in \strucclass$; we denote them by $\strucclass \intro*\modelsfml \term[1] = \term[2]$ ("resp", $\strucclass \reintro*\modelsfml \fml$).
The ""equational theory"" ("resp", ""theory"") on $\strucclass$ is the class of all 
"equations" ("resp", "formulas") "valid" on $\strucclass$.
\end{scope}

\begin {scope}\knowledgeimport {REwLA}
\subsection{REwLA: Regular Expressions with LookAhead}\label{subsection: REwLA}
\AP
\phantomintro*\kl{term}%
\intro*\kl{Regular expressions with lookahead} (\reintro*\kl{\REwLA}) \cite{morihataTranslationRegularExpression2012,miyazakiDerivativesRegularExpressions2019} are generated by the following grammar:%
\footnote{%
For later convenience, we use "Kleene plus" ($\bl^{+}$) instead of "Kleene star" ($\bl^{*}$) as a primitive operator.
The notations for "negative lookahead" ($\adom$) and "positive lookahead" ($\dom$)
are based on "antidomain" and "domain" (in the context of Kleene algebra with (anti)domain \cite{desharnaisKleeneAlgebraDomain2006,desharnaisModalSemiringsRevisited2008,desharnaisInternalAxiomsDomain2011}), where we use the superscript notation for short.}
\phantomintro\termclassREwLA%
\phantomintro*\kl{character}%
\phantomintro*\kl{empty string}%
\phantomintro*\kl{empty set}%
\phantomintro*\kl{concatenation}%
\phantomintro*\kl{union}%
\phantomintro*\kl{Kleene plus}%
\phantomintro*\kl{negative lookahead}%
\begingroup
\allowdisplaybreaks
\begin{align*}
    \term[1], \term[2], \term[3] \in \reintro*\termclassREwLA{\vsig}
    &\;\Coloneqq\; &&&&\hspace{.5em}\phantom{ \mid } \aterm && \text{[\reintro*\kl{character}, "term variable" $\aterm \in \vsig$]} \\
    &\hspace{-5em}\mid \id &&\hspace{-5em} \text{[\reintro*\kl{empty string}]} &
    &\mid \emp && \text{[\reintro*\kl{empty set}]}\\
    &\hspace{-5em}\mid \term[2] \compo \term[3] &&\hspace{-5em} \text{[\reintro*\kl{concatenation}]}&
    &\mid \term[2] \union \term[3] && \text{[\reintro*\kl{union}]}\\
    &\hspace{-5em}\mid \term[2]^+  &&\hspace{-5em}\text{[\reintro*\kl{Kleene plus}]} &
    &\mid \term[2]^{\adom} && \text{[\reintro*\kl{negative lookahead} ``$\texttt{?!} \term[2]$'']}
\end{align*}
\endgroup
We usually abbreviate $\term[1] \compo \term[2]$ to $\term[1] \term[2]$.
We use parentheses in ambiguous situations. %
We write $\sum_{i = 1}^{n} \term[1]_i$ for the term $\emp \union \term[1]_1 \union \dots \union \term[1]_n$,
and write $\bigcompo_{i = 1}^{n} \term[1]_i$ for the term $\id \compo \term[1]_1 \compo \dots \compo \term[1]_n$.
We use the following abbreviations:
\AP
\phantomintro*\kl{Kleene star}%
\phantomintro*\kl{iteration}%
\phantomintro*\kl{positive lookahead}%
\begin{align*}
    \term^{*} &\defeq \id \union \term^{+}  && \text{[\reintro*\kl{Kleene star}]} &
    \term^{n} &\defeq \bigcompo_{i = 1}^{n} \term[1] && \text{[$n$-th \reintro*\kl{iteration} ($n \in \nat$)]}\\
    \term^{\dom} &\defeq (\term^{\adom})^{\adom} && \text{[\reintro*\kl{positive lookahead} ``$\texttt{?{=}} \term$'']}
\end{align*}

\subsection{Relational Semantics, Match-Languages, and Languages}
\AP We recall the algebraic semantics from \emph{slices of strings} \cite{mamourasEfficientMatchingRegular2024}.%
\footnote{"Cf" \emph{matching relation} \cite{chidaLookaheadsRegularExpressions2023}, where assignments for backward references are forgotten.}
Below, we slightly reformulate in terms of algebras of binary relations.
Given a set $X$, we consider the following operators on binary relations on $X$:
\AP%
\phantomintro*\kl{identity relation}%
\phantomintro*\kl{empty relation}%
\phantomintro*\kl{relational composition}%
\phantomintro*\kl{transitive closure}%
\phantomintro*\kl{antidomain}%
\phantomintro*\diagonal%
\begin{align*}
    \id &\;\defeq\; \reintro*\diagonal_{X} \;\defeq\; \set{\tuple{\vertex[1], \vertex[1]} \mid \vertex[1] \in X} \tag*{[\reintro*\kl{identity relation}]}\\
    \emp &\;\defeq\; \emptyset \tag*{[\reintro*\kl{empty relation}]}\\
    \rel \compo \rel[2] &\;\defeq\; \set{\tuple{\vertex[1], \vertex[3]} \mid \exists \vertex[2] \in X, \tuple{\vertex[1], \vertex[2]} \in \rel \text{ and } \tuple{\vertex[2], \vertex[3]} \in \rel[2]} \tag*{[\reintro*\kl{relational composition}]}\\
    \rel \union \rel[2] &\;\defeq\; \rel \cup \rel[2] \tag*{[\reintro*\kl{union}]}\\
    \rel^{+} &\;\defeq \set{\tuple{\vertex[1]_0, \vertex[1]_n} \mid \exists n \ge 1, \exists \vertex[1]_1, \dots, \exists \vertex_{n-1}, \forall i < n, \tuple{\vertex_{i}, \vertex_{i+1}} \in \rel} \tag*{[\reintro*\kl{TC}]}\\
    \rel^{\adom} &\;\defeq\; \set{\tuple{\vertex[1], \vertex[1]} \in \diagonal_{X} \mid \forall \vertex[2] \in X, \tuple{\vertex[1], \vertex[2]} \not\in \rel} \tag*{[\reintro*\kl{antidomain}]}
\end{align*}
\AP
Additionally, we define the following operators:%
\phantomintro*\kl{reflexive transitive closure}%
\phantomintro*\kl{domain}%
\begin{align*}
    \rel^{*} &\defeq \id \union \rel^{+}  && \text{[\reintro*\kl{RTC}]} &
    \rel^{n} &\defeq \bigcompo_{i = 1}^{n} \rel[1] && \text{[$n$-th \reintro*\kl{iteration}]} &
    \rel^{\dom} &\defeq (\rel^{\adom})^{\adom} && \text{[\reintro*\kl{domain}]}
\end{align*}
\AP
Given an $\struc \in \GRELpreorder{}$ (on "preorder"),
the ""semantics"" $\intro*\semREwLA{\cdot}{\struc}$ of "\REwLA"
is well-defined, where each operator on $\pset(\strucuniv^{\struc})$ \footnote{%
Each "\REwLA" operator on $\pset(\strucuniv^{\struc})$ is well-defined, because $\strucuniv^{\struc}$ is a "preorder".
Note that $\id$ requires reflexivity of $\strucuniv^{\struc}$ and $\compo$ requires transitivity of $\strucuniv^{\struc}$.} is interpreted as above.

\AP
The ""string structure"" $\intro*\wordstruc^{\word}$ of a "string" $\word = \aterm[1]_1 \dots \aterm[1]_n \in \vsig^*$
is the "generalized structure" defined by
$\strucuniv^{\wordstruc^{\word}} \defeq \set{\tuple{i, j} \in \range{0}{n}^2 \mid i \le j}$,
$\aterm[2]^{\wordstruc^{\word}} \defeq \set{\tuple{i, i+1} \mid i \in \range{0}{n - 1} \text{ and } \aterm[1]_{i+1} = \aterm[2]}$ for $\aterm[2] \in \vsig$, and $\afml^{\wordstruc^{\word}} = \emptyset$ for $\afml \in \psig$.

\AP
\phantomintro{match-language equivalence}%
\phantomintro{language equivalence}%
The ""match-language"" $\intro*\mlang(\term)$ and the ""language@string"" %
$\intro*\wlang(\term)$ are defined as follows \cite{mamourasEfficientMatchingRegular2024}:
\begin{align*}
\reintro*\mlang(\term) &\defeq \set{\tuple{\word, i, j} \mid \word \in \vsig^*, \tuple{i, j} \in \sem{\term}{\wordstruc^{\word}}}, &\hspace{-.7em}
\reintro*\wlang(\term) &\defeq \set{\word \in \vsig^* \mid \tuple{0, \len{\word}} \in \sem{\term}{\wordstruc^{\word}}}.
\end{align*}

Let $\intro*\GRELstfinlin{}$ denote the class of all $\struc[2] \in \GRELfinlin{}$ (with "finite linear order") that are "isomorphic" to $\wordstruc^{\word}$ for some $\word$ after replacing $\afml^{\struc[2]}$ with $\emptyset$ for each $\afml \in \psig$.
Since "\REwLA" "terms@@REwLA" do not contain any "formula variables", we have:
\[\mlang(\term[1]) = \mlang(\term[2])
~\iff~ \set{\wordstruc^{\word} \mid \word \in \vsig^*} \modelsfml \term[1] = \term[2]
~\iff~ \GRELstfinlin{} \modelsfml \term[1] = \term[2].\]
We can embed the "language equivalence" into the "match-language equivalence" via a slight encoding,
where the "term" $\termendmarker$ is defined to express the end of the "string".
\AP\phantomintro\termendmarker%
\ifthenelse{\boolean{conference}}{\begin{proposition}}{%
\begin{proposition}[\Cref{section: proposition: language equivalence}]}\label{proposition: language equivalence}%
\gdef\propositionlanguageequivalence{%
$\wlang(\term[1]) = \wlang(\term[2])$ "iff" $\GRELstfinlin{} \modelsfml \term[1] \termendmarker = \term[2] \termendmarker$,
where $\reintro*\termendmarker \defeq \left(\sum_{\aterm \in \exprvsig(\term[1]\, \term[2])}\aterm \right)^{\adom}$.
}
\propositionlanguageequivalence
\end{proposition}
\noindent
We thus consider $\GRELstfinlin{}$ for the standard "(match)-language equivalence@match-language equivalence".

\subsection{Substitution-Closed Equivalences}\label{section: substitution-closed equivalnce}
\AP
For a "term" $\term$ and a "substitution" $\Theta$ mapping each "term variable" to a "term",
we write $\term\intro*\termsubst{\Theta}$ for the "term" obtained by substituting each "term variable" $x$ with $\Theta(x)$.

\AP \phantomintro*\kl{substitution-closed equivalence}%
A "binary relation" $R$ on \kl{terms} is \intro*\kl{substitution-closed} if,
for all \kl{terms} $\term[1], \term[2]$ and all "substitutions" $\Theta$,
if $\tuple{\term[1], \term[2]} \in R$, then $\tuple{\term[1]\termsubst{\Theta}, \term[2]\termsubst{\Theta}} \in R$.
In this paper, we consider the \emph{largest} "substitution-closed" equivalence relation contained in "match-language equivalence" ("resp", "language equivalence"); henceforth, just \emph{the "substitution-closed" equivalence sound for "(match-)language equivalence@match-language equivalence"} or the ""substitution-closed@substitution-closed language equivalence"" (\AP""match@substitution-closed match-language equivalence""-)\reintro*\kl[substitution-closed language equivalence]{language equivalence}.
By definition, they are characterized by the equivalence of the ""match-languages with substitutions"" $\AP\intro*\smlang(\term)$ ("resp", the ""languages  with substitutions"" $\AP\intro*\swlang(\term)$), defined as follows:
\begin{align*}
\reintro*\smlang(\term) &~\defeq~~\smashoperator{\bigcup_{\Theta \text{ a "substitution"}}}~ \set{\Theta} \times \mlang(\term\termsubst{\Theta}), &
\reintro*\swlang(\term) &~\defeq~~\smashoperator{\bigcup_{\Theta \text{ a "substitution"}}}~ \set{\Theta} \times \wlang(\term\termsubst{\Theta}).
\end{align*}

For every "string" $\word$ and relation $\rel \subseteq \strucuniv^{\wordstruc^{\word}}$, one can see that $\rel = \semREwLA{\term[3]}{\wordstruc^{\word}}$ holds by some "\REwLA" $\term[3]$.
Thus, every $\struc \in \GRELfinlin{}$ can be ``represented'' by some $\wordstruc^{\word}$ with some "substitution".
From this observation, we see the following:
\ifthenelse{\boolean{conference}}{\begin{proposition}}{
\begin{proposition}[\Cref{section: largest substitution-closed}]}
\label{proposition: largest substitution-closed}%
\gdef\largestsubstitutionclosed{%
$\smlang(\term[1]) = \smlang(\term[2])$ "iff" $\GRELfinlin{} \modelsfml \term[1] = \term[2]$.
}
\largestsubstitutionclosed
\end{proposition}
Additionally, by a similar encoding as in \Cref{proposition: language equivalence} where we use the operator $\capcomid$ (see "\REwLAp" in \Cref{section: PDLREwLAp}), the "substitution-closed language equivalence" also can be characterized, as follows:
\AP\phantomintro\termendmarkertwo
\ifthenelse{\boolean{conference}}{\begin{proposition}}{%
\begin{proposition}[\Cref{section: largest substitution-closed language equivalence}]}%
\label{proposition: largest substitution-closed language equivalence}%
\gdef\largestsubstitutionclosedlanguageequivalence{%
$\swlang(\term[1]) = \swlang(\term[2])$ "iff" $\GRELfinlin{} \modelsfml \term[1] \termendmarkertwo = \term[2] \termendmarkertwo$,
where $\reintro*\termendmarkertwo \defeq \left(\sum_{\aterm \in \exprvsig(\term[1]\, \term[2])} \aterm^{\capcomid} \right)^{\adom}$.
}
\largestsubstitutionclosedlanguageequivalence
\end{proposition}
\noindent
We thus consider $\GRELfinlin{}$ for the "substitution-closed (match-)language equivalence@substitution-closed match-language equivalence".
Such equivalences are useful from the perspective of reusability,
as well as
when extending our system with additional operators ("cf", the axiom \text{\nameref{equation: PDL}} in \Cref{figure: PDLREwLAp axioms} and the axiom \text{\nameref{equation: PROP}} in \Cref{figure: PDL- axioms}).

\Cref{figure: four equivalences} summarizes the inclusions among the equivalences; they are strict by the given "equations".
For $(\aterm \aterm^{\adom})^{\dom} = \emp$,
observe that $\tuple{\vertex[2]_M, \vertex[2]_M} \not\in \sem{(\aterm \aterm^{\adom})^{\dom}}{\struc}$ where $\vertex[2]_M$ is maximum on $\struc$.
Hence, $\swlang((\aterm \aterm^{\adom})^{\dom}) = \emptyset$, although clearly $\mlang((\aterm \aterm^{\adom})^{\dom}) \neq \emptyset$.
\begin{figure}[t]
\centering
\begin{tikzpicture}[baseline = -.5ex]
\tikzstyle{vert}=[draw = black, circle, fill = gray!10, inner sep = 2pt, minimum size = 1pt, font = \scriptsize]
\tikzstyle{class}=[draw = black, rectangle, fill = gray!10, inner sep = 2pt, minimum size = 1.5em, font = \scriptsize]
\graph[grow right = 2.cm, branch down = 1.cm]{
{scm/{$\smlang$}[class, line width = 1.2pt], m/{$\mlang$}[class, line width = 1.2pt]} -!- {scw/{$\swlang$}[class], w/{$\wlang$}[class]}
};
\node[left = 9.5em of scw](scwr) {$\GRELfinlin{}$ {\scriptsize (\Cref{proposition: largest substitution-closed language equivalence,proposition: largest substitution-closed})}};
\node[below = .5cm of scwr](wr){$\GRELstfinlin{}$ {\scriptsize (\Cref{proposition: language equivalence})}};
\node[above = -.2ex of scwr](rel) {\footnotesize (corresponding relational class)};
\draw (scm) edge[opacity=0] node[pos = .5](mmed){} (m);
\draw (scw) edge[opacity=0] node[pos = .5](wmed){} (w);
\draw ($(mmed)+(-6.,0)$) edge[dashed, draw=gray] ($(wmed)+(.5,0)$);
\node[right = .5cm of wmed, align = left](eq){{\footnotesize $((\aterm[1]\aterm[2])^{\adom}(\aterm[1] \union \aterm[2]))^* = \aterm[2]^*\aterm[1]^*$}\\
{\footnotesize $\aterm[1]^{\dom} \aterm[2] = \emp$ \qquad (\Cref{section: introduction})}
};
\draw (scm) edge[opacity=0] node[pos = .5](scmed){} (scw);
\draw (m) edge[opacity=0] node[pos = .5](stmed){} (w);
\draw ($(scmed)+(0,+.2)$) edge[dashed, draw=gray] ($(stmed)+(0,-.2)$);
\node[above = 1ex of scmed, align = right](eq2){%
{\footnotesize $(\aterm[1]\aterm[1]^{\adom})^{\dom} = \emp$}
};
\graph[use existing nodes, edges={color=black, pos = .5, earrow}, edge quotes={fill=white, inner sep=1pt,font= \scriptsize}]{
    scm -> scw;
    scm -> m;
    scw -> w;
    m -> w;
};
\end{tikzpicture}
\caption{Inclusions among the four language equivalences.
Each arrow from $X$ to $Y$ means that the equivalence induced by $X$ is a subset of that induced by $Y$.}
\label{figure: four equivalences}
\end{figure}
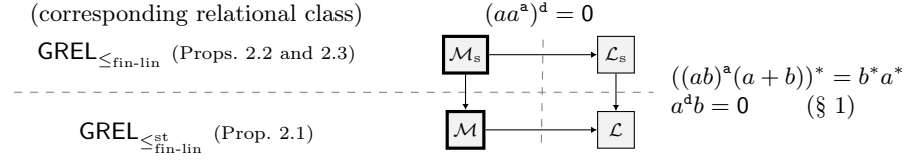
\end{scope}

\begin {scope}\knowledgeimport {PDL}

\subsection{PDL}\label{section: PDL}
\AP
We recall ""propositional dynamic logic"" (\reintro*\kl{\PDL}) of regular programs with rich tests \cite{fischerPropositionalDynamicLogic1979,harelDynamicLogic2000}.
The set of ""formulas@@PDL"" and ""terms@@PDL"" are defined by the following grammar:
\phantomintro{\fmlclassPDL}
\phantomintro{\termclassPDL}
\phantomintro{test}
\phantomintro{\exprclassPDL}
\phantomintro(PDL){expressions}
\begin{align*}
    \fml[1], \fml[2], \fml[3] \in \reintro*\fmlclassPDL{\vsig, \psig} &\;\Coloneqq\; \afml \mid \fml[1] \to \fml[2] \mid \falsec \mid \bo{\term[1]} \fml && \text{[modal logic, $\afml \in \psig$]} \tag{\reintro*\kl{formulas}}\\
    \term[1], \term[2], \term[3] \in \reintro*\termclassPDL{\vsig, \psig} &\;\Coloneqq\; \aterm
    \mid \term[2] \compo \term[3]
    \mid \term[2] \union \term[3]
    \mid \term[2]^{+} && \text{[regular programs, $\aterm \in \vsig$]}\\
    &\qquad\mid \fml[1]? && \text{[\reintro*\kl{test}]}
    \tag{\reintro*\kl{terms}}
\end{align*}
We use the following abbreviations as usual:
$\lnot \fml \defeq \fml \to \falsec$,
$\fml[1] \lor \fml[2] \defeq \lnot \fml[1] \to \fml[2]$,
$\fml[1] \land \fml[2] \defeq \lnot ((\lnot \fml[1]) \lor (\lnot \fml[2]))$,
$(\fml[1] \leftrightarrow \fml[2]) \defeq (\fml[1] \to \fml[2]) \land (\fml[2] \to \fml[1])$,
$\truec \defeq \lnot \falsec$,
$\dia{\term[1]}\fml[1] \defeq \lnot \bo{\term[1]} \lnot \fml[1]$,
$\id \defeq \truec?$, and
$\emp \defeq \falsec?$.
We call $\bo{\term}$ the ""box operator"" and $\dia{\term}$ the ""diamond operator"".

\AP
Given a set $X$, we consider the following operators on unary relations on $X$:
\begin{align*}
    A \to B & \;\defeq\; (X \setminus A) \cup B,&
    \falsec & \;\defeq\; \emptyset,\\
    \bo{\rel} A & \;\defeq\; \set{\vertex[1] \in X \mid \forall \vertex[2], \tuple{\vertex[1], \vertex[2]} \in \rel \text{ implies } \vertex[2] \in A},&
    A? &\;\defeq\; \set{\tuple{\vertex[1], \vertex[1]} \mid \vertex[1] \in A},
\end{align*}
\AP
Given an $\struc \in \GRELpreorder{\vsig, \psig}$,
the ""semantics@@PDL"" $\intro*\semPDL{\bl}{\struc}$ of "\PDL" is well-defined, where the operators are interpreted as above and as in "\REwLA".

\AP
"\PDL" has a sound and complete Hilbert-style finite axiomatization on $\REL{}$ \cite{gabbayAxiomatizationsLogicsPrograms1977,parikhCompletenessPropositionalDynamic1978}:\footnote{\label{footnote: PDL completeness}%
The axioms were introduced by Segerberg \cite{segerbergCompletenessTheoremModal1982}
and the following completeness was shown independently by Gabbay \cite{gabbayAxiomatizationsLogicsPrograms1977} and Parikh \cite{parikhCompletenessPropositionalDynamic1978}; see \cite{harelDynamicLogic2000}.
Here, precisely, our syntax is a minor variant using "Kleene plus" ($\bl^{+}$).\ifthenelse{\boolean{conference}}{}{ (See \Cref{section: footnote: PDL completeness}.)}}
$\REL{} \modelsfml \fml$ "iff" $\vdashPDL \fml$,
where \AP$\intro*\vdashPDL \fml$ if $\fml$ is "derivable" in the system.
As $\strucuniv^{\struc}$ is not used in the evaluation,
we also have:
$\REL{} \modelsfml \fml$ "iff" $\GRELpreorder{} \modelsfml \fml$.
\end{scope}
\begin{scope} \knowledgeimport{PDLREwLAp}
\section{PDL for REwLA+}\label{section: PDLREwLAp}
\AP
In this section, we introduce PDL for a slight extension of "\REwLA".
""Extended regular expressions with lookahead"" (\reintro*\kl{\REwLAp}) are generated by the following grammar:
\phantomintro\termclassREwLAp
\phantomintro(REwLAp){terms}
\begin{align*}
    \term[1], \term[2], \term[3] \in \reintro*\termclassREwLAp{\vsig} &\;\Coloneqq\; \aterm
    \mid \term[2] \compo \term[3]
    \mid \term[2] \union \term[3]
    \mid \term[2]^{+}
    \mid \term[2]^{\adom} \hspace{4em} \tag*{[syntax of "\REwLA"]}\\
    &\hspace{-2em}\mid \term[2]^{\capid} \tag*{[restriction to the identity relation]}\\
    &\hspace{-2em}\mid \term[2]^{\capcomid} \tag*{[restriction to the complement of the identity relation]}
\end{align*}
Given a set $X$, we define the following two operators on $\pset(X^2)$ as follows:%
\footnote{The operators $\bl^{\capid}$ and $\bl^{\capcomid}$ can be found, "eg",
in \cite[strong loop predicate]{daneckiPropositionalDynamicLogic1984}\cite[graph loop]{nakamuraUndecidabilityPositiveCalculus2024} and
in \cite{nakamuraFiniteVariableOccurrenceFragment2023}, respectively.
The notation is based on \cite{nakamuraFiniteVariableOccurrenceFragment2023}.}

\begin{align*}
    R^{\capid} &\;\defeq\; R \cap \diagonal_{X},&
    R^{\capcomid} &\;\defeq\; R \setminus \diagonal_{X}.
\end{align*}

\AP
For the "\PDL" of "\REwLAp", denoted by ""\PDLREwLAp"",
the set of ""formulas"" and ""terms"" are mutually defined by the following grammar:
\phantomintro\fmlclassPDLREwLAp
\phantomintro\termclassPDLREwLAp
\begin{align*}
    \fml[1], \fml[2], \fml[3] \in \reintro*\fmlclassPDLREwLAp{\vsig, \psig}
    &\;\Coloneqq\;
    \afml \mid \fml[1] \to \fml[2] \mid \falsec \mid \bo{\term[1]} \fml \tag{\reintro*\kl{formulas}}\\
    \term[1], \term[2], \term[3] \in \reintro*\termclassPDLREwLAp{\vsig, \psig}
    &\;\Coloneqq\;
    \aterm
    \mid \term[2] \compo \term[3]
    \mid \term[2] \union \term[3]
    \mid \term[2]^{+}
    \mid \term[2]^{\adom}
    \mid \term[2]^{\capid} 
    \mid \term[2]^{\capcomid}
    \mid \fml[1]? \tag{\reintro*\kl{terms}}
\end{align*}
\AP
Let $\intro*\exprclassPDLREwLAp{\vsig, \psig} \defeq \fmlclassPDLREwLAp{\vsig, \psig} \dcup \termclassPDLREwLAp{\vsig, \psig}$
denote the set of ""expressions@@PDLREwLAp"".
\AP
\phantomintro(REwLAp){semantics}%
\phantomintro\semREwLAp%
Given an $\struc \in \GRELpreorder{\vsig, \psig}$,
the ""semantics@@PDLREwLAp"" $\intro*\semPDLREwLAp{\bl}{\struc}$ of "\PDLREwLAp"
is well-defined,
where the operators are interpreted as in "\REwLAp" and "\PDL".
We use the notations as "\PDL" (in \Cref{section: PDL}).

\subsection{Embedding equivalences}
The "theories" of "\PDLREwLAp" embed ("substitution-closed" | standard) ("match-language@substitution-closed match-language equivalence" | "language@substitution-closed language equivalence") equivalences,
because we can embed "equations" into "formulas" as an analog of the embedding in "\PDL" \cite[\S 5]{fischerPropositionalDynamicLogic1979}, as follows.
\ifthenelse{\boolean{conference}}{\begin{proposition}}{%
\begin{proposition}[\Cref{section: proposition: REwLAp to PDLREwLAp}]}
\label{proposition: REwLAp to PDLREwLAp}
\gdef\REwLAptoPDLREwLAp{%
Let $\algclass \in \set{\GRELfinlin{}, \GRELstfinlin{}}$.
For all "\REwLAp", $\term[1]$ and $\term[2]$, we have:
\[\algclass \modelsfml \term[1] = \term[2] \quad\iff\quad \algclass \modelsfml \bo{\term[1]}\afml \leftrightarrow \bo{\term[2]}\afml,\]
where $\afml$ is a "fresh" "formula variable".
}
\REwLAptoPDLREwLAp
\end{proposition}

Unlike the complexity gap between "\PDL" ("EXPTIME"-complete \cite{fischerPropositionalDynamicLogic1979}) and "\RE" ("PSPACE"-complete \cite{stockmeyerWordProblemsRequiring1973}),
"\PDLREwLAp" does not increase the complexity from "\REwLAp",
because the "antidomain" ("negative lookahead") can express the "box" and (unary) negation (see also, "eg", \cite{desharnaisKleeneAlgebraDomain2006}\cite[Def.\ 6]{sedlarComplexityKleeneAlgebra2023}).
\ifthenelse{\boolean{conference}}{\begin{proposition}}{%
\begin{proposition}[\Cref{section: proposition: PDLREwLAp to REwLAp}]}%
\label{proposition: PDLREwLAp to REwLAp}%
\gdef\PDLREwLAptoREwLAp{%
There exists some polynomial-time reduction from the "theory" of "\PDLREwLAp" on $\GRELfinlin{}$ ("resp" $\GRELstfinlin{}$)
to the "equational theory" of "\REwLAp" on $\GRELfinlin{}$ ("resp" $\GRELstfinlin{}$).
}
\PDLREwLAptoREwLAp
\end{proposition}

\subsection{Soundness and completeness}
\AP
We write $\intro*\vdashPDLREwLAp \fml$ if the "formula" $\fml$ is derivable in the proof system of \Cref{figure: PDLREwLAp axioms}.
\begin{figure}[t]
    \begin{tcolorbox}[colback=black!3, top = .3ex, bottom = .3ex, left = .3em, right = .3em]
    \begin{minipage}[t]{0.1\columnwidth}
    \textbf{Rules}:
    \end{minipage}
    \hfill
    \begin{minipage}[t]{0.25\columnwidth}
        \vspace{-4ex}
        \begin{align*}
        \begin{prooftree}
            \hypo{\fml}
            \hypo{\fml \to \fml[2]}
            \infer2{\fml[2]}
        \end{prooftree}
        \tag{MP} \label{rule: MP}
        \end{align*}
    \end{minipage}
    \hfill
    \begin{minipage}[t]{0.15\columnwidth}
        \vspace{-4ex}
        \begin{align*}
        \begin{prooftree}
            \hypo{\fml}
            \infer1{\bo{\term}\fml}
        \end{prooftree} 
        \tag{Nec} \label{rule: NEC}
        \end{align*}
    \end{minipage}
    \hfill\mbox{}

    \tcbline
    \textbf{Axioms}:

    \noindent
    \proofcase{\kl{\PDL} axioms} \hfill All \kl{substitution-instances} of  valid \kl{\PDL} \kl{formulas} (on $\REL{}$) \hfill \labeltext{(PDL)}{equation: PDL}

    \noindent
    \proofcase{Axiom for $\adom$} \hfill $\bo{\term[1]^{\adom}}\fml[1] \leftrightarrow \bo{\bo{\term[1]}\falsec?} \fml[1]$ \hfill \labeltext{($\adom$)}{equation: adom}

    \noindent
    \proofcase{Axioms for $\capid$ and $\capcomid$}

    \noindent
    \hspace{-.5em}
    \begin{minipage}[b]{0.42\columnwidth}
    \vspace{-2ex}
    \scalebox{.95}{\parbox{1.05\linewidth}{%
    \begin{align*}
    &    \bo{\term[1]^{\capid}} \fml[1] \leftrightarrow \bo{\dia{\term[1]^{\capid}}\truec?} \fml[1] \tag{$\capid$-T} \label{equation: T capid}\\
    &    \bo{(\term[1] \term[2])^{\capid}} \fml[1] \leftrightarrow \bo{\term[1]^{\capid} \term[2]^{\capid}}\fml[1] \tag{$\capid$-$\compo$} \label{equation: capid-compo}\\
    &   \bo{(\term[1] \union \term[2])^{\capid}} \fml[1] \leftrightarrow \bo{\term[1]^{\capid} \union \term[2]^{\capid}}\fml[1] \tag{$\capid$-$\union$} \label{equation: capid-union}\\
    &    \bo{(\term[1]^+)^{\capid}} \fml[1] \leftrightarrow \bo{\term[1]^{\capid}} \fml[1] \tag{$\capid$-$\bl^{+}$} \label{equation: capid-*}\\
    & \bo{(\term[1]^{\adom})^{\capid}} \fml[1] \leftrightarrow \bo{\term[1]^{\adom}} \fml[1] \tag{$\capid$-$\bl^{\adom}$} \label{equation: capid-adom}\\
    &    \bo{(\term[1]^{\capid})^{\capid}} \fml[1] \leftrightarrow \bo{\term[1]^{\capid}} \fml[1] \tag{$\capid$-$\bl^{\capid}$} \label{equation: capid-capid}\\
    &    \bo{(\term[1]^{\capcomid})^{\capid}} \fml[1] \leftrightarrow \truec \tag{$\capid$-$\bl^{\capcomid}$} \label{equation: capid-capcomid}\\
    &    \bo{(\fml[2]?)^{\capid}} \fml[1] \leftrightarrow \bo{\fml[2]?} \fml[1] \tag{$\capid$-$?$} \label{equation: capid-test}
  \end{align*}
    }}
    \end{minipage}
    \hfill
    \hspace{-0.2em}
    \begin{minipage}[b]{0.569\columnwidth}
    \vspace{-2ex}
    \scalebox{.95}{\parbox{1.05\linewidth}{%
    \begin{align*}
    &    \bo{\term[1]} \fml[1] \leftrightarrow \bo{\term[1]^{\capid} \union \term[1]^{\capcomid}} \fml[1] \tag{$\capid$-$\union$-$\capcomid$} \label{equation: capcomid-union-capid}\\
    &    \bo{(\term[1] \term[2])^{\capcomid}} \fml[1] \leftrightarrow \bo{\term[1]^{\capcomid} \term[2]^{\capid} \union \term[1]^{\capid} \term[2]^{\capcomid} \union \term[1]^{\capcomid} \term[2]^{\capcomid}} \fml[1]\tag{$\capcomid$-$\compo$} \label{equation: capcomid-compo}\\
    &    \bo{(\term[1] \union \term[2])^{\capcomid}} \fml[1] \leftrightarrow \bo{\term[1]^{\capcomid} \union \term[2]^{\capcomid}} \fml[1]\tag{$\capcomid$-$\union$} \label{equation: capcomid-union}\\
    &    \bo{(\term[1]^{+})^{\capcomid}} \fml[1] \leftrightarrow \bo{(\term[1]^{\capcomid})^{+}} \fml[1] \tag{$\capcomid$-$\bl^{+}$} \label{equation: capcomid-*}\\
    & \bo{(\term[1]^{\adom})^{\capcomid}} \fml[1] \leftrightarrow \truec \tag{$\capcomid$-$\bl^{\adom}$} \label{equation: capcomid-adom}\\
    &    \bo{(\term[1]^{\capid})^{\capcomid}} \fml[1] \leftrightarrow \truec \tag{$\capcomid$-$\bl^{\capid}$} \label{equation: capcomid-capid}\\
    &    \bo{(\term[1]^{\capcomid})^{\capcomid}} \fml[1] \leftrightarrow \bo{\term[1]^{\capcomid}} \fml[1] \tag{$\capcomid$-$\bl^{\capcomid}$} \label{equation: capcomid-capcomid}\\
    &    \bo{(\fml[2]?)^{\capcomid}} \fml[1] \leftrightarrow \truec \tag{$\capcomid$-$?$} \label{equation: capcomid-test}
    \end{align*}
    }}
    \end{minipage}

    \vspace{-1ex}
    \proofcase{(Restricted) \kl{L{\"o}b's axiom}}
    \hfill $\bo{(\term^{\capcomid})^{+}}(\bo{(\term^{\capcomid})^{+}} \fml \to \fml) \to \bo{(\term^{\capcomid})^{+}}\fml$
    \hfill \labeltext{(L{\"o}b-$\bl^{\capcomid +}$)}{equation: Lob'}
    \end{tcolorbox}
    \caption{$\HilbertstylePDLREwLAp$: Rules and axioms for \kl{\PDLREwLAp} on \kl{finite linear orders}.}
    \label{figure: PDLREwLAp axioms}
\end{figure}
Intuitively, the axioms are given based on as follows.
\begin{itemize}
    \item The axiom for $\adom$ is for replacing the "antidomain" $\term^{\adom}$ with the "box" test $\bo{\term}\falsec?$;

    \item The axioms for $\capid$ and $\capcomid$ are
    for transforming the "formula" so that both $\bl^{\capid}$ and $\bl^{\capcomid}$ only apply to "term variables";

    \item A variant of the \AP""L{\"o}b's axiom"" \nameref{equation: Lob'} is employed for the well-foundedness of $\GRELfinlin{}$.
\end{itemize}

In this paper, we prove the following completeness result.
\begin{theorem}[Main theorem]\label{theorem: PDL REwLA completeness}
    For every "\PDLREwLAp" "formula" $\fml$, we have:
    \[\GRELfinlin{} \modelsfml \fml \quad\iff\quad {} \vdashPDLREwLAp \fml.\]
\end{theorem}
The soundness ($\Longleftarrow$) is straightforward.
Note that some axioms only hold on $\GRELfinlin{}$.
For instance, the axiom \eqref{equation: capid-compo} fails on $\REL{}$; consider the following case:
\begin{tikzpicture}[baseline = -.5ex]
\tikzstyle{vert}=[draw = black, circle, fill = gray!10, inner sep = 2pt, minimum size = 1pt, font = \scriptsize]
\graph[grow right = 1.5cm, branch down = 4.5ex]{
{0/{}[vert]} -!- {1/{}[vert]}
};
\graph[use existing nodes, edges={color=black, pos = .5, earrow}, edge quotes={fill=white, inner sep=1pt,font= \scriptsize}]{
    0 ->["$\term[1]$", bend left = 15] 1;
    1 ->["$\term[2]$", bend left = 15] 0;
};
\end{tikzpicture}
(which does not appear in $\GRELfinlin{}$ by antisymmetricity).
This axiom \eqref{equation: capid-compo} is based on \cite[(3.1)]{andrekaEquationalTheoryKleene2011} for language Kleene lattices.
Also, the axiom \text{\nameref{equation: Lob'}} fails on non-well-founded relations.
\ifthenelse{\boolean{conference}}{}{See \Cref{section: soundness} for more details.}
For the completeness ($\Longrightarrow$),
after an identity-free variant of "\PDL" on "finite strict linear orders" (denoted by $\GRELsfinlin{}$) is introduced in \Cref{section: PDL-},
it is shown by the following two steps:
\begin{enumerate}
    \item We give a reduction from the completeness of the "identity-free \PDL". (\Cref{section: reduction to identity-free})
    \item We show the completeness of the "identity-free \PDL". (\Cref{section: PDL- completeness})
\end{enumerate}

\end{scope}
\begin{scope} \knowledgeimport{PDL-}

\section{Identity-Free PDL on Finite Strict Linear Orders} \label{section: PDL-}
\AP
In this section,
we define ""identity-free \PDL""\footnote{A similar fragment is found in \cite{dasCyclicProofsHypersequents2022},
which, in our setting, can be regarded as ``test-free'' "identity-free PDL". See also \cite[\S 2.2]{kozenTypedKleeneAlgebra1998} for identity-free Kleene algebra.}
(\reintro*\kl{\ifreePDL}, for short), as a syntax fragment of "\PDL".
\AP
The ""formulas@@PDL-"" and ""terms@@PDL-"" are mutually generated by the following grammar:
\phantomintro\fmlclassifreePDL%
\phantomintro\termclassifreePDL%
\begin{align*}
    \fml[1], \fml[2], \fml[3] \in \reintro*\fmlclassifreePDL{\vsig, \psig} &\;\Coloneqq\; \afml \mid \fml[1] \to \fml[2] \mid \falsec \mid \bo{\term[1]} \fml \tag{\reintro*\kl{formulas}}\\
    \term[1], \term[2], \term[3] \in \reintro*\termclassifreePDL{\vsig, \psig}
    &\;\Coloneqq\; \aterm
    \mid \term[2] \compo \term[3]
    \mid \term[2] \union \term[3]
    \mid \term[2]^{+}
    \mid \fml? \compo \term[3]
    \mid \term[2] \compo \fml? \tag{\reintro*\kl{terms}}
\end{align*}
\AP
Let $\intro*\exprclassifreePDL{\vsig, \psig} \defeq \fmlclassifreePDL{\vsig, \psig} \dcup \termclassifreePDL{\vsig, \psig}$
denote the set of ""expressions@@PDL-"".
\AP
Given a "generalized structure" $\struc$ on a \emph{transitive} "frame",
the ""semantics@@PDL-"" $\intro*\semifreePDL{\bl}{\struc}$ of "\ifreePDL" is well-defined,
where the operators are interpreted based on "\PDL".
We use the same notations as with "\PDL" (in \Cref{section: PDL}).

Our syntax restriction is given so that
if each "term variable" $\aterm$ is ``identity-free'' ("ie", $\semifreePDL{\aterm}{\struc} \cap \diagonal_{\univ{\struc}} = \emptyset$),
then each "\ifreePDL" "term@@PDL-" $\term$ is also identity-free.
Thus, for the well-definedness of the "semantics@@PDL-" $\semifreePDL{\bl}{\struc}$ of "\ifreePDL",
we do not have to require the \emph{reflexivity} for $\strucuniv^{\struc}$ in "\ifreePDL", contrary to "\PDL".\footnote{\label{footnote: PDL- and PDL}%
Yet, "\ifreePDL" "formulas@@PDL-" have the same expressive power as "\PDL" "formulas@@PDL" (on $\GRELpreorder{}$). \ifthenelse{\boolean{conference}}{}{(See \Cref{section: footnote: PDL- and PDL}.)}}

\begin{figure}[t]
    \begin{tcolorbox}[colback=black!3, top = .3ex, bottom = .3ex, left = .3em, right = .3em]
    \textbf{Rules}: \eqref{rule: MP} and \eqref{rule: NEC}

    \tcbline
    \textbf{Axioms}:

    \noindent
    \proofcase{Prop. axioms} \hfill All \kl{substitution-instances} of valid propositional formulas \hfill \labeltext{(Prop)}{equation: PROP}

    \noindent
    \proofcase{Normal modal logic axiom $+$ variants of Segerberg's axioms $+$ L{\"o}b's axiom}

    \begin{minipage}[t]{0.40\columnwidth}
        \vspace{-4ex}
        \begin{align*}
            &\bo{\term[1] \term[2]} \fml[1] \;\leftrightarrow\; \bo{\term[1]}\bo{\term[2]}\fml[1]
            \tag{$\compo$} \label{equation: compo}\\
            &\bo{\term[1]^+}\fml[1] \;\leftrightarrow\; (\bo{\term}\fml[1] \land \bo{\term}\bo{\term[1]^+} \fml[1]) \tag{$\bl^{+}$} \label{equation: *}\\
            & \bo{\fml[2]? \compo \term} \fml[1] \leftrightarrow (\fml[2] \to \bo{\term}\fml[1]) \tag{?-L} \label{equation: test L}\\
            &\bo{\term}(\fml[1] \to \fml[2]) \to (\bo{\term}\fml[1] \to \bo{\term} \fml[2])
            \tag{K} \label{equation: K}
        \end{align*}
    \end{minipage}
    \hfill
    \begin{minipage}[t]{0.55\columnwidth}
        \vspace{-4.5ex}
        \begin{align*}
            &\bo{\term[1] \union \term[2]} \fml[1] \;\leftrightarrow\; \bo{\term[1]}\fml[1] \land \bo{\term[2]}\fml[2]
            \tag{$\union$} \label{equation: union}\\
            &(\bo{\term[1]}\fml[1] \land \bo{\term[1]^+}(\fml[1] \to \bo{\term[1]}\fml[1])) \to \bo{\term[1]^+}\fml[1] \tag{$\bl^{+}$-Ind} \label{equation: ind}\\
            &\bo{\term \compo \fml[2]?} \fml[1] \leftrightarrow \bo{\term}(\fml[2] \to \fml[1]) \tag{?-R} \label{equation: test R}\\
            & \bo{\term^{+}}(\bo{\term^{+}} \fml \to \fml) \to \bo{\term^{+}}\fml \tag{L{\"o}b-$\bl^{+}$}\label{equation: Lob}
        \end{align*}
        \end{minipage}
    \end{tcolorbox}
    \caption{$\HilbertstyleifreePDLsfinlin$: Rules and axioms for \kl{\ifreePDL} on \kl{finite strict linear orders}.}
    \label{figure: PDL- axioms}
\end{figure}

\AP
We write $\intro*\vdashifreePDLsfinlin \fml \phantomintro\nvdashifreePDLsfinlin$ if $\fml$ is "derivable" in the system of \Cref{figure: PDL- axioms}.
The rules \eqref{equation: test L} and \eqref{equation: test R} are variants of the axiom for tests in "\PDL".
The rule \eqref{equation: Lob} is a variant of the "L{\"o}b's axiom".
Since $\HilbertstyleifreePDLsfinlin$ has \eqref{rule: NEC}\eqref{equation: K},
from the completeness result of the modal logic $\mathbf{K}$, we have:
\begin{align*}
&\text{For each "valid" $\fml$ on $\REL{}$ with modalities restricted to $\bo{\aterm}$ for a fixed $\aterm \in \vsig$,} \\
&\text{all "substitution-instances" of $\fml$ are "derivable" in $\HilbertstyleifreePDLsfinlin$.} \tag{$\mathbf{K}$}
\label{equation: Ksub}
\end{align*}
This system is sound and complete "wrt" $\GRELsfinlin{}$, which will be shown (\Cref{section: PDL- completeness}).
\begin{theorem}[\Cref{section: PDL- completeness}]\label{theorem: PDL- completeness}
    For every "\ifreePDL" "formula@@PDL-" $\fml$,
    we have:
    \[\GRELsfinlin{} \modelsfml \fml \quad\iff\quad {} \vdashifreePDLsfinlin \fml.\]
\end{theorem}

\end{scope}
\section{Reduction from the Completeness of Identity-Free PDL}\label{section: reduction to identity-free}
In this section,
assuming the completeness theorem for "\ifreePDL" on $\GRELsfinlin{}$ (\Cref{theorem: PDL- completeness}),
we prove the completeness theorem for "\PDLREwLAp" on $\GRELfinlin{}$ (\Cref{theorem: PDL REwLA completeness}).
To this end, we transform "\PDLREwLAp" into a normal form such that
\begin{itemize}
    \item the "antidomain" does not occur (by using the "rule" \text{\nameref{equation: adom}}),
    \item each "term variable" $\aterm$ occurs either in the form $\aterm^{\capid}$ or the form $\aterm^{\capcomid}$ (by using the "rules" for $\capid$ and $\capcomid$).
\end{itemize}
First, the following congruence rules \eqref{rule: cong} are "derivable" in $\vdashPDLREwLAp$ inheriting those in "\PDL"\ifthenelse{\boolean{conference}}{}{ (\Cref{proposition: PDL cong})},
where $\fml[2]$ in $\set{\bl}_{\fml[2]}$ ranges over any "formulas@@PDLREwLAp":%
\begin{align*}
    \begin{aligned}
    &
    \begin{prooftree}[small]
        \hypo{\fml[1]_1 \;\leftrightarrow\; \fml[2]_1}
        \hypo{\fml[1]_2 \;\leftrightarrow\; \fml[2]_2}
        \infer2{\fml[1]_1 \to \fml[1]_2 \;\leftrightarrow\; \fml[2]_1 \to \fml[2]_2}
    \end{prooftree}
    \hspace{2.em}
    \begin{prooftree}[small]
        \hypo{\fml[1] \leftrightarrow \fml[2]}
        \infer1{\bo{\term} \fml[1] \leftrightarrow \bo{\term} \fml[2]}
    \end{prooftree}
    \hspace{2.em}
    \begin{prooftree}[small]
        \hypo{\fml[2] \leftrightarrow \fml[3]}
        \infer1{\bo{\fml[2]?}\fml \leftrightarrow \bo{\fml[3]?}\fml}
    \end{prooftree}    
    \\
    &
    \begin{prooftree}[small]
        \hypo{\set{\bo{\term[1]_1}\fml[2] \leftrightarrow \bo{\term[2]_1}\fml[2]}_{\fml[2]}}
        \hypo{\set{\bo{\term[1]_2}\fml[2] \leftrightarrow \bo{\term[2]_2}\fml[2]}_{\fml[2]}}
        \infer2{\bo{\term[1]_1 \star \term[1]_2}\fml \leftrightarrow \bo{\term[2]_1 \star \term[2]_2}\fml}
    \end{prooftree}
    \mbox{ $\star \in \set{\compo, \union}$}
    \hspace{2.em}
    \begin{prooftree}[small]
        \hypo{\set{\bo{\term[1]}\fml[2] \leftrightarrow \bo{\term[2]}\fml[2]}_{\fml[2]}}
        \infer1{\bo{\term[1]^{+}}\fml \leftrightarrow \bo{\term[2]^{+}}\fml}
    \end{prooftree}
    \end{aligned}
    \tag{Cong}\label{rule: cong}
\end{align*}

We define the following translation.
Intuitively,
$\tonormone{\term}$ and $\tonormtwo{\term}$
express the identity-part and the identity-free-part of a "term@@PDLREwLAp" $\term$, respectively (\Cref{proposition: normal form}).
\begin{definition}\label{definition: normal form}
\AP
    The function $\intro*\tonorm{\bl}$ of type
    \[(\fmlclassPDLREwLAp{} \to \fmlclassPDLREwLAp{})
    \dcup
    (\termclassPDLREwLAp{} \to (\fmlclassPDLREwLAp{} \times \termclassPDLREwLAp{})),\]
    is defined as follows,
    where we write $\tuple{\intro*\tonormone{\term}, \intro*\tonormtwo{\term}} \defeq \tonorm{\term}$ for $\term \in \termclassPDLREwLAp{}$:
    \begin{gather*}
        \tonorm{\afml} \defeq \afml,\hspace{.8em}
        \tonorm{(\fml[2] \to \fml[3])} \defeq \tonorm{\fml[2]} \to \tonorm{\fml[3]},\hspace{.8em} 
        \tonorm{\falsec} \defeq \falsec,\hspace{.8em} 
        \tonorm{(\bo{\term} \fml[2])} \defeq (\tonormone{\term} \to \tonorm{\fml[2]}) \land \bo{\tonormtwo{\term}} \tonorm{\fml[2]}, \\
        \tonormone{\aterm} \defeq \dia{\aterm^{\capid}}\truec,\hspace{.2em}
        \tonormone{(\term[2] \compo \term[3])} \defeq \tonormone{\term[2]} \land \tonormone{\term[3]},\hspace{.2em}
        \tonormone{(\term[2] \union \term[3])} \defeq \tonormone{\term[2]} \lor \tonormone{\term[3]},\hspace{.2em}
        \tonormone{(\term[2]^{+})} \defeq \tonormone{\term[2]},\\
        \tonormone{(\term[2]^{\adom})} \defeq \lnot \tonormone{\term[2]} \land \bo{\tonormtwo{\term[2]}} \falsec,\hspace{.8em}
        \tonormone{(\term[2]^{\capid})} \defeq \tonormone{\term[2]},\hspace{.8em}
        \tonormone{(\term[2]^{\capcomid})} \defeq \falsec,\hspace{.8em}
        \tonormone{(\fml[1]?)} \defeq \tonorm{\fml[1]}, \\
        \tonormtwo{\aterm} \defeq \aterm^{\capcomid},\hspace{.5em}
        \tonormtwo{(\term[2] \compo \term[3])} \defeq
        \tonormtwo{\term[2]} \tonormone{\term[3]}? \union
        \tonormone{\term[2]}? \tonormtwo{\term[3]} \union
        \tonormtwo{\term[2]} \tonormtwo{\term[3]},\hspace{.5em}
        \tonormtwo{(\term[2] \union \term[3])} \defeq \tonormtwo{\term[2]} \union \tonormtwo{\term[3]},\\ 
        \tonormtwo{(\term[2]^{+})} \defeq (\tonormtwo{\term[2]})^{+},\hspace{.35em}
        \tonormtwo{(\term[2]^{\adom})} \defeq \emp,\hspace{.35em}
        \tonormtwo{(\term[2]^{\capid})} \defeq \emp,\hspace{.35em}
        \tonormtwo{(\term[2]^{\capcomid})} \defeq \tonormtwo{\term[2]},\hspace{.35em}
        \tonormtwo{(\fml[1]?)} \defeq \emp. \tag*{\lipicsEnd}
    \end{gather*}
\end{definition}
\ifthenelse{\boolean{conference}}{\begin{proposition}}{%
\begin{proposition}[\Cref{section: proposition: normal form}]}%
\label{proposition: normal form}%
\gdef\propositionnormalform{%
For every "expression@@PDLREwLAp" $\expr \in \exprclassPDLREwLAp{}$, we have the following.
\begin{enumerate}
    \item \label{proposition: normal form 1} If $\expr = \fml$ is a "formula@@PDLREwLAp", ${} \vdashPDLREwLAp \tonorm{\fml} \leftrightarrow \fml$.
    \item \label{proposition: normal form 2} If $\expr = \term$ is a "term@@PDLREwLAp", ${} \vdashPDLREwLAp \bo{\tonormone{\term}?} \fml[3] \leftrightarrow \bo{\term^{\capid}} \fml[3]$, where $\fml[3]$ is any "formula@@PDLREwLAp".
    \item \label{proposition: normal form 3} If $\expr = \term$ is a "term@@PDLREwLAp", ${} \vdashPDLREwLAp \bo{\tonormtwo{\term}} \fml[3] \leftrightarrow \bo{\term^{\capcomid}} \fml[3]$, where $\fml[3]$ is any "formula@@PDLREwLAp".
\end{enumerate}
}
\propositionnormalform
\end{proposition}
\begin{proof}
    By easy induction on $\expr$, using the "axioms" for $\capid$ and $\capcomid$ with \eqref{rule: cong}.
\end{proof}

Let $\psig'$ and $\vsig'$ be sets disjoint from $\psig$ and $\vsig$ (and having the same "cardinality" as $\vsig$),
let $\pi_1$ be a bijection from $\psig'$ to $\vsig$, and
let $\pi_2$ be a bijection from $\vsig'$ to $\vsig$.
\AP
Let $\intro*\fromifreePDLsubst$ be the "substitution" mapping
each $\afml \in \psig$ to itself,
each $\afml \in \psig'$ to $\dia{\pi_1(\afml)^{\capid}} \truec$, and
each $\aterm \in \vsig'$ to $\pi_2(\aterm)^{\capcomid}$.
For each "\ifreePDL" "expression@@PDL-" $\expr[2] \in \exprclassifreePDL{\vsig', \psig \dcup \psig'}$,
\AP
we write $\intro*\fromifreePDL{\expr[2]} \in \exprclassPDLREwLAp{\vsig, \psig}$ for the "\PDLREwLAp" "expression@@PDLREwLAp" obtained from $\expr[2]$
by applying $\Theta_0$.
\AP
We observe that for every $\fml \in \fmlclassPDLREwLAp{\vsig, \psig}$,
there is some $\fml[2] \in \fmlclassifreePDL{\vsig', \psig \dcup \psig'}$ such that $\tonorm{\fml} = \fromifreePDL{\fml[2]}$.
By construction, we have the following lemma.\footnote{We may explicitly write, "eg", $\GRELfinlin{\vsig, \psig}$ instead of $\GRELfinlin{}$, for clarity.}
\begin{lemma}\label{lemma: finite linear orders to finite strict linear orders}
    For every "formula@@PDLREwLAp" $\fml \in \fmlclassifreePDL{\vsig', \psig \dcup \psig'}$,
    we have:
    \[\GRELfinlin{\vsig, \psig} \modelsfml \fromifreePDL{\fml} \quad\Longrightarrow\quad \GRELsfinlin{\vsig', \psig \dcup \psig'} \modelsfml \fml.\]
\end{lemma}
\begin{proof}
    Let $\struc \in \GRELsfinlin{\vsig', \psig \dcup \psig'}$.
    We define the $\struc[2] \in \GRELfinlin{\vsig, \psig}$ as follows:
    \begin{align*}
        \univ{\struc[2]} &~\defeq~ \univ{\struc},&
        \strucuniv^{\struc[2]} &~\defeq~ \strucuniv^{\struc} \cup \diagonal_{\univ{\struc}},\\
        \afml^{\struc[2]} &~\defeq~ \afml^{\struc} \text{ for $\afml \in \psig$},&
        \aterm^{\struc[2]} &~\defeq~ \semPDLREwLAp{\pi_1^{-1}(\aterm)^{\capid}}{\struc} \cup \semPDLREwLAp{\pi_2^{-1}(\aterm)^{\capcomid}}{\struc} \text{ for $\aterm \in \vsig$}.
    \end{align*}
    Then, 
    $\afml^{\struc[1]} = \semPDLREwLAp{\dia{\pi_1(\afml)^{\capid}}\truec}{\struc[2]}$ for $\afml \in \psig'$ and 
    $\aterm^{\struc[1]} = \semPDLREwLAp{\pi_2(\aterm)^{\capcomid}}{\struc[2]}$ for $\aterm \in \vsig'$.
    By easy induction on $\fml[2]$,
    we have
    $\semifreePDL{\fml[2]}{\struc[1]} = \semPDLREwLAp{\fromifreePDL{\fml[2]}}{\struc[2]}$ for all $\fml[2] \in \fmlclassifreePDL{\vsig', \psig \dcup \psig'}$.
    By $\semPDLREwLAp{ \fromifreePDL{\fml}}{\struc[2]} = \univ{\struc[2]}$, we have $\semifreePDL{\fml}{\struc[1]} = \univ{\struc[1]}$.
    Hence, this completes the proof.
\end{proof}
We also have the following lemma.
\begin{lemma}\label{lemma: identity-free PDL to PDLREwLA}
    For every "formula@@PDL-" $\fml \in \fmlclassifreePDL{\vsig', \psig \dcup \psig'}$,
    we have:
    \[{}\vdashifreePDLsfinlin \fml \quad\Longrightarrow\quad {}\vdashPDLREwLAp \fromifreePDL{\fml}.\]
\end{lemma}
\begin{proof}
    By induction on the derivation tree of ${} \vdashifreePDLsfinlin \fml$.
    Since $\fromifreePDL{\fml}$ is a "substitution-instance" of $\fml$ and each "rule" is closed under "substitutions" ("wrt" both "formula variables" and "term variables"),
    the crucial cases are for \eqref{equation: test L}, \eqref{equation: test R}, and \eqref{equation: Lob},
    which are not equipped in "\PDLREwLAp".
    For \eqref{equation: test L} and \eqref{equation: test R},
    we can show these cases easily by \text{\nameref{equation: PDL}}.
    The remaining case is for \eqref{equation: Lob}.
    We prepare the following claims.
    \begin{claim}\label{lemma: identity-free PDL to PDLREwLA lem}
    For every $\term \in \termclassifreePDL{\vsig', \psig \dcup \psig'}$,
    ${} \vdashPDLREwLAp \tonormone{\fromifreePDL{\term}} \leftrightarrow \falsec$.
    \end{claim}
    \begin{claimproof}
    By induction on $\term$.
    Below, we show some selected cases.
    \proofcase{Case $\term = \aterm$}
    By $\tonormone{(\fromifreePDL{\aterm})}
    = \tonormone{(\pi_2(\aterm)^{\capcomid})}
    = \falsec$.
    \proofcase{Case $\term = \fml[1]? \compo \term[3]$}
    By $\tonormone{(\fromifreePDL{(\fml[1]? \compo \term[3])})}
    = \tonormone{(\fromifreePDL{\fml[1]?})} \land \tonormone{(\fromifreePDL{\term[3]})}
    ~\leftrightarrow_{\text{IH}}~ \tonormone{(\fromifreePDL{\fml[1]?})} \land \falsec
    ~\leftrightarrow_{\text{\nameref{equation: PDL}}}~ \falsec$.
    \end{claimproof}
    \begin{claim}\label{lemma: identity-free PDL to PDLREwLA lem 2}
    For every $\term \in \termclassifreePDL{\vsig', \psig \dcup \psig'}$,
    ${} \vdashPDLREwLAp \tonorm{(\bo{\fromifreePDL{\term}^{+}}\fml[3])} \leftrightarrow \bo{(\fromifreePDL{\term}^{\capcomid})^{+}} \tonorm{\fml[3]}$,
    where $\fml[3]$ is any "formula@@PDLREwLAp".
    \end{claim}
    \begin{claimproof}
        We have:
        \begin{align*}
        \tonorm{(\bo{\fromifreePDL{\term}^{+}}\fml[3])}
        &
        \;\leftrightarrow_{\substack{\text{claim}\\\text{above}}}\;
        ((\falsec \to \tonorm{\fml[3]}) \land \bo{\tonormtwo{(\fromifreePDL{\term}^{+})}} \tonorm{\fml[3]})
        \;\leftrightarrow_{\text{\nameref{equation: PDL}}}\;
        \bo{\tonormtwo{(\fromifreePDL{\term}^{+})}} \tonorm{\fml[3]}\\
        &\;\leftrightarrow_{\text{\Cref{proposition: normal form}}}\;
        \bo{(\fromifreePDL{\term}^{+})^{\capcomid}} \tonorm{\fml[3]}
        \;\leftrightarrow_{\eqref{equation: capcomid-*}}\;
        \bo{(\fromifreePDL{\term}^{\capcomid})^{+}} \tonorm{\fml[3]}. \tag*{\claimqedhere}
        \end{align*}
    \end{claimproof}
    Now, let $\fml = \bo{\term^{+}}(\bo{\term^{+}} \fml[2] \to \fml[2]) \to \bo{\term^{+}} \fml[2]$.
    To prove $\vdashPDLREwLAp \fromifreePDL{\fml}$,
    by \Cref{proposition: normal form},
    it suffices to prove $\vdashPDLREwLAp \tonorm{(\fromifreePDL{\fml})}$.
    By the claim above,
    the "formula@@PDLREwLAp" $\tonorm{(\fromifreePDL{\fml})}$ is equivalent to:
    \[\bo{(\fromifreePDL{\term}^{\capcomid})^{+}}(\bo{(\fromifreePDL{\term}^{\capcomid})^{+}} \tonorm{(\fromifreePDL{\fml})} \to \tonorm{(\fromifreePDL{\fml})}) \to \bo{(\fromifreePDL{\term}^{\capcomid})^{+}} \tonorm{(\fromifreePDL{\fml})}.\]
    As this "formula@@PDLREwLAp" is exactly \text{\nameref{equation: Lob'}}, this completes the proof.
\end{proof}

\begin{proof}[Proof of \Cref{theorem: PDL REwLA completeness} assuming \Cref{theorem: PDL- completeness}]
    Summarizing the lemmas above,
    by letting $\fml[2]$ be such that $\fromifreePDL{\fml[2]} = \tonorm{\fml[1]}$, we have:
    \begin{align*}
        &\GRELfinlin{\vsig, \psig} \modelsfml \fml
        ~\Longrightarrow_{\substack{\text{\Cref{proposition: normal form}}\\\text{with soundness}}}~
        \GRELfinlin{\vsig, \psig} \modelsfml \fromifreePDL{\fml[2]}
        ~\Longrightarrow_{\text{\Cref{lemma: finite linear orders to finite strict linear orders}}}\\
        &
        \GRELsfinlin{\vsig', \psig \dcup \psig'} \modelsfml \fml[2]
        ~\Longrightarrow_{\substack{\text{\Cref{theorem: PDL- completeness}}\\\text{(later)}}}~
        \vdashifreePDLsfinlin \fml[2]
        ~\Longrightarrow_{\text{\Cref{lemma: identity-free PDL to PDLREwLA}}}~
        ~\vdashPDLREwLAp \fromifreePDL{\fml[2]}\\
        &
        \Longrightarrow_{\text{\Cref{proposition: normal form}}}~
        \vdashPDLREwLAp \fml
        ~\Longrightarrow_{\text{soundness}}
        \GRELfinlin{\vsig, \psig} \modelsfml \fml. \tag*{\qedhere}
    \end{align*}
\end{proof}

\subsection*{Example: showing substitution-closed equivalence in $\HilbertstylePDLREwLAp$}
Finally, we give examples of our derivation system to show the "substitution-closed equivalence" in "\REwLA".
First, for "equations" without $\capid$ and $\capcomid$ and not requiring \text{\nameref{equation: Lob'}},
we can easily prove them by \text{\nameref{equation: PDL}} after removing $\adom$.
\begin{example}\label{example 1}
We prove $\GRELfinlin{} \modelsfml (\term[1] \union \term[2])^{\adom} = \term[1]^{\adom} \compo \term[2]^{\adom}$ \cite[\S 9 (8)]{desharnaisModalSemiringsRevisited2008}\cite[Def. 2.2 (4)]{miyazakiDerivativesRegularExpressions2019},
which is equivalent to show $\bo{(\term[1] \union \term[2])^{\adom}} \afml \leftrightarrow \bo{\term[1]^{\adom} \compo \term[2]^{\adom}} \afml$ (\Cref{proposition: REwLAp to PDLREwLAp}).
By transforming it into
$\bo{\bo{\term[1] \union \term[2]}\falsec?} \afml \leftrightarrow \bo{\bo{\term[1]}\falsec? \compo \bo{\term[2]}\falsec?} \afml$ via \text{\nameref{equation: adom}}\eqref{rule: cong},
this is a consequence of \text{\nameref{equation: PDL}}.
\end{example}
By the same argument, we can also show, for example, the following:\footnote{They are the variants of \cite[Lemma 11]{mamourasEfficientMatchingRegular2024} where \Verb!?>! has been replaced with \Verb!?=!.}
\begin{gather*}
\term[1]^{\dom}\term[2]^{\dom} = \term[2]^{\dom} \term[1]^{\dom}, \quad
\term[1]^{\dom}\term[1]^{\dom} = \term[1]^{\dom}, \quad
(\term[1]^{\dom})^{+} = \term[1]^{\dom}, \quad
(\term[1] \union \term[2])^{\dom} = \term[1]^{\dom} \union \term[2]^{\dom}, \quad
(\term[1]^{\dom} \term[2])^{\dom} = \term[1]^{\dom} \term[2]^{\dom}, \\
(\term[1] \term[2]^{\dom})^{\dom} = (\term[1] \term[2])^{\dom}, \quad
\term[1]^{\dom} \union \term[1]^{\adom} = \id, \quad 
\term[1]^{\dom} \term[1]^{\adom} = \emp, \quad
(\term[1]_1^{\dom} \term[2]_1)^{\dom} \term[1]_2^{\dom} \term[2]_2 =
(\term[1]_1^{\dom} \term[1]_2^{\dom}) \term[2]_1^{\dom} \term[2]_2.
\end{gather*}
Particularly, we can also show the "axioms" in boolean domain semiring \cite{desharnaisModalSemiringsRevisited2008,desharnaisInternalAxiomsDomain2011}:
\begin{gather*}
    \term^{\adom} \term = \emp, \quad
    (\term[1] \term[2])^{\adom} \le (\term[1] \term[2]^{\dom})^{\adom}, \quad
    \term^{\dom} \union \term^{\adom} = \id.
\end{gather*}

For "equations" requiring \text{\nameref{equation: Lob'}},
to apply this "axiom",
we transform "formulas@@PDLREwLAp" using the translation $\tonorm{\bl}$.
\begin{example}\label{example 2}
We prove $\GRELfinlin{} \modelsfml \term[3] \le \term[3] \term[3]^{\adom}$
where $\term[3] \defeq (\aterm[1] \aterm[2]^{\adom})^+ \aterm[1] \aterm[2]^{\dom}$,%
\footnote{E.g., for $\aterm[1] = (\mathsf{a} \union \mathsf{b})^{*}$ and $\aterm[2] = \mathsf{\mathsf{b}}$, on $\GRELstfinlin{}$,
the "term@@REwLAp" $\term[3]$ matches when $\mathsf{a}$ is read at least once and the next "character" is $\mathsf{b}$.
By taking the farthest position that $\term[3]$ matches (such position always exists on $\GRELfinlin{}$ with "finite linear order"),
we can not match $\term[3] \term[3]^{\adom}$ from the position, and thus $\term[3] \le \term[3] \term[3]^{\adom}$ holds on $\GRELfinlin{}$.
}
in our system.
Similar to \Cref{example 1}, 
we show $\bo{\term[3] \compo \bo{\term[3]}\falsec?}\afml \to \bo{\term[3]} \afml$.
By \nameref{equation: PDL}, it suffices to show $\bo{\term[3]}(\bo{\term[3]}\afml \to \afml) \to \bo{\term[3]} \afml$.
Note that
\begin{align*}
    \tonormone{\term[3]}
    &~\leftrightarrow~ \tonormone{\aterm[1]} \land \tonormone{(\aterm[2]^{\adom})} \land \tonormone{\aterm[1]} \land (\lnot \tonormone{(\aterm[2]^{\adom})} \land \bo{\tonormone{(\aterm[2]^{\adom})}}\falsec) ~\leftrightarrow~ \falsec \tag{By \text{\nameref{equation: PDL}}}\\
    \bo{\tonormtwo{\term[3]}}\fml[2]
    &~\leftrightarrow~
    \bo{\tonormtwo{((\aterm[1] \aterm[2]^{\adom} \aterm[1])^+)}}\bo{\tonormone{(\aterm[2]^{\dom})}?}\fml[2] \tag{By $\tonormtwo{(\aterm[2]^{\dom})} \leftrightarrow \falsec$}\\
    &~\leftrightarrow~
    \bo{((\aterm[1] \aterm[2]^{\adom} \aterm[1])^{\capcomid})^{+}}\bo{\tonormone{(\aterm[2]^{\dom})}?}\fml[2]  \tag{By \Cref{proposition: normal form} and \eqref{equation: capcomid-*}}.
\end{align*}
Thus by \text{\nameref{equation: PDL}}\eqref{rule: cong}, the "formula@@PDLREwLAp" above is replaced with:
\begin{align*}
&\bo{((\aterm[1] \aterm[2]^{\adom} \aterm[1])^{\capcomid})^{+}}(\bo{((\aterm[1] \aterm[2]^{\adom} \aterm[1])^{\capcomid})^{+}}\bo{\tonormone{(\aterm[2]^{\dom})}?}\afml \to \bo{\tonormone{(\aterm[2]^{\dom})}?}\afml)
\to \bo{((\aterm[1] \aterm[2]^{\adom} \aterm[1])^{\capcomid})^{+}}\bo{\tonormone{(\aterm[2]^{\dom})}?} \afml.
\end{align*}
Hence, by \text{\nameref{equation: Lob'}}, we have shown the "equation" above.
\end{example}

\begin {scope}\knowledgeimport {PDL-}
\section{Completeness of Identity-Free PDL}\label{section: PDL- completeness}
In this section, we finally prove the completeness theorem for "\ifreePDL" on "finite strict linear orders".
Our proof is based on \cite{kozenElementaryProofCompleteness1981}.
Particularly,
\Cref{lemma: ordering atoms} and the direction ($\Longrightarrow$) of \Cref{lemma: prunning'}.\ref{lemma: prunning' 3} are based on \cite[Section 4]{blackburnLinguisticsLogicFinite1994}\cite[Section 4.4]{blackburnProofSystemFinite1996}.
We first defined the "closure@FL-closure" for "\ifreePDL" based on that for standard PDL ("cf" \cite{fischerPropositionalDynamicLogic1979,kozenElementaryProofCompleteness1981}).
\begin{definition}\label{definition: FL-closure}
\AP
The \intro*\kl{FL-closure} $\intro*\cl(\fml)$ of a "\ifreePDL"  "formula" $\fml$ is the smallest set of "\ifreePDL" "formulas" closed under the following rules:
\begin{align*}
    &\fml \in \cl(\fml),&
    &\fml[2] \to \fml[3] \in \cl(\fml) \Longrightarrow \fml[2], \fml[3] \in \cl(\fml),\\
    &\bo{\term[1]} \fml[2] \in \cl(\fml) \Longrightarrow \fml[2] \in \cl(\fml),&
    &\bo{\term[1] \union \term[2]} \fml[2] \in \cl(\fml) \Longrightarrow \bo{\term[1]}\fml[2], \bo{\term[2]}\fml[2] \in \cl(\fml),\\
    &\bo{\term[1]^{+}} \fml[2] \in \cl(\fml) \Longrightarrow \bo{\term[1]}\bo{\term[1]^{+}}\fml[2] \in \cl(\fml),&
    &\bo{\term[1] \compo \term[2]} \fml[2] \in \cl(\fml) \Longrightarrow \bo{\term[1]}\bo{\term[2]}\fml[2] \in \cl(\fml),\\
    &\bo{\fml[3]? \compo \term[1]} \fml[2] \in \cl(\fml) \Longrightarrow \fml[3] \to \bo{\term[1]} \fml[2] \in \cl(\fml),&
    &\bo{\term[1] \compo \fml[3]?} \fml[2] \in \cl(\fml) \Longrightarrow \bo{\term[1]} (\fml[3] \to \fml[2]) \in \cl(\fml).
\end{align*}
\AP
Additionally, we let $\intro*\tcl(\fml) = \set{\term \mid \bo{\term} \fml[2] \in \cl(\fml)}$.
\lipicsEnd\end{definition}
We observe that $\card \cl(\fml)$ is finite and polynomial in the size of $\fml$\ifthenelse{\boolean{conference}}{}{ (see \Cref{section: definition: FL-closure})}.

For a finite "formula" set $\fmlset = \set{\fml[1]_1, \dots, \fml[1]_n}$,
we write $\atomtof{\fmlset}$ for the "formula" $\bigwedge_{i = 1}^{n} \fml[1]_i$.
For a finite set $\atomset$ of finite "formula" sets,
we write $\atomsettof{\atomset}$ for the "formula" $\bigvee_{\fmlset \in \atomset} \atomtof{\fmlset}$.

\AP A "formula" $\fml$ is \intro*\kl{consistent} if ${} \nvdashifreePDLsfinlin \lnot \fml$.
A finite "formula" set $\fmlset$ is \reintro*\kl{consistent} if the "formula" $\atomtof{\fmlset}$ is "consistent".
An \intro*\kl{atom} of a finite "formula" set $\fmlset = \set{\fml[1]_1, \dots, \fml[1]_n}$ is a \kl{consistent} set $\set{\fml[2]_1, \dots, \fml[2]_n}$, 
where each $\fml[2]_i$ is either $\fml[1]_i$ or $\lnot \fml[1]_i$.
We use $\atom[1], \atom[2], \atom[3], \dotsc$ to denote \kl{atoms}.
For a "formula" $\fml$, we write $\intro*\at(\fml)$ for the set of all \kl{atoms} of $\cl(\fml)$.
Note that $\vdashifreePDLsfinlin \bigvee \at(\fml)$.

Below, we list some basic properties of "atoms".
The following proposition gives saturation arguments to obtain an "atom" from a "consistent" set $\fmlset$,
which are easily shown by \eqref{rule: MP}\eqref{equation: Ksub}.
\ifthenelse{\boolean{conference}}{\begin{proposition}}{%
\begin{proposition}[\Cref{section: proposition: atoms consistent saturates}]}%
\label{proposition: atoms consistent saturates}%
\gdef\propositionatomsaturates{%
For every "formula" $\fml[1]$, $\fml[2]$, $\fml[3]$ and "term" $\term$,
the following hold.
\begin{enumerate}
    \item \label{proposition: atoms consistent saturate}
    If $\fml[1]$ is "consistent",
    either $\fml[1] \land \fml[2]$ or
    $\fml[1] \land \lnot \fml[2]$ is "consistent".

    \item \label{proposition: atoms consistent saturate dia}
    If $\fml[1] \land \dia{\term} \fml[2]$ is "consistent",
    either $\fml[1] \land \dia{\term} (\fml[2] \land \fml[3])$ or
    $\fml[1] \land \dia{\term} (\fml[2] \land \lnot \fml[3])$ is "consistent".
\end{enumerate}
}
\propositionatomsaturates
\end{proposition}
The following proposition gives properties of "atoms", which are shown in the same manner as \cite[Lemma 2]{kozenElementaryProofCompleteness1981}.
\ifthenelse{\boolean{conference}}{\begin{proposition}}{%
\begin{proposition}[\Cref{section: proposition: atoms cond}]}\label{proposition: atoms cond}%
\gdef\propositionatomscondprop{%
For each "formula" $\fml_0$ and "atom" $\atom \in \at(\fml_0)$,
the following hold.
\begin{enumerate}
    \item \label{proposition: atoms cond to}
    For every $\fml[2] \to \fml[3] \in \cl(\fml_0)$,
    $\fml[2] \to \fml[3] \in \atom$ "iff" $\fml[2] \not\in \atom$ or $\fml[3] \in \atom$.
    
    \item \label{proposition: atoms cond false}
    $\falsec \not\in \atom$.
    
    \item \label{proposition: atoms cond union} 
    For every $\bo{\term[2] \union \term[3]} \fml[2] \in \cl(\fml_0)$,
    $\bo{\term[2] \union \term[3]} \fml[2] \in \atom$ "iff"
    $\bo{\term[2]} \fml[2] \in \atom$ and $\bo{\term[3]} \fml[2] \in \atom$.
    
    \item \label{proposition: atoms cond compo}
    For every $\bo{\term[2] \compo \term[3]} \fml[2] \in \cl(\fml_0)$,
    $\bo{\term[2] \compo \term[3]} \fml[2] \in \atom$ "iff"
    $\bo{\term[2]}\bo{\term[3]} \fml[2] \in \atom$.
    
    \item \label{proposition: atoms cond *}
    For every $\bo{\term[2]^{+}} \fml[2] \in \cl(\fml_0)$,
    $\bo{\term[2]^{+}} \fml[2] \in \atom$ "iff"
    $\bo{\term[2]} \fml[2] \in \atom$ and
    $\bo{\term[2]}\bo{\term[2]^{+}} \fml[2] \in \atom$.

    \item \label{proposition: atoms cond test L}
    For every $\bo{\fml[3]? \compo \term[2]} \fml[2] \in \cl(\fml_0)$,
    $\bo{\fml[3]? \compo \term[2]} \fml[2] \in \atom$ "iff"
    $\fml[3] \to \bo{\term[2]} \fml[2] \in \atom$.

    \item \label{proposition: atoms cond test R}
    For every $\bo{\term[2] \compo \fml[3]?} \fml[2] \in \cl(\fml_0)$,
    $\bo{\term[2] \compo \fml[3]?} \fml[2] \in \atom$ "iff"
    $\bo{\term[2]} (\fml[3] \to \fml[2]) \in \atom$.
\end{enumerate}
}
\propositionatomscondprop
\end{proposition}
The following proposition gives the properties of the "consistency",
which are shown in the same manner as \cite[Lemma 1]{kozenElementaryProofCompleteness1981}.
\ifthenelse{\boolean{conference}}{\begin{proposition}}{%
\begin{proposition}[\Cref{section: proposition: atoms consistent}]}
\label{proposition: atoms consistent}%
\gdef\propositionatoms{%
    Let $\fml_0$ be a "formula".
    Let $\atom, \atom' \in \at(\fml_0)$ be "atoms".
    Let $\term[2], \term[3]$ be "terms" and $\fml[3] \in \cl(\fml_0)$ be a "formula".
    Then the following hold.
    \begin{enumerate}
        \item \label{proposition: atoms consistent union} 
        If $\atomtof{\atom} \land \dia{\term[2] \union \term[3]} \atomtof{\atom}'$ is "consistent",
        $\atomtof{\atom} \land \dia{\term[2]} \atomtof{\atom}'$ is "consistent" or $\atomtof{\atom} \land \dia{\term[3]} \atomtof{\atom}'$ is "consistent".
        
        \item \label{proposition: atoms consistent compo}
        If $\atomtof{\atom} \land \dia{\term[2] \compo \term[3]} \atomtof{\atom}'$ is "consistent",
        there is some "atom" $\atom''$ such that $\atomtof{\atom} \land \dia{\term[2]} \atomtof{\atom}''$ and $\atomtof{\atom}'' \land \dia{\term[3]} \atomtof{\atom}'$ are "consistent".
        
        \item \label{proposition: atoms consistent *}
        If $\atomtof{\atom} \land \dia{\term[2]^+} \atomtof{\atom}'$ is "consistent",
        there are $n \ge 1$ and "atoms" $\atom_0'', \dots, \atom_n''$ with $\atom_0'' = \atom$ and $\atom_n'' = \atom'$ such that
        $\atomtof{\atom}_{i}'' \land \dia{\term[2]} \atomtof{\atom}_{i+1}''$ is "consistent" for every $i < n$.

        \item \label{proposition: atoms consistent test L}
        If $\atomtof{\atom} \land \dia{\fml[3]? \compo \term[2]} \atomtof{\atom}'$ is "consistent",
        $\fml[3] \in \atom$ and $\atom \land \dia{\term[2]} \atom'$ is "consistent".

        \item \label{proposition: atoms consistent test R}
        If $\atomtof{\atom} \land \dia{\term[2] \compo \fml[3]?} \atomtof{\atom}'$ is "consistent",
        $\atom \land \dia{\term[2]} \atom'$ is "consistent" and $\fml[3] \in \atom'$.
    \end{enumerate}
}
\propositionatoms
\end{proposition}

\subsection{Ordering Atoms}\label{section: ordering atoms}
In this subsection,
we give a "finite linear order" on the "atoms" of $\fml_0$,
which will be used to construct a \kl{canonical model} in the class $\GRELsfinlin{}$.

Using \nameref{equation: Lob'}, we have the following lemma.
\begin{lemma}\label{lemma: add root}
    Let $\fml_0$ be a "formula".
    If $\atomset \subseteq \at(\fml_0)$ is a non-empty set, then
    there is some $\atom \in \atomset$
    such that $\atomtof{\atom} \land \bo{\term^{+}} \atomsettof{(\at(\fml_0) \setminus \atomset)}$ is \kl{consistent}.
\end{lemma}
\begin{proof}
    Towards a contradiction,
    assume that
    $\vdashifreePDLsfinlin \bo{\term^{+}} \atomsettof{(\at(\fml_0) \setminus \atomset)} \to \lnot \atomtof{\atom}$ for all $\atom \in \atomset$.
    Combining them and 
    $\vdashifreePDLsfinlin \bo{\term^{+}}((\bigvee \at(\fml_0) \setminus \atomset) \leftrightarrow \lnot \bigvee \atomset)$
    (by \nameref{equation: PROP}\eqref{rule: NEC}) with \eqref{rule: MP}\eqref{equation: Ksub}
    yields $\vdashifreePDLsfinlin \bo{\term^{+}} (\lnot \atomsettof{\atomset}) \to \lnot \atomsettof{\atomset}$.
    By applying the \AP""L{\"o}b's rule"" \begin{prooftree}[small]
        \hypo{\bo{\term^{+}} \fml \to \fml}
        \infer1{\fml}
    \end{prooftree},
    which is "derivable" from \eqref{equation: Lob}%
    \ifthenelse{\boolean{conference}}{}{ (see \Cref{section: proposition: Lob})},
    we have
    $\vdashifreePDLsfinlin \lnot \atomsettof{\atomset}$.
    Hence, each $\atomtof{\atom}$ is not "consistent" for every $\atom \in \atomset$.
    This reaches a contradiction, since each "atom" is "consistent" and $\atomset$ is non-empty.
\end{proof}
Using this lemma, we obtain an appropriate "finite linear order" on "atoms".
\begin{lemma}\label{lemma: ordering atoms}
    Let $\fml_0$ be a "formula".
    There is a sequence $\atom_1 \dots \atom_n$ of pairwise distinct \kl{atoms}
    with $\at(\fml_0) = \set{\atom_1, \dots, \atom_n}$
    such that, for all $\bo{\term} \fml[2] \in \cl(\fml_0)$,
    if $\bo{\term} \fml[2] \not\in \atom_i$,
    then there is a $j > i$ such that
    $\atomtof{\atom}_i \land \dia{\term} \atomtof{\atom}_j$ is \kl{consistent} and $\fml[2] \not\in \atom_j$.
\end{lemma}
\begin{proof}
    By applying \Cref{lemma: add root} (with $\term = \sum \exprvsig(\fml_0)$), iteratively,
    there is a sequence $\atom_1 \dots \atom_n$ of pairwise distinct \kl{atoms}
    with $\at(\fml_0) = \set{\atom_1, \dots, \atom_n}$
    such that
    $\atomtof{\atom}_i \land \bo{(\sum \exprvsig(\fml_0))^{+}} \bigvee_{j > i} \atomtof{\atom}_{j}$ is \kl{consistent} for each $i \in \rangeone{n}$.
    By $\vdashifreePDLsfinlin \bo{(\sum \exprvsig(\fml_0))^{+}} \fml[3] \to \bo{\term}\fml[3]$\ifthenelse{\boolean{conference}}{}{ (\Cref{section: proposition: all trace})},
    for every $\term$ with $\exprvsig(\term) \subseteq \exprvsig(\fml_0)$,
    the "formula"
    $\atomtof{\atom}_i \land \bo{\term} \bigvee_{j > i} \atomtof{\atom}_{j}$ is also \kl{consistent} for each $i \in \rangeone{n}$.
    Suppose that
    there is no $j > i$ such that
    $\atomtof{\atom}_i \land \dia{\term} \atomtof{\atom}_j$ is \kl{consistent} and 
    $\fml[2] \not\in \atom_j$;
    namely, $\vdashifreePDLsfinlin \atomtof{\atom}_i \to \lnot \dia{\term} (\lnot \fml[2] \land \atomtof{\atom}_j)$ for all $j > i$.
    Then $\vdashifreePDLsfinlin \atomtof{\atom}_i \to \bo{\term} (\fml[2] \lor \lnot \bigvee_{j > i} \atomtof{\atom}_{j})$ by \eqref{rule: MP}\eqref{equation: Ksub}.
    As $\atomtof{\atom}_i \land \bo{\term} \bigvee_{j > i} \atomtof{\atom}_{j}$ is \kl{consistent},
    $\atomtof{\atom}_i \land \bo{\term} \fml[2]$ is "consistent" by \eqref{rule: MP}\eqref{equation: Ksub}.
    Hence $\bo{\term} \fml[2] \in \atom_i$, which contradicts the assumption.
\end{proof}

\subsection{Canonical Model} \label{section: canonical model}
Let $\fml_0$ be a "formula".
Let $\atom_1, \dots, \atom_n \in \at(\fml_0)$ be the linearly ordered \kl{atoms} obtained from \Cref{lemma: ordering atoms}.
The \AP""canonical model"" $\intro*\canonicalmodel^{\fml_0}$ is the "generalized structure" defined as follows:
\begin{align*}
    \univ{\canonicalmodel^{\fml_0}} &\defeq \set{\atom_1, \dots, \atom_n}, \hspace{3em}
    \strucuniv^{\canonicalmodel^{\fml_0}} \defeq \set{\tuple{\atom_i, \atom_j} \mid 1 \le i < j \le n},\\
    \aterm^{\canonicalmodel^{\fml_0}} &\defeq \set{\tuple{\atom, \atom'} \in \strucuniv^{\canonicalmodel^{\fml_0}} \mid \atomtof{\atom} \land \dia{\aterm}\atomtof{\atom}' \mbox{ is \kl{consistent}}},\ 
    \afml^{\canonicalmodel^{\fml_0}} \defeq \set{\atom \in \univ{\canonicalmodel^{\fml_0}} \mid \afml \in \atom}.
\end{align*}
By definition, we have $\canonicalmodel^{\fml_0} \in \GRELsfinlin{}$.
We have the following truth lemma.

\begin{lemma}[Truth lemma]\label{lemma: prunning'}
    Let $\fml_0$ be a "formula".
    For every "expression" $\expr$,
    we have the following:
    \begin{enumerate}
        \item \label{lemma: prunning' 1}
        If $\expr = \fml \in \cl(\fml_0)$ is a "formula",
        then for all $\atom[2] \in \at(\fml_0)$,
        $\fml \in \atom[2]$ "iff"
        $\atom[2] \in \semifreePDL{\fml}{\canonicalmodel^{\fml_0}}$.

        \item \label{lemma: prunning' 2}
        If $\expr = \term \in \tcl(\fml_0)$ is a "term", then
        for all $\tuple{\atom[2], \atom[2]'} \in \strucuniv^{\canonicalmodel^{\fml_0}}$,
        if $\atomtof{\atom[2]} \land \dia{\term}\atomtof{\atom[2]}'$ is "consistent", then $\tuple{\atom[2], \atom[2]'} \in \semifreePDL{\term}{\canonicalmodel^{\fml_0}}$.

        \item \label{lemma: prunning' 3}
        If $\expr = \term$ is a "term"
        and $\fml[2]$ is a "formula" "st" $\bo{\term}\fml[2] \in \cl(\fml_0)$,
        then for all $\atom[2] \in \at(\fml_0)$,
        $\bo{\term}\fml[2] \not\in \atom[2]$ "iff"
        there is some $\atom[2]'$ such that
        $\tuple{\atom[2], \atom[2]'} \in \semifreePDL{\term}{\canonicalmodel^{\fml_0}}$ and $\fml[2] \not\in \atom[2]'$.
    \end{enumerate}
\end{lemma}
\begin{proof}
    By induction on $\expr$.
    
    \proofcase{For \ref{lemma: prunning' 1}}
    We distinguish the following cases.
    \proofcase{Case $\fml = \afml$}
    By the definition of $\afml^{\canonicalmodel^{\fml_0}}$.
    \proofcase{Case $\fml = \fml[2] \to \fml[3]$, Case $\fml = \falsec$}
    By \Cref{proposition: atoms cond}, with IH.
    \proofcase{Case $\fml = \bo{\term} \fml[2]$}
    We have:
    $\bo{\term} \fml[2] \not\in \atom[2]$ "iff"
    there is some $\atom[2]' \in \at(\fml_0)$ "st" $\tuple{\atom[2], \atom[2]'} \in \semifreePDL{\term}{\canonicalmodel^{\fml_0}}$ and $\fml[2] \not\in \atom[2]'$ (by IH "wrt" $\term$) "iff"
    $\atom[2] \not\in \semifreePDL{\bo{\term} \fml[2]}{\canonicalmodel^{\fml_0}}$ (by IH "wrt" $\fml[2]$).

    \proofcase{For \ref{lemma: prunning' 2}}
    We distinguish the following cases.
    \proofcase{Case $\term = \aterm$}
    By the definition of $\aterm^{\canonicalmodel^{\fml_0}}$,
    \proofcase{Case $\term = \term[2] \union \term[3]$,
    Case $\term = \term[2] \compo \term[3]$,
    Case $\term = \term[2]^+$,
    Case $\term = \fml[3]? \compo \term[2]$,
    Case $\term = \term[2] \compo \fml[3]?$
    }
    By \Cref{proposition: atoms consistent} with IH for each case.

    \proofcase{For $\Longrightarrow$ of \ref{lemma: prunning' 3}}
    Let $\atom[2] = \atom_i$.
    By \Cref{lemma: ordering atoms},
    there is some $j > i$ such that
    $\atomtof{\atom}_i \land \dia{\term} \atomtof{\atom}_j$ is \kl{consistent} and 
    $\fml[2] \not\in \atom_j$.
    By IH of \ref{lemma: prunning' 2}, $\tuple{\atom_i, \atom_j} \in \semifreePDL{\term}{\canonicalmodel^{\fml_0}}$.

    \proofcase{For $\Longleftarrow$ of \ref{lemma: prunning' 3}}
    We distinguish the following cases.
    Below are shown in the same manner as \cite[Lemma 2]{kozenElementaryProofCompleteness1981}.
    In each case, \Cref{definition: FL-closure} is crucial for using IH.

    \subproofcase{Case $\term = \aterm$}
    By the definition of $\aterm^{\canonicalmodel^{\fml_0}}$, the "formula" $\atomtof{\atom[2]} \land \lnot \bo{\aterm} \lnot \atomtof{\atom[2]}'$ is "consistent".
    By $\lnot \fml[2] \in \atom[2]'$ with \eqref{equation: Ksub}, $\atom[2] \land \lnot \bo{\aterm} \fml[2]$ is "consistent".
    We thus have $\bo{\aterm} \fml[2] \not\in \atom[2]$.

    \subproofcase{Case $\term = \term[2] \union \term[3]$}
    By $\tuple{\atom[2], \atom[2]'} \in \semifreePDL{\term[2] \union \term[3]}{\canonicalmodel^{\fml_0}}$,
    we have
    $\tuple{\atom[2], \atom[2]'} \in \semifreePDL{\term[2]}{\canonicalmodel^{\fml_0}}$ or
    $\tuple{\atom[2], \atom[2]'} \in \semifreePDL{\term[3]}{\canonicalmodel^{\fml_0}}$.
    By IH, we have $\bo{\term[2]} \fml[2] \not\in \atom[2]$ or $\bo{\term[3]} \fml[2] \not\in \atom[2]$.
    By \Cref{proposition: atoms cond},
    in either case,
    we have $\bo{\term[2] \union \term[3]} \fml[2] \not\in \atom[2]$.
    
    \subproofcase{Case $\term = \term[2] \compo \term[3]$}
    By $\tuple{\atom[2], \atom[2]'} \in \semifreePDL{\term[2] \compo \term[3]}{\canonicalmodel^{\fml_0}}$,
    there is some $\atom[2]''$ such that $\tuple{\atom[2], \atom[2]''} \in \semifreePDL{\term[2]}{\canonicalmodel^{\fml_0}}$ and $\tuple{\atom[2]'', \atom[2]'} \in \semifreePDL{\term[3]}{\canonicalmodel^{\fml_0}}$.
    By IH,
    $\bo{\term[3]} \fml[2] \not\in \atom[2]''$.
    By IH,
    $\bo{\term[2]} \bo{\term[3]} \fml[2] \not\in \atom[2]$.
    By \Cref{proposition: atoms cond},
    we have $\bo{\term[2] \compo \term[3]} \fml[2] \not\in \atom[2]$.

    \subproofcase{Case $\term = \term[2]^+$}
    By $\tuple{\atom[2], \atom[2]'} \in \semifreePDL{\term[2]^{+}}{\canonicalmodel^{\fml_0}}$,
    there are some $n \ge 1$ and $\atom[2]_0, \dots, \atom[2]_n$ with $\atom[2]_0 = \atom[2]$ and $\atom[2]_n = \atom[2]'$ such that
    $\tuple{\atom[2]_{i}, \atom[2]_{i+1}} \in \semifreePDL{\term[2]}{\canonicalmodel^{\fml_0}}$ for all $i < n$.
    By induction on $i$ from $n-1$ to $0$, we show $\bo{\term[2]^+} \fml[2] \not\in \atom[2]_i$.
    
    \subsubproofcase{Case $i = n - 1$}
    By IH "wrt" $\term[2]$,
    $\bo{\term[2]} \fml[2] \not\in \atom[2]_{n-1}$.
    By \Cref{proposition: atoms cond},
    $\bo{\term[2]^{+}} \fml[2] \not\in \atom[2]_{n-1}$.

    \subsubproofcase{Case $i < n - 1$}
    By IH "wrt" $i$,
    $\bo{\term[2]} \bo{\term[2]^{+}} \fml[2] \not\in \atom[2]_{i}$.
    By \Cref{proposition: atoms cond},
    $\bo{\term[2]^{+}} \fml[2] \not\in \atom[2]_{i}$.

    Hence, we have $\bo{\term[2]^{+}} \fml[2] \not\in \atom[2]$.
    
    \subproofcase{Case $\term = \fml[3]? \compo \term[2]$}
    By $\tuple{\atom[2], \atom[2]'} \in \semifreePDL{ \fml[3]? \compo \term[2]}{\canonicalmodel^{\fml_0}}$,
    we have $\atom[2] \in \semifreePDL{\fml[3]}{\canonicalmodel^{\fml_0}}$ and $\tuple{\atom[2], \atom[2]'} \in \semifreePDL{\term[2]}{\canonicalmodel^{\fml_0}}$.
    By IH "wrt" $\term[2]$, we have $\bo{\term[2]}\fml[2] \not\in \atom[2]$.
    By IH "wrt" $\fml[3]$, we have $\fml[3] \in \atom[2]$.
    Thus by \Cref{proposition: atoms cond},
    we have $\bo{\fml[3]? \compo \term[2]} \fml[2] \not\in \atom[2]$.

    \subproofcase{Case $\term = \term[2] \compo \fml[3]?$}
    By $\tuple{\atom[2], \atom[2]'} \in \semifreePDL{ \term[2] \compo \fml[3]?}{\canonicalmodel^{\fml_0}}$,
    we have $\tuple{\atom[2], \atom[2]'} \in \semifreePDL{\term[2]}{\canonicalmodel^{\fml_0}}$ and $\atom[2]' \in \semifreePDL{\fml[3]}{\canonicalmodel^{\fml_0}}$.
    By IH "wrt" $\fml[3]$, we have $\fml[3] \in \atom[2]'$.
    We thus have $\fml[3] \to \fml[2] \not\in \atom[2]'$.
    By IH "wrt" $\term[2]$, we have $\bo{\term[2]} (\fml[3] \to \fml[2]) \not\in \atom[2]$.
    Thus by \Cref{proposition: atoms cond},
    we have $\bo{\term[2] \compo \fml[3]?} \fml[2] \not\in \atom[2]$.
\end{proof}

By \Cref{lemma: prunning'}, we are now ready to prove the completeness theorem.
\begin{proof}[Proof of \Cref{theorem: PDL- completeness}]
    \proofcase{Soundness ($\Longleftarrow$)}
    Easy.\ifthenelse{\boolean{conference}}{}{ (See \Cref{section: soundness})}
    \proofcase{Completeness ($\Longrightarrow$)}
    We prove the contrapositive.
    Suppose $\nvdashifreePDLsfinlin \fml$, namely $\lnot \fml$ is "consistent", by \text{\nameref{equation: PROP}}.
    By \Cref{proposition: atoms consistent saturates},
    there is an "atom" $\atom \in \at(\fml)$ such that $\lnot\fml \in \atom$.
    By \Cref{lemma: prunning'}, $\atom \in \semifreePDL{\lnot \fml}{\canonicalmodel^{\fml}}$.
    Hence, $\GRELsfinlin{} \not\modelsfml \fml$.
\end{proof}

\subsection{Remark on "identity-free PDL" on REL}
\AP \phantomintro\nvdashifreePDL%
We write $\intro*\vdashifreePDL \fml$ if $\fml$ is derivable in the system of \Cref{figure: PDL- axioms}
without the "axiom" \eqref{equation: Lob}.
As a corollary, we also have the following completeness theorem:
\begin{corollary}["Cf" \Cref{theorem: PDL- completeness}\ifthenelse{\boolean{conference}}{}{ (\Cref{section: theorem: PDL- completeness without Lob})}]\label{theorem: PDL- completeness without Lob}%
    \gdef\theoremPDLcompletenesswithoutLob{%
    For all "\ifreePDL" "formulas" $\fml$,
    \[
        \REL{} \modelsfml \fml \quad\iff\quad {}\vdashifreePDL \fml.
    \]
    }
    \theoremPDLcompletenesswithoutLob
\end{corollary}
\begin{proof}
    \proofcase{Soundness ($\Longleftarrow$)}
    Easy.
    \proofcase{Completeness ($\Longrightarrow$)}
    Similar to \Cref{section: canonical model} where we forget linear ordering ("ie", we 
    set the full relation $\strucuniv^{\canonicalmodel^{\fml}} \defeq
    \set{\tuple{\atom_i, \atom_j} \mid i, j \in \rangeone{n}}$),
    we can construct a "canonical model" (without \eqref{equation: Lob}).
\end{proof}

\end{scope}

\section{Completeness for the (Match-)Language Equivalence}\label{section: completeness match-language equivalence}
In this section,
we moreover show the completeness for the standard "match-language equivalence" and "language equivalence".

\AP
We write $\intro*\vdashstPDLREwLAp \fml$ if the "formula@@PDLREwLAp" $\fml$ is derivable in the system of \Cref{figure: PDLREwLAp axioms}
with the additional axioms in \Cref{figure: additional axioms}.
\begin{figure}[t]
\begin{tcolorbox}[colback=black!3, top = .6ex, bottom = .6ex, left = .3em, right = .3em]
\begin{center}
\begin{minipage}[t]{.23\textwidth}
\vspace{-3ex}
\begin{align*}
& \dia{\aterm^{\capid}} \truec \leftrightarrow \falsec \tag{$\capid \aterm$} \label{equation: capid x}
\end{align*}
\end{minipage}
\hfill
\begin{minipage}[t]{.27\textwidth}
\vspace{-3ex}
\begin{align*}
& \dia{\aterm} \fml \rightarrow \bo{\aterm} \fml \tag{Det-1} \label{equation: Det 1}
\end{align*}
\end{minipage}
\hfill
\begin{minipage}[t]{.39\textwidth}
\vspace{-3ex}
\begin{align*}
& \dia{\aterm} \fml \rightarrow \bo{\aterm[2]} \fml[2] \; \text{ for $\aterm \neq \aterm[2]$}\tag{Det-2} \label{equation: Det 2}
\end{align*}
\end{minipage}
\end{center}
\end{tcolorbox}
\caption{Additional axioms for $\HilbertstylestPDLREwLAp$. Here, $\aterm[1], \aterm[2]$ are \kl{term variables}.}
\label{figure: additional axioms}
\end{figure}
The axiom \eqref{equation: capid x} states that the identity-part of $\aterm$ is empty.
The axiom \eqref{equation: Det 1} states that
the binary relation denoted by $\aterm$ is \emph{deterministic} ("ie", the number of outgoing edges labelled by $\aterm$ is at most one), "cf" \AP""deterministic PDL"" \cite[\S 6.1 (8)]{ben-ariDeterministicPropositionalDynamic1982}.
The axiom \eqref{equation: Det 2} states that the number of outgoing labels is at most one.
Note that these additional axioms are \emph{not \kl{substitution-closed}} ("wrt" \kl{term variables}).

\AP
We also write $\intro*\vdashstifreePDL \fml \phantomintro\nvdashstifreePDL$ if $\fml$ is derivable in the system of \Cref{figure: PDL- axioms} with \eqref{equation: Det 1} and \eqref{equation: Det 2}.
Let $\intro*\GRELstsfinlin{}$ denote the class of all $\struc \in \GRELstfinlin{}$ in which $\strucuniv^{\struc}$ has been replaced with $\strucuniv^{\struc} \setminus \diagonal_{\univ{\struc}}$.
We then have:
\begin{theorem}["Cf" \Cref{theorem: PDL- completeness}\ifthenelse{\boolean{conference}}{}{ (\Cref{section: theorem: PDL- completeness match-language equivalence})}]\label{theorem: PDL- completeness match-language equivalence}%
    \gdef\theoremPDLcompletenessmatchlanguageequivalence{%
    For every "\ifreePDL" "formula@@PDL-" $\fml$,
    \[\GRELstsfinlin{} \modelsfml \fml \quad\iff\quad {} \vdashstifreePDL \fml.\]
    }
    \theoremPDLcompletenessmatchlanguageequivalence
\end{theorem}
\begin{proof}[Proof Sketch]
\proofcase{Soundness ($\Longleftarrow$)}
Easy.
\proofcase{Completeness ($\Longrightarrow$)}
By the construction as in \Cref{section: canonical model}
with pruning unnecessary "edges" and "vertices" using \eqref{equation: Det 1} and \eqref{equation: Det 2},
we can construct a "canonical model" "isomorphic" to $\wordstruc^{\word}$ for some $\word$.
\end{proof}
From this, we have the following completeness theorem.
\begin{theorem}["Cf" \Cref{theorem: PDL REwLA completeness}]\label{theorem: completeness match-language equivalence}
    For every "\PDLREwLAp" "formula@@PDLREwLAp" $\fml$,
    we have:
    \[\GRELstfinlin{} \modelsfml \fml \quad\iff\quad {} \vdashstPDLREwLAp \fml.\]
\end{theorem}
\begin{proof}
\proofcase{Soundness ($\Longleftarrow$)}
Easy.
\proofcase{Completeness ($\Longrightarrow$)}
We can give a reduction from \Cref{theorem: PDL- completeness match-language equivalence}
by the same argument as in \Cref{section: reduction to identity-free},
where the translation $\tonorm{\bl}$ (\Cref{definition: normal form}) is redefined by $\tonormone{\aterm} \defeq \falsec$, using \eqref{equation: capid x}.
\end{proof}

\section{On the Complexity}\label{section: complexity}
In this section, we consider the complexity of the ("match@match-language equivalence"-)"language@language equivalence" equivalences.
The proof of \Cref{theorem: PDL REwLA completeness} only gives the "2NEXPTIME" upper bound;
observe that the "formula@@PDLREwLAp" transformation in \Cref{section: reduction to identity-free} makes an exponential blowup "wrt" the size of the "formula@@PDLREwLAp".
Below, we show that extending "\REwLA" with the two operators $\bl^{\capid}$ and $\bl^{\capcomid}$ (thus reaching "\REwLAp") does not increase the complexity (\Cref{corollary: complexity REwLA standard,corollary: complexity REwLA substitution-closed}).
The following theorems are shown by an analog of the standard automata construction for modal logics.

\ifthenelse{\boolean{conference}}{\begin{theorem}}{%
\begin{theorem}[\Cref{section: theorem: complexity PDLREwLAp substitution-closed}]}
\label{theorem: complexity PDLREwLAp substitution-closed}%
\gdef\theoremcomplexityPDLREwLApsubstitutionclosed{%
For "\PDLREwLAp" ("resp", "\PDL"), the "theory" is "EXPTIME"-complete on $\GRELfinlin{}$.
}
\theoremcomplexityPDLREwLApsubstitutionclosed
\end{theorem}
\begin{proof}[Proof Sketch]
\proofcase{Lower bound}
"\PDL" on $\GRELfinlin{}$ is "EXPTIME"-hard,
by the same reduction as in \cite[Section 4]{fischerPropositionalDynamicLogic1979}, which encodes the non-membership problem of alternating polynomial space Turing machines (see also \cite[p.~18]{spaanComplexityModalLogics1993}\cite[Theorem 2]{afanasievPDLOrderedTrees2005}).
\proofcase{Upper bound}
By a variant of the ``"tree unwinding"'' argument,
we can give a polynomial-time reduction to the "theory" of $\PDLREwLAp$ on finite trees.
Thus by an analog of "eg", \cite{vardiAutomatatheoreticTechniquesModal1986,calvaneseAutomataTheoreticApproachRegular2009,degiacomoLinearTemporalLogic2013},
we can give a polynomial-time reduction to the "emptiness problem" of alternating finite tree automata,
which is in "EXPTIME" \cite{comonTreeAutomataTechniques2007}.
\end{proof}

\ifthenelse{\boolean{conference}}{\begin{theorem}}{%
\begin{theorem}[\Cref{section: theorem: complexity PDLREwLAp standard}]}
\label{theorem: complexity PDLREwLAp standard}%
\gdef\theoremcomplexityPDLREwLApstandard{%
For "\PDLREwLAp" ("resp", "\PDL"), the "theory" is "PSPACE"-complete on $\GRELstfinlin{}$.
}
\theoremcomplexityPDLREwLApstandard
\end{theorem}
\begin{proof}[Proof Sketch]
\proofcase{Lower bound}
The "language equivalence" is "PSPACE"-hard for "\RE" \cite{stockmeyerWordProblemsRequiring1973}.
By the reduction of \Cref{proposition: REwLAp to PDLREwLAp}, the "theory" of "\PDL" on $\GRELstfinlin{}$ is "PSPACE"-hard.
\proofcase{Upper bound}
We can give a reduction to the "emptiness problem" of alternating finite \emph{string} automata, similar to \Cref{theorem: complexity PDLREwLAp substitution-closed},
which is in "PSPACE" \cite{jiangNoteSpaceComplexity1991}.
\end{proof}
We thus have the following complexity results.

\begin{corollary}\label{corollary: complexity REwLA substitution-closed}
The "substitution-closed@substitution-closed equivalence" $($"match-language@substitution-closed match-language equivalence" $|$ "language@substitution-closed language equivalence"$)$ equivalences
for $($"\REwLAp" $|$ "\REwLA"$)$ are all "EXPTIME"-complete.
\end{corollary}
\begin{proof}
\proofcase{Lower bound}
Similar to \Cref{theorem: complexity PDLREwLAp substitution-closed},
as a corollary of the reduction of \cite[Section 4]{fischerPropositionalDynamicLogic1979},
they are also shown to be "EXPTIME"-hard.
\proofcase{Upper bound}
By \Cref{theorem: complexity PDLREwLAp substitution-closed} with \Cref{proposition: REwLAp to PDLREwLAp,proposition: largest substitution-closed,proposition: largest substitution-closed language equivalence}.
\end{proof}

\begin{corollary}\label{corollary: complexity REwLA standard}
The $($"match-language@match-language equivalence" $|$ "language@language equivalence"$)$ equivalences
for $($"\REwLAp" $|$ "\REwLA"$)$ are all "PSPACE"-complete.\footnote{
For the "language equivalence" of "\REwLA",
the "PSPACE" upper bound can also be derived form the polynomial-time reduction to the "language equivalence" of alternating string automata (which is in "PSPACE" \cite{jiangNoteSpaceComplexity1991}) presented by Morihata \cite{morihataTranslationRegularExpression2012}.}
\end{corollary}
\begin{proof}
\proofcase{Lower bound}
The "language equivalence" is "PSPACE"-hard already for "\RE" \cite{stockmeyerWordProblemsRequiring1973}.
By \Cref{proposition: language equivalence}, the "match-language equivalence" is also "PSPACE"-hard.
\proofcase{Upper bound}
By \Cref{theorem: complexity PDLREwLAp standard} with \Cref{proposition: REwLAp to PDLREwLAp,proposition: language equivalence}.
\end{proof}

\begin{remark}\label{remark: REL and GRELfinlin}
For the decidability results in this section, it is crucial that we consider $\GRELfinlin{}$ ("resp", $\GRELstfinlin{}$), not $\REL{}$.
On $\REL{}$, the "equational theory" of $\REwLAp$ (and also "\RE" with $\bl^{\capcomid}$) is $\mathrm{\Pi}^{0}_{1}$-hard \cite{nakamuraUndecidabilityPositiveCalculus2024,nakamuraUndecidabilityEmptinessProblem2025} (more precisely, it is shown from the reduction of \cite[Corollary 7.2]{nakamuraUndecidabilityEmptinessProblem2025},
by replacing each hypothesis of the form $\term \le \id$ with $\term^{\capcomid} = \emp$),
"cf", decidable for "\RE" with $\cap$ \cite{nakamuraPartialDerivativesGraphs2017,nakamuraDerivativesGraphsPositive2025} and for "\RE" with $\com{\id}$ \cite{nakamuraExistentialCalculiRelations2023}.
\end{remark}
\section{Conclusion and Future Work}\label{section: conclusion}
We have introduced "\PDLREwLAp" and presented a sound and complete Hilbert-style finite axiomatization (\Cref{theorem: PDL REwLA completeness}), which characterizes the "substitution-closed" equivalence for both "match-language" and "language@@string".
Moreover, we have presented a sound and complete axiomatization (\Cref{theorem: completeness match-language equivalence}),
characterizing the "match-language equivalence" and "language equivalence".
Additionally, "\PDLREwLAp" is "EXPTIME"-complete for the "substitution-closed equivalences" (\Cref{theorem: complexity PDLREwLAp substitution-closed})
and is "PSPACE"-complete for the standard equivalences (\Cref{theorem: complexity PDLREwLAp standard}).

A near future work is to present a direct (quasi-)equational axiomatization for $\REwLAp$,
in the style of test algebra \cite{hollenbergEquationalAxiomsTest1998} %
or Kleene algebra with (anti)domain \cite{desharnaisKleeneAlgebraDomain2006,desharnaisInternalAxiomsDomain2011,sedlarComplexityKleeneAlgebra2023}.
This paper's work in the paper would possibly be a first step towards this direction,
as the boolean sort of test algebras is heavily based on $\PDL$ \cite{hollenbergEquationalAxiomsTest1998}
and the completeness of Kleene algebras with domain shown in \cite{sedlarComplexityKleeneAlgebra2023} is based on the completeness of test algebras.
Extending the syntax with \AP""lookbehind"" \cite{mamourasEfficientMatchingRegular2024} (and with the ""leftmost anchor"" and ""rightmost anchor"") would also be interesting.
For $\REwLA$ with \AP""backreferences"", we have no recursive axiomatization, because the standard \phantomintro{language emptiness}"language emptiness" %
is undecidable \cite[Theorem 4]{chidaLookaheadsRegularExpressions2023}\cite[Theorem 4]{uezatoRegularExpressionsBackreferences2024} (precisely, $\mathrm{\Pi}^{0}_{1}$-hard).
An open question is whether we can extend the completeness and decidability for "\REwLAp" with full intersection $\cap$ on $\GRELfinlin{}$, "cf" decidable and "EXPSPACE"-complete for the "equational theory" of "\RE" with $\cap$ (and mirror image) \cite{brunetReversibleKleeneLattices2017} on algebras of languages, which is equivalent to the "equational theory" on ``$\mathsf{RSUB}$'' a variant of $\GRELfinlin{}$ \cite{nakamuraFiniteRelationalSemantics2025} (if mirror image is omitted).
Additionally, it would also be interesting to investigate GKAT (guarded Kleene algebra with tests \cite{smolkaGuardedKleeneAlgebra2019})-like fragments of our deterministic PDLs on $\GRELstfinlin{}$, "cf" \cite{benevidesDeterminismPDLRelations2025}.

In this paper, using the two operators $\bl^{\capid}$ and $\bl^{\capcomid}$,
we have shown the completeness for "\PDLREwLAp", via the completeness of "identity-free \PDL".
Another open question is whether we can give a direct axiomatization for pure "\PDL" on $\GRELfinlin{}$ (without $\bl^{\capid}$ or $\bl^{\capcomid}$).
For instance,
the following "formula@@PDL" derived from modal Grzegorczyk logic \cite[p. 96]{segerbergEssayClassicalModal1971},
is valid on $\GRELfinlin{}$.
\begin{align*}
\bo{\term^{*}}(\bo{\term^*}(\fml \to \bo{\term^*} \fml) \to \fml) \to \fml \tag{Grz-$\bl^{*}$} \label{equation: Grz}
\end{align*}

\begin{credits}
\subsubsection{\ackname}
We are grateful to Paul Brunet for introducing his technique \cite{brunetCompleteAxiomatisationFragment2020} for eliminating the identity in algebras of languages.
We also thank anonymous reviewers for their helpful and insightful comments.
This work was supported by JSPS KAKENHI Grant Number JP25K14985.
\end{credits}

\bibliographystyle{splncs04}
\bibliography{main}
\setcounter{secnumdepth}{3}
\appendix
\clearpage

\clearpage

\clearpage
\section{Appendix to \Cref{section: preliminaries}~``\nameref*{section: preliminaries}''}

\begin{scope}\knowledgeimport{REwLA}

\subsection{Proof of \Cref{proposition: language equivalence}} \label{section: proposition: language equivalence}
For $\struc \in \GRELfinlin{}$ and $\vertex[1] \in \univ{\struc}$,
the \AP""generated substructure"" $\intro*\generatedsubstruc{\struc}{\vertex[1]}$ of $\struc$ by $\vertex[1]$
is defined as follows:
\begin{align*}
    \univ{\generatedsubstruc{\struc}{\vertex[1]}} &\defeq
    \set{\vertex[3] \in \univ{\struc} \mid \tuple{\vertex[1], \vertex[3]} \in (\bigcup_{\aterm \in \vsig} \aterm^{\struc})^*}, &
    \strucuniv^{\generatedsubstruc{\struc}{\vertex[1]}} &\defeq \strucuniv^{\struc} \cap \univ{\generatedsubstruc{\struc}{\vertex[1]}}^2, \\
    \aterm^{\generatedsubstruc{\struc}{\vertex[1]}} &\defeq \aterm^{\struc} \cap \univ{\generatedsubstruc{\struc}{\vertex[1]}}^2, &
    \afml^{\generatedsubstruc{\struc}{\vertex[1]}} &\defeq \afml^{\struc} \cap \univ{\generatedsubstruc{\struc}{\vertex[1]}}.
\end{align*}

\begin{lemma}\label{lemma: upward substructure}
    Let $\struc \in \GRELfinlin{}$ and $\vertex_0 \in \univ{\struc}$.
    For all "\REwLA" $\term$, we have:
    $\semREwLA{\term}{\generatedsubstruc{\struc}{\vertex_0}} = \semREwLA{\term}{\struc} \cap \strucuniv^{\generatedsubstruc{\struc}{\vertex_0}}$.
\end{lemma}
\begin{proof}
    Let $\struc[2]$ denote $\generatedsubstruc{\struc}{\vertex_0}$.
    We show that, for all $\tuple{\vertex[1], \vertex[2]} \in \strucuniv^{\struc[2]}$, $\tuple{\vertex[1], \vertex[2]} \in \semREwLA{\term}{\struc[2]}$ "iff" $\tuple{\vertex[1], \vertex[2]} \in \semREwLA{\term}{\struc[1]}$.
    \proofcase{Cases $\term = \aterm, \emp, \id, \term[2] \union \term[3]$}
    By unfolding the definitions (with IH).
    \proofcase{Case $\term = \term[2] \compo \term[3]$}
    \proofcase{($\Longrightarrow$)}
    By unfolding the definitions with IH.
    \proofcase{($\Longleftarrow$)}
    Let $\vertex[3]$ be such that $\tuple{\vertex[1], \vertex[3]} \in \semREwLA{\term[2]}{\struc}$ and $\tuple{\vertex[3], \vertex[2]} \in \semREwLA{\term[3]}{\struc}$.
    By $\vertex[1] \in \univ{\struc[2]}$ and $\tuple{\vertex[1], \vertex[3]} \in \semREwLA{\term[2]}{\struc[1]} \subseteq (\bigcup_{\aterm \in \vsig} \aterm^{\struc})^*$, we have $\vertex[3] \in \univ{\struc[2]}$.
    By IH, we have both $\tuple{\vertex[1], \vertex[3]} \in \semREwLA{\term[2]}{\struc[2]}$ and $\tuple{\vertex[3], \vertex[2]} \in \semREwLA{\term[3]}{\struc[2]}$.
    Hence, $\tuple{\vertex[1], \vertex[2]} \in \semREwLA{\term}{\struc[2]}$.
    \proofcase{Case $\term = \term[2]^{+}$}
    Similar to Case $\term = \term[2] \compo \term[3]$.
    \proofcase{Case $\term = \term[2]^{\adom}$}
    \proofcase{($\Longrightarrow$)}
    We show the contrapositive.
    Suppose $\tuple{\vertex[1], \vertex[1]} \not\in \semREwLA{\term[2]^{\adom}}{\struc}$.
    Let $\vertex[3]$ be such that $\tuple{\vertex[1], \vertex[3]} \in \semREwLA{\term[2]}{\struc}$.
    By $\vertex[1] \in \univ{\struc[2]}$, we have $\vertex[3] \in \univ{\struc[2]}$.
    By IH, $\tuple{\vertex[1], \vertex[3]} \in \semREwLA{\term[2]}{\struc[2]}$.
    Hence, $\tuple{\vertex[1], \vertex[1]} \not\in \semREwLA{\term[2]^{\adom}}{\struc[2]}$.
    \proofcase{($\Longleftarrow$)}
    By unfolding the definitions with IH.
\end{proof}

\begin{proposition*}[Restatement of \Cref{proposition: language equivalence}]
\propositionlanguageequivalence
\end{proposition*}
\begin{proof}
    \proofcase{($\Longleftarrow$)}
    Let $\word$ be any "string".
    By $\semREwLA{\termendmarker}{\wordstruc^{\word}} = \set{\tuple{\len{\word}, \len{\word}}}$ with assumption, we have
    that $\tuple{0, \len{\word}} \in \semREwLA{\term[1]}{\wordstruc^{\word}}$ "iff"
    $\tuple{0, \len{\word}} \in \semREwLA{\term[2]}{\wordstruc^{\word}}$.
    \proofcase{($\Longrightarrow$)}
    We show the contrapositive.
    Suppose $\wordstruc^{\word}$ and $\tuple{\vertex[1]_0, \vertex[2]_0} \in \semREwLA{\term[1] \termendmarker}{\wordstruc^{\word}} \setminus \semREwLA{\term[2] \termendmarker}{\wordstruc^{\word}}$ for some $\word = a_1 \dots a_m$.
    By taking the "generated substructure" by $\vertex[1]_0$,
    "wlog", we can assume $\vertex[1]_0 = 0$.
    By $\semREwLA{\termendmarker}{\wordstruc^{\word}} = \set{\tuple{m, m}}$, $\vertex[2]_0 = m$.
    Thus, $\word \in \wlang(\term[1]) \setminus \wlang(\term[2])$ for some $\word$.
\end{proof}

\subsection{Proof of \Cref{proposition: largest substitution-closed}} \label{section: largest substitution-closed}

\begin{proposition*}[Restatement of \Cref{proposition: largest substitution-closed}]
\largestsubstitutionclosed
\end{proposition*}
\begin{proof}
    \proofcase{($\supseteq$)}
    Since each $\wordstruc^{\word}$ in which each $\aterm^{\wordstruc^{\word}}$ has been replaced with some subset of $\strucuniv^{\wordstruc^{\word}}$ is in the class $\GRELfinlin{}$.
    \proofcase{($\subseteq$)}
    Let $\struc[2] \in \GRELfinlin{}$,
    where $\strucuniv^{\struc[2]} = \set{\tuple{i, j} \in \range{0}{m}^2 \mid i \le j}$ for some $m \ge 0$, up to "isomorphism".
    Let $\word \defeq \const{c}^{m}$ where $\const{c}$ is a fixed "term variable"
    and let 
    \begin{align*}
    \term[3]_{\tuple{i, j}}
    &~\defeq~ (\const{c}^{m-i})^{\dom} (\const{c}^{m-i+1})^{\adom} \const{c}^{j-i} \tag*{where $0 \le i \le j \le m$.}
    \end{align*}
    Observe that $\semREwLA{\term[3]_{\tuple{i, j}}}{\wordstruc^{\word}} = \set{\tuple{i, j}}$.
    For a subset $\rel = \set{\tuple{i_1, j_1}, \dots, \tuple{i_n, j_n}} \subseteq \strucuniv^{\struc}$,
    we define $\term[3]_{\rel} \defeq \sum_{i = 1}^{n} \term[3]_{\tuple{i_i, j_i}}$.
    We then have $\semREwLA{\term[3]_{\rel}}{\wordstruc^{\word}} = \rel$.%
    \footnote{%
    We can take such $\term[3]_{\rel}$ even if $\word$ is over a finite set $A$ (not the singleton $\set{\const{c}}$),
    by replacing $\const{c}$ with $\sum_{\aterm \in A} \aterm$ in the definition $\term[3]_{\tuple{i, j}}$.}
    Let $\Theta \colon \vsig \to \termclassREwLA{\vsig}$ be
    the \kl{substitution} given by $\aterm \mapsto \term[3]_{(\aterm^{\struc[2]})}$.
    We then have $\semREwLA{\term[3][\Theta]}{\wordstruc^{\word}} = \semREwLA{\term[3]}{\struc[2]}$ for every $\term[3]$.
    Hence, by assumption, this completes the proof.
\end{proof}
\end{scope}

\begin{scope}\knowledgeimport{REwLA}
\subsection{Proof of \Cref{proposition: largest substitution-closed language equivalence}} \label{section: largest substitution-closed language equivalence}

\begin{proposition*}[Restatement of \Cref{proposition: largest substitution-closed language equivalence}]
    \largestsubstitutionclosedlanguageequivalence
\end{proposition*}
\begin{proof}
    \proofcase{($\Longleftarrow$)}
    Let $\word$ be a "string" and let $\Theta$ be a "substitution".
    Let $\wordstruc^{\word}\reintro*\strucsubst{\Theta} \in \GRELfinlin{}$ denote the $\wordstruc^{\word}$ in which
    $\aterm^{\wordstruc^{\word}\strucsubst{\Theta}} \defeq \semREwLA{\Theta(\aterm)}{\wordstruc^{\word}}$.
    By easy induction on $\term[3]$,
    we have $\semREwLA{\term[3]}{\wordstruc^{\word}\strucsubst{\Theta}} = \semREwLA{\term[3]\termsubst{\Theta}}{\wordstruc^{\word}}$ for every $\term[3]$.
    By construction, we have $\tuple{m, m} \in \semREwLA{\termendmarkertwo}{\wordstruc^{\word}\strucsubst{\Theta}}$.
    Thus, by assumption,
    $\tuple{0, \len{\word}} \in \semREwLA{\term[1]}{\wordstruc^{\word}\strucsubst{\Theta}}$ "iff" $\tuple{0, \len{\word}} \in \semREwLA{\term[2]}{\wordstruc^{\word}\strucsubst{\Theta}}$.
    Hence, $\word \in \wlang(\term[1]\termsubst{\Theta})$ "iff" $\word \in \wlang(\term[2]\termsubst{\Theta})$.
    
    \proofcase{($\Longrightarrow$)}
    Suppose $\struc \in \GRELfinlin{}$ and $\tuple{\vertex[1]_0, \vertex[2]_0} \in \semREwLA{\term[1] \termendmarkertwo}{\struc} \setminus \semREwLA{\term[2] \termendmarkertwo}{\struc}$,
    where, "wlog", we assume that
    $\aterm^{\struc[1]} = \emptyset$ for each $\aterm \not\in \exprvsig(\term[1] \term[2])$
    and $\strucuniv^{\struc} = \set{\tuple{i, j} \in \range{0}{m}^2 \mid i \le j}$ for some $m \ge 0$.
    By $\tuple{\vertex[1]_0, \vertex[2]_0} \in \semREwLA{\term[1] \termendmarkertwo}{\struc}$ and $\semREwLA{\termendmarkertwo}{\struc} \subseteq \triangle_{\univ{\struc}}$,
    we have $\tuple{\vertex[2]_0, \vertex[2]_0} \in \semREwLA{\termendmarkertwo}{\struc}$,
    and hence there is no outgoing edge from $\vertex[2]_0$.
    Let $\struc[2]$ be the $\struc$ in which the ordering $\strucuniv^{\struc[2]}$ has been replaced so that $\vertex[2]_0$ becomes the maximum element.
    Since $\strucuniv$ does not affect the evaluation, $\semREwLA{\term[2]}{\struc} = \semREwLA{\term[2]}{\struc[2]}$ for all $\term[2]$.
    Thus, "wlog", we can assume $\vertex[2]_0 = m$.
    By taking the "generated substructure" by $\vertex[1]_0$ (\Cref{lemma: upward substructure}),
    "wlog", we can also assume $\vertex[1]_0 = 0$.
    By the same argument as in the proof of \Cref{proposition: largest substitution-closed},
    we can take a "string" $\word$ and a \kl{substitution} $\Theta$ such that $\tuple{0, m} \in \semREwLA{\term[1]\termsubst{\Theta}}{\wordstruc^{\word}} \setminus \semREwLA{\term[2]\termsubst{\Theta}}{\wordstruc^{\word}}$.
\end{proof}

\begin{remark}\label{remark: proposition: largest substitution-closed language equivalence needs b}
    \Cref{proposition: largest substitution-closed language equivalence} fails if we use $\termendmarker = (\sum_{\aterm \in \vsig(\term[1]\, \term[2])} \aterm)^{\adom}$ ("cf", \Cref{proposition: language equivalence}).
    For instance, for the "equation" $\aterm^{\dom} = \emp$,
    we have $\GRELfinlin{} \modelsfml \aterm^{\dom} \aterm^{\adom} = \emp= \emp \aterm^{\adom}$
    but $\swlang(\aterm^{\dom}) \neq \swlang(\emp)$ (by letting $\aterm = \id$).
\end{remark}

\end{scope}

\clearpage
\section{Supplement of \Cref{footnote: PDL completeness}: PDL with Kleene plus}\label{section: footnote: PDL completeness}
In this section, we explicitly present a complete Hilbert system for PDL with "Kleene plus" (instead of "Kleene star")
and give several basic facts.
They are analogs of the results in PDL with "Kleene star".

We write $\reintro*\vdashPDL \fml$ if $\fml$ is derivable in the system of \Cref{figure: PDL axioms}.
\begin{figure}[h]
    \begin{tcolorbox}[colback=black!3, top = .3ex, bottom = .3ex, left = .3em, right = .3em]
    \begin{minipage}[t]{0.15\columnwidth}
    \textbf{Rules}:
    \end{minipage}
    \begin{minipage}[t]{0.25\columnwidth}
        \vspace{-4ex}
        \begin{align*}
        \begin{prooftree}
            \hypo{\fml}
            \hypo{\fml \to \fml[2]}
            \infer2{\fml[2]}
        \end{prooftree}
        \tag*{\eqref{rule: MP}}
        \end{align*}
    \end{minipage}
    \hspace{2em}
    \begin{minipage}[t]{0.15\columnwidth}
        \vspace{-4ex}
        \begin{align*}
        \begin{prooftree}
            \hypo{\fml}
            \infer1{\bo{\term}\fml}
        \end{prooftree} 
        \tag*{\eqref{rule: NEC}}
        \end{align*}
    \end{minipage}

    \tcbline%
    \textbf{Axioms}:

    \noindent
    \proofcase{Prop. axioms} \hfill All \kl{substitution-instances} of valid propositional formulas \hfill \labeltext{(Prop)}{equation: PDL PROP}

    \noindent
    \proofcase{Normal modal logic axiom and Segerberg's axioms}

    \begin{minipage}[t]{0.42\columnwidth}
    \vspace{-3ex}
    \begin{align*}
        &\bo{\term}(\fml[1] \to \fml[2]) \to (\bo{\term}\fml[1] \to \bo{\term} \fml[2])
        \tag{K} \label{equation: PDL K}\\
        &\bo{\term[1] \union \term[2]} \fml[1] \;\leftrightarrow\; \bo{\term[1]}\fml[1] \land \bo{\term[2]}\fml[2]
        \tag{$\union$} \label{equation: PDL union}\\
        &\bo{\term[1] \term[2]} \fml[1] \;\leftrightarrow\; \bo{\term[1]}\bo{\term[2]}\fml[1]
        \tag{$\compo$} \label{equation: PDL compo}
    \end{align*}
    \end{minipage}
    \hfill
    \begin{minipage}[t]{0.56\columnwidth}
    \vspace{-3ex}
    \begin{align*}
        &\bo{\fml[2]?} \fml[1] \;\leftrightarrow\; (\fml[2] \to \fml[1]) \tag{Test} \label{equation: PDL test}\\
        &\bo{\term[1]^+}\fml[1] \;\leftrightarrow\; (\bo{\term}\fml[1] \land \bo{\term}\bo{\term[1]^+} \fml[1]) \tag{$\bl^{+}$} \label{equation: PDL *}\\
        &(\bo{\term[1]}\fml[1] \land \bo{\term[1]^+}(\fml[1] \to \bo{\term[1]}\fml[1])) \to \bo{\term[1]^+}\fml[1] \tag{$\bl^{+}$-Ind} \label{equation: PDL ind}
    \end{align*}
    \end{minipage}
    \end{tcolorbox}
    \caption{$\HilbertstylePDL$: Rules and axioms for \kl{\PDL}.}
    \label{figure: PDL axioms}
\end{figure}

We first prepare the following "derivable" "rules":
\begin{align*}
    &\begin{prooftree}
        \hypo{\fml_1}
        \hypo{\dots}
        \hypo{\fml_n}
        \infer3{\fml[2]}
    \end{prooftree} \qquad
    &\text{\footnotesize where}\quad 
    \begin{aligned}
    &\text{\footnotesize $\fml_1 \to \fml_2 \to \dots \to \fml_n \to \fml[2]$}\\
    &\text{\footnotesize is obtained from \text{\nameref{equation: PDL PROP}}}
    \end{aligned}
    \tag{$\mathrm{MP}^*\mathrm{Prop}$}\label{rule: MP PROP}\\
    &\begin{prooftree}
        \hypo{\fml_1 \to \dots \to \fml_n \to \fml[2]}
        \infer1{\bo{\term} \fml_1 \to \dots  \to \bo{\term}\fml_n \to \bo{\term} \fml[2]}
    \end{prooftree}
    \tag{$\mathrm{Mon}$}\label{rule: Mon}
\end{align*}
\eqref{rule: MP PROP} can be derived from \eqref{rule: MP}\text{\nameref{equation: PDL PROP}},
and \eqref{rule: Mon} can be derived from \eqref{rule: MP PROP}\eqref{equation: K}\eqref{rule: NEC},
respectively, by standard arguments in normal modal logic.

The following congruence rules are also "derivable", by a standard argument in normal modal logic (see, "eg", \cite[Example 3.49 \& 3.50, \& p.~89]{chagrovModalLogic1997}):
\begin{gather*}
    \begin{prooftree}[small]
        \hypo{\fml[1]_1 \leftrightarrow \fml[2]_1}
        \hypo{\fml[1]_2 \leftrightarrow \fml[2]_2}
        \infer2{(\fml[1]_1 \to \fml[1]_2) \leftrightarrow (\fml[2]_1 \to \fml[2]_2)}
    \end{prooftree}
    \qquad\qquad
    \begin{prooftree}[small]
        \hypo{\fml[1] \leftrightarrow \fml[2]}
        \infer1{\bo{\term} \fml[1] \leftrightarrow \bo{\term} \fml[2]}
    \end{prooftree}
    \tag{Cong-f} \label{rule: PDL cong-f}\\
    \bo{\term}(\fml[1] \land \fml[2]) \leftrightarrow (\bo{\term}\fml[1] \land \bo{\term} \fml[2]) \tag{K$\land$} \label{equation: PDL K'}
\end{gather*}

\subsection{Completeness}
Here, we give a proof using the completeness result of \cite{harelDynamicLogic2000} (PDL with "Kleene star" and "rich tests");
we denote by \AP""\HKTPDL"" for distinction.
The syntax is presented as follows ("ie", $\bl^{*}$ is primitive instead of $\bl^{+}$).
\AP
\phantomintro(HKTPDL){formulas}%
\phantomintro(HKTPDL){terms}%
\begin{align*}
    \fml[1], \fml[2], \fml[3] &\;\Coloneqq\; \afml \mid \fml[1] \to \fml[2] \mid \falsec \mid \bo{\term[1]} \fml \tag{\reintro*\kl(HKTPDL){formulas}}\\
    \term[1], \term[2], \term[3]  &\;\Coloneqq\; \aterm
    \mid \term[2] \compo \term[3]
    \mid \term[2] \union \term[3]
    \mid \term[2]^{*} \mid \fml[1]?
    \tag{\reintro*\kl(HKTPDL){terms}}
\end{align*}
We use the same abbreviations:
$\lnot \fml \defeq \fml \to \falsec$,
$\fml[1] \lor \fml[2] \defeq \lnot \fml[1] \to \fml[2]$,
$\fml[1] \land \fml[2] \defeq \lnot ((\lnot \fml[1]) \lor (\lnot \fml[2]))$,
$(\fml[1] \leftrightarrow \fml[2]) \defeq (\fml[1] \to \fml[2]) \land (\fml[2] \to \fml[1])$,
$\truec \defeq \lnot \falsec$,
$\dia{\term[1]}\fml[1] \defeq \lnot \bo{\term[1]} \lnot \fml[1]$.
Additionally, we let $\term^{+} \defeq \term \compo \term^{*}$.

\AP
We write $\intro*\vdashHKTPDL \fml$ if $\fml$ is derivable in the system of \cite[Section 5.5]{harelDynamicLogic2000} (\Cref{figure: HKTPDL axioms}).
\begin{remark}
    Precisely, in \Cref{figure: HKTPDL axioms}, compared to \cite[Section 5.5]{harelDynamicLogic2000},
    we omit the following axiom from the system in \cite[Section 5.5]{harelDynamicLogic2000}:
    $\bo{\term}(\fml[1] \land \fml[2]) \leftrightarrow (\bo{\term}\fml[1] \land \bo{\term} \fml[2])$.
    This is "derivable" from \eqref{equation:' K}\eqref{rule: NEC},
    by a standard argument in normal modal logic (see, "eg", \cite[Example 3.50]{chagrovModalLogic1997}).
    Also, we introduce \nameref{equation:' PDL PROP} instead of the following three "axioms" in propositional logic:
    \begin{gather}
        (\fml[1] \to \fml[2] \to \fml[3]) \to (\fml[1] \to \fml[2]) \to \fml[1] \to \fml[3]
        \tag{S-com} \label{equation: S combinator}\\
        \fml[1] \to \fml[2] \to \fml[1]
        \tag{K-com} \label{equation: distribution combinator}\\
        \lnot \lnot \fml[1] \to \fml[1]
        \tag{DN} \label{equation: DN}
    \end{gather}
    based on the (well-known) completeness theorem of propositional logic.
\end{remark}

\begin{figure}
    \begin{tcolorbox}[colback=black!3, top = .3ex, bottom = .3ex, left = .3em, right = .3em]
    \textbf{Rules}:
    \eqref{rule: MP} and \eqref{rule: NEC}

    \tcbline%
    \textbf{Axioms}:
    
    \proofcase{Prop. axioms} \hfill All \kl{substitution-instances} of valid propositional formulas \hfill \labeltext{(Prop)}{equation:' PDL PROP}

    \proofcase{Normal modal logic axiom and Segerberg's axioms}

    \vspace{-2ex}
    \hspace{-2em}
    \begin{minipage}[b]{0.45\columnwidth}
        \begin{gather}
            \bo{\term}(\fml[1] \to \fml[2]) \to (\bo{\term}\fml[1] \to \bo{\term} \fml[2])
            \tag{K} \label{equation:' K}\\
            \bo{\term[1] \union \term[2]} \fml[1] ~\leftrightarrow~ \bo{\term[1]}\fml[1] \land \bo{\term[2]}\fml[2]
            \tag{$\union$} \label{equation:' union}\\
            \bo{\term[1] \term[2]} \fml[1] ~\leftrightarrow~ \bo{\term[1]}\bo{\term[2]}\fml[1]
            \tag{$\compo$} \label{equation:' compo}
        \end{gather}
        \end{minipage}
        \hfill
        \begin{minipage}[b]{0.55\columnwidth}
        \begin{gather}
            \bo{\fml[2]?} \fml[1] ~\leftrightarrow~ (\fml[2] \to \fml[1]) \tag{Test} \label{equation:' test}\\
            \bo{\term[1]^*}\fml[1] ~\leftrightarrow~ (\fml[1] \land \bo{\term}\bo{\term[1]^*} \fml[1]) \tag{$\bl^{*}$} \label{equation:' *}\\
            (\fml[1] \land \bo{\term[1]^*}(\fml[1] \to \bo{\term[1]}\fml[1])) \to \bo{\term[1]^*}\fml[1] \tag{$\bl^{*}$-Ind} \label{equation:' ind}
        \end{gather}
    \end{minipage}
    \end{tcolorbox}
    \caption{$\HilbertstyleHKTPDL$: Rules and axioms for \kl{\HKTPDL} (from \cite[Section 5.5]{harelDynamicLogic2000}).}
    \label{figure: HKTPDL axioms}
\end{figure}

We then have the following completeness theorem.
\begin{proposition}[Completeness theorem for $\vdashPDL$]\label{proposition: PDL and HKTPDL}
    For all $\fml$, we have:
    \[\vdashPDL \fml \quad\iff\quad \modelsfml \fml.\]
\end{proposition}
\begin{proof}
    \proofcase{($\Longrightarrow$)} Easy.
    \proofcase{($\Longleftarrow$)}
    Relying on the completeness theorem of "\HKTPDL" \cite[Theorem 7.6]{harelDynamicLogic2000},
    we prove ${} \vdashHKTPDL \fml \Longrightarrow {} \vdashPDL \fml$ by induction on derivation trees.
    Because \eqref{rule: MP}\nameref{equation:' PDL PROP}\eqref{equation:' K}\eqref{equation:' union}\eqref{equation:' compo}\eqref{equation:' test}
    are also equipped in $\HilbertstylePDL$,
    the crucial cases are for \eqref{equation:' *} and \eqref{equation:' ind}.
    
    \proofcase{Case \eqref{equation:' *}}
    Using \eqref{rule: PDL cong-f}, first we have:
    \begin{align}
        \bo{\term^*}\fml
        \;=\; \bo{\truec? \union \term^+}\fml
        \;\leftrightarrow_{\eqref{equation: PDL union}\eqref{equation: PDL test}\text{\nameref{equation: PROP}}}\;
        \fml \land \bo{\term^+}\fml. 
        \tag{$\bl^{*}\shortminus\bl^{+}$}
        \label{equation: PDL * to +} 
    \end{align}
    Using \eqref{rule: PDL cong-f}, we thus have:
    \begin{align*}
        \bo{\term^*}\fml
        &\;\leftrightarrow_{\eqref{equation: PDL * to +}}\; \fml \land \bo{\term^+}\fml
        \;\leftrightarrow_{\eqref{equation: PDL *}}\; \fml \land (\bo{\term}\fml \land \bo{\term}\bo{\term^+}\fml) \\
        &\;\leftrightarrow_{\text{\eqref{equation: PDL K'}}}\; \fml \land \bo{\term}(\fml \land \bo{\term^+}\fml)
        \;\leftrightarrow_{\eqref{equation: PDL * to +}}\; \fml \land \bo{\term}\bo{\term^*}\fml.
    \end{align*}
    \proofcase{Case \eqref{equation:' ind}}
    By \eqref{equation: PDL ind},
    we have $\vdashPDL (\bo{\term}\fml \land \bo{\term^+}(\fml \to \bo{\term}\fml)) \to \bo{\term^+}\fml$,
    By \nameref{equation: PDL PROP},
    we have
    $\vdashPDL (\fml \land ((\fml \to \bo{\term}\fml) \land \bo{\term^+}(\fml \to \bo{\term}\fml))) \to
    (\fml \land \bo{\term^+}\fml)$.
    By \eqref{rule: PDL cong-f}\eqref{equation: PDL * to +},
    we have
    $\vdashPDL (\fml \land \bo{\term^*}(\fml \to \bo{\term}\fml)) \to \bo{\term^*}\fml$.
\end{proof}

\subsection{Equivalence of \eqref{equation: ind}, \eqref{rule: LI}, and \eqref{rule: TC}}\label{section: equivalence of ind, LI, and TC}
Additionally, 
we also introduce the loop invariance rule \eqref{rule: LI} and the transitive closure rule \eqref{rule: TC}, "cf" {\cite[Theorem 5.18]{harelDynamicLogic2000}}, where we use $\bl^{+}$ instead of $\bl^{*}$.
\begin{center}\vspace{-4ex}
\begin{minipage}[t]{.22\textwidth}
\begin{align*}
\begin{prooftree}
    \hypo{\fml \to \bo{\term}\fml}
    \infer1{\fml \to \bo{\term[1]^+}\fml}
\end{prooftree}  
\tag{\textcolor{black}{LI}}\label{rule: LI}  
\end{align*}
\end{minipage}
\hspace{3em}
\begin{minipage}[t]{.37\textwidth}
\begin{align*}
\begin{prooftree}
    \hypo{\fml[1] \to \bo{\term}\fml[1]}
    \hypo{\fml[1] \to \bo{\term}\fml[2]}
    \infer2{\fml[1] \to \bo{\term[1]^+}\fml[2]}
\end{prooftree}  
\tag{\textcolor{black}{TC}}\label{rule: TC}  
\end{align*}
\end{minipage}
\end{center}
We then have that \eqref{equation: ind}, \eqref{rule: LI}, and \eqref{rule: TC} are equivalent, as follows.
\begin{proposition}["Cf", {\cite[Theorem 5.18]{harelDynamicLogic2000}}]\label{proposition: LI and TC}
In $\HilbertstylePDL$ without \eqref{equation: ind},
the following "axioms" and "rules" are inter"derivable":
\eqref{equation: ind}, \eqref{rule: LI}, and \eqref{rule: TC}.
\end{proposition}
\begin{proof}
\proofcase{\eqref{rule: LI} from \eqref{equation: ind}}
We have:
\begin{center}\scalebox{1}{
\begin{prooftree}[small]
    \hypo{\fml \to \bo{\term}\fml}
    \hypo{\mathstrut}
    \infer1[\eqref{equation: ind}]{(\bo{\term[1]}\fml[1] \land \bo{\term[1]^+}(\fml[1] \to \bo{\term[1]}\fml[1])) \to \bo{\term[1]^+}\fml[1]}
    \hypo{\fml \to \bo{\term}\fml}
    \infer1[\eqref{rule: NEC}]{\bo{\term^+}(\fml \to \bo{\term}\fml)}
    \infer3[\eqref{rule: MP PROP}]{\fml[1] \to \bo{\term^+} \fml}
\end{prooftree}}
\end{center}

\proofcase{\eqref{rule: TC} from \eqref{rule: LI}}
We have:
\begin{center}\scalebox{.95}{
\begin{prooftree}[small]
    \hypo{\fml[1] \to \bo{\term}\fml[1]}
    \hypo{\fml[1] \to \bo{\term}\fml[2]}
    \hypo{\fml[1] \to \bo{\term}\fml[1]}
    \hypo{\fml[1] \to \bo{\term}\fml[2]}
    \hypo{\mathstrut}
    \infer1[\text{\nameref{equation: PROP}}]{\fml[1] \to \fml[2] \to (\fml[1] \land \fml[2])}
    \infer1[\eqref{rule: Mon}]{\bo{\term}\fml[1] \to \bo{\term}\fml[2] \to \bo{\term}(\fml[1] \land \fml[2])}
    \infer3[\eqref{rule: MP PROP}]{(\fml[1] \land \fml[2]) \to \bo{\term} (\fml[1] \land \fml[2])}
    \infer1[\eqref{rule: LI}]{(\fml[1] \land \fml[2]) \to \bo{\term^{+}}(\fml[1] \land \fml[2])}
    \infer1[\eqref{rule: MP PROP}]{\fml[1] \to \fml[2] \to \bo{\term^{+}}(\fml[1] \land \fml[2])}
    \infer1[\eqref{rule: Mon}]{\bo{\term}\fml[1] \to \bo{\term}\fml[2] \to \bo{\term}\bo{\term^{+}}(\fml[1] \land \fml[2])}
    \infer3[\eqref{rule: MP PROP}]{\fml[1] \to \bo{\term}\bo{\term^{+}}(\fml[1] \land \fml[2])}
\end{prooftree}}
\end{center}
By \text{\nameref{equation: PROP}} with \eqref{rule: Mon},
we also have
$\bo{\term}\bo{\term^{+}}(\fml[1] \land \fml[2]) \to \bo{\term}\bo{\term^{+}}\fml[2]$.
Combining these two "formulas@@PDL", $\fml[1] \to \bo{\term}\fml[2]$ (by assumption),
and $\bo{\term^{+}}\fml[2] \leftrightarrow \bo{\term}\bo{\term^{+}}\fml[2] \land \bo{\term}\fml[2]$ (by \eqref{equation: *}) using \eqref{rule: MP PROP} yields $\fml[1] \to \bo{\term^{+}} \fml[2]$.

\proofcase{\eqref{equation: ind} from \eqref{rule: TC}}
We first have:
\begin{center}\scalebox{.85}{
\begin{prooftree}[small]
    \hypo{\eqref{equation: *}}
    \hypo{\mathstrut}
    \infer1[\text{\nameref{equation: PROP}}]{\fml[1] \to
    (\fml[1] \to \bo{\term}\fml[1]) \to
    \bo{\term[1]^{+}}(\fml[1] \to \bo{\term}\fml[1]) \to
    (\bo{\term[1]}\fml[1] \land \bo{\term[1]^+}(\fml[1] \to \bo{\term[1]}\fml[1]))}
    \infer1[\eqref{rule: Mon}]{\bo{\term[1]}\fml[1] \to
    \bo{\term[1]}(\fml[1] \to \bo{\term}\fml[1]) \to
    \bo{\term[1]}\bo{\term[1]^{+}}(\fml[1] \to \bo{\term}\fml[1]) \to
    \bo{\term}(\bo{\term[1]}\fml[1] \land \bo{\term[1]^+}(\fml[1] \to \bo{\term[1]}\fml[1]))}
    \infer2[\eqref{rule: MP PROP}]{(\bo{\term[1]}\fml[1] \land \bo{\term[1]^+}(\fml[1] \to \bo{\term[1]}\fml[1])) \to \bo{\term}(\bo{\term[1]}\fml[1] \land \bo{\term[1]^+}(\fml[1] \to \bo{\term[1]}\fml[1]))}
\end{prooftree}
}
\end{center}
We thus have:
\begin{center}\scalebox{.85}{
\begin{prooftree}[small]
    \hypo{(\bo{\term[1]}\fml[1] \land \bo{\term[1]^+}(\fml[1] \to \bo{\term[1]}\fml[1])) \to \bo{\term}(\bo{\term[1]}\fml[1] \land \bo{\term[1]^+}(\fml[1] \to \bo{\term[1]}\fml[1]))}
    \hypo{\mathstrut}
    \infer1[\text{\nameref{equation: PROP}}]{(\bo{\term[1]}\fml[1] \land \bo{\term[1]^+}(\fml[1] \to \bo{\term[1]}\fml[1]))\to \bo{\term[1]}\fml[1]}
    \infer2[\eqref{rule: TC}]{(\bo{\term[1]}\fml[1] \land \bo{\term[1]^+}(\fml[1] \to \bo{\term[1]}\fml[1])) \to \bo{\term[1]^+}\fml[1]}
\end{prooftree}
}
\end{center}
Hence, this completes the proof.
\end{proof}

In particular, the "rules" \eqref{rule: LI} and \eqref{rule: TC} are "derivable" in $\HilbertstylePDL$.

\subsection{Congruence rules}\label{section: congruence}
Beyond \eqref{rule: PDL cong-f},
the following congruence "rules" are also "derivable" in $\HilbertstylePDL$, where $\fml[2]$ in $\set{\bl}_{\fml[2]}$ ranges over any "formulas@@PDL":\footnote{We do not need infinitely many hypotheses, but just for fitting to the congruence rules.}
\begin{align*}
    \begin{aligned}
    &
    \begin{prooftree}[small]
        \hypo{\fml[1]_1 \;\leftrightarrow\; \fml[2]_1}
        \hypo{\fml[1]_2 \;\leftrightarrow\; \fml[2]_2}
        \infer2{\fml[1]_1 \to \fml[1]_2 \;\leftrightarrow\; \fml[2]_1 \to \fml[2]_2}
    \end{prooftree}
    \hspace{2.em}
    \begin{prooftree}[small]
        \hypo{\fml[1] \leftrightarrow \fml[2]}
        \infer1{\bo{\term} \fml[1] \leftrightarrow \bo{\term} \fml[2]}
    \end{prooftree}
    \hspace{2.em}
    \begin{prooftree}[small]
        \hypo{\fml[2] \leftrightarrow \fml[3]}
        \infer1{\bo{\fml[2]?}\fml \leftrightarrow \bo{\fml[3]?}\fml}
    \end{prooftree}    
    \\
    &
    \begin{prooftree}[small]
        \hypo{\set{\bo{\term[1]_1}\fml[2] \leftrightarrow \bo{\term[2]_1}\fml[2]}_{\fml[2]}}
        \hypo{\set{\bo{\term[1]_2}\fml[2] \leftrightarrow \bo{\term[2]_2}\fml[2]}_{\fml[2]}}
        \infer2{\bo{\term[1]_1 \star \term[1]_2}\fml \leftrightarrow \bo{\term[2]_1 \star \term[2]_2}\fml}
    \end{prooftree}
    \mbox{ $\star \in \set{\compo, \union}$}
    \hspace{2.em}
    \begin{prooftree}[small]
        \hypo{\set{\bo{\term[1]}\fml[2] \leftrightarrow \bo{\term[2]}\fml[2]}_{\fml[2]}}
        \infer1{\bo{\term[1]^{+}}\fml \leftrightarrow \bo{\term[2]^{+}}\fml}
    \end{prooftree}
    \end{aligned}
    \tag*{\eqref{rule: cong}}
\end{align*}

\begin{proposition}\label{proposition: PDL cong}
    \eqref{rule: cong} are "derivable" in $\HilbertstylePDL$.
\end{proposition}
\begin{proof}
    \proofcase{Case $\to$}
    By \eqref{rule: MP PROP}.
    \proofcase{Case $\bo{\term}$}
    By \eqref{rule: Mon}\eqref{rule: MP PROP}.
    \proofcase{Case $?$}
    By \eqref{equation: PDL test}\eqref{rule: MP PROP}.
    \proofcase{Case $\compo$}
    We have:
    \begin{center}
    \begin{prooftree}[small]
        \hypo{\bo{\term[1]_2} \fml \leftrightarrow \bo{\term[2]_2} \fml}
        \infer1[\eqref{rule: cong} for $\bo{\term[1]_1}$]{\bo{\term[1]_1}\bo{\term[1]_2} \fml \leftrightarrow \bo{\term[1]_1}\bo{\term[2]_2} \fml}
        \hypo{\bo{\term[1]_1}\bo{\term[2]_2} \fml \leftrightarrow \bo{\term[2]_1}\bo{\term[2]_2} \fml}
        \infer2[\eqref{rule: MP PROP}]{\bo{\term[1]_1}\bo{\term[1]_2} \fml \leftrightarrow \bo{\term[2]_1}\bo{\term[2]_2} \fml}
    \end{prooftree}
    \end{center}
    Thus by \eqref{equation: compo}\eqref{rule: MP PROP}, this case has been shown.
    
    \proofcase{Case $\union$}
    By \eqref{equation: union}\eqref{rule: MP PROP}, this case has been shown.

    \proofcase{Case $\bl^{+}$}
    We have:
    \begin{center}
    \begin{prooftree}[small]
        \hypo{\eqref{equation: *}}
        \hypo{\bo{\term[1]}\bo{\term[1]^{+}} \fml[1] \leftrightarrow \bo{\term[2]} \bo{\term[1]^{+}} \fml[1]}
        \infer2[\eqref{rule: MP PROP}]{\bo{\term[1]^{+}} \fml[1] \to \bo{\term[2]} \bo{\term[1]^{+}} \fml[1]}
        \hypo{\eqref{equation: *}}
        \hypo{\bo{\term[1]}\fml[1] \leftrightarrow \bo{\term[2]}\fml[1]}
        \infer2[\eqref{rule: MP PROP}]{\bo{\term[1]^+}\fml[1] \to \bo{\term[2]}\fml[1]}
        \infer2[\eqref{rule: TC}]{\bo{\term[1]^{+}} \fml[1] \to \bo{\term[2]^{+}} \fml[1]}
    \end{prooftree}
    \end{center}
    Similarly, we have $\bo{\term[2]^{+}} \fml[1] \to \bo{\term[1]^{+}} \fml[1]$.
    Hence, this completes the proof.
\end{proof}
  
\clearpage
\begin{scope}\knowledgeimport{PDLREwLAp}
\section{Appendix to \Cref{section: PDLREwLAp}~``\nameref*{section: PDLREwLAp}''}

\subsection{Proof of \Cref{proposition: REwLAp to PDLREwLAp}}\label{section: proposition: REwLAp to PDLREwLAp}
The following is an analog of \cite[\S 5]{fischerPropositionalDynamicLogic1979}.
The proof is done in a standard way, by ranging the valuation of the "fresh" "formula variable" $\afml$ arbitrary.
\begin{proposition*}[Restatement of \Cref{proposition: REwLAp to PDLREwLAp}]
\REwLAptoPDLREwLAp
\end{proposition*}
\begin{proof}
    \proofcase{($\Longrightarrow$)}
    Since $\semREwLAp{\term[1]}{\struc} = \semPDLREwLAp{\term[2]}{\struc}$ implies
    $\semPDLREwLAp{\bo{\term[1]}\afml}{\struc} = \semPDLREwLAp{\bo{\term[2]}\afml}{\struc}$.
    \proofcase{($\Longleftarrow$)}
    Suppose $\struc \in \algclass$ and $\vertex[1], \vertex[2] \in \univ{\struc}$.
    Let $\struc[2]$ be the $\struc$ in which
    $\afml[1]^{\struc[2]} \defeq \univ{\struc[1]} \setminus \set{\vertex[2]}$
    (observe that $\struc[2] \in \algclass$ holds in either cases when $\algclass$ is $\GRELfinlin{}$ or $\GRELstfinlin{}$).
    By assumption with the construction of $\struc[2]$,
    $\tuple{\vertex[1], \vertex[2]} \in \semPDLREwLAp{\term[1]}{\struc[2]}$
    "iff"
    $\tuple{\vertex[1], \vertex[2]} \in \semPDLREwLAp{\term[2]}{\struc[2]}$.
    Since $\afml[1]$ does not appear in $\term[1]$ and $\term[2]$,
    $\tuple{\vertex[1], \vertex[2]} \in \semPDLREwLAp{\term[1]}{\struc[1]}$
    "iff"
    $\tuple{\vertex[1], \vertex[2]} \in \semPDLREwLAp{\term[2]}{\struc[1]}$.
    Hence, this completes the proof.
\end{proof}
 
\subsection{Proof of \Cref{proposition: PDLREwLAp to REwLAp}}\label{section: proposition: PDLREwLAp to REwLAp}
Let $\vsig'$ be a set disjoint from $\vsig$ and $\psig$ having the same "cardinality" as $\psig$ and 
let $f$ be a bijection from $\psig$ to $\vsig'$.
\begin{definition}\label{definition: PDLREwLAp to REwLAp}
    The translation function $\intro*\trPDLREwLAptoREwLAp \colon \exprclassPDLREwLAp{\psig, \vsig} \to \termclassREwLAp{\vsig \dcup \vsig'}$ is inductively defined as follows:

    \vspace{-3ex}
    \begin{align*}
        \trPDLREwLAptoREwLAp(\afml) &\defeq f(\afml)^{\dom}, &
        \trPDLREwLAptoREwLAp(\fml[2] \to \fml[3]) &\defeq \trPDLREwLAptoREwLAp(\fml[2])^{\adom} \union \trPDLREwLAptoREwLAp(\fml[3]), \\
        \trPDLREwLAptoREwLAp(\falsec) &\defeq \emp, &
        \trPDLREwLAptoREwLAp(\bo{\term[2]} \fml[2]) &\defeq (\trPDLREwLAptoREwLAp(\term[2]) \compo \trPDLREwLAptoREwLAp(\fml[2])^{\adom})^{\adom},
    \end{align*}

    \vspace{-6ex}
    \begin{minipage}[t]{0.45\textwidth}\vspace{0pt}
    \begin{align*}
        \trPDLREwLAptoREwLAp(\aterm) &\defeq \aterm, \\
        \trPDLREwLAptoREwLAp(\term[2] \compo \term[3]) &\defeq \trPDLREwLAptoREwLAp(\term[2]) \compo \trPDLREwLAptoREwLAp(\term[3]),\\
        \trPDLREwLAptoREwLAp(\term[2] \union \term[3]) &\defeq \trPDLREwLAptoREwLAp(\term[2]) \union \trPDLREwLAptoREwLAp(\term[3]), \\
        \trPDLREwLAptoREwLAp(\term[2]^{+}) &\defeq \trPDLREwLAptoREwLAp(\term[2])^{+}, 
    \end{align*}
    \end{minipage}
    \begin{minipage}[t]{0.35\textwidth}\vspace{0pt}
    \begin{align*}
        \trPDLREwLAptoREwLAp(\term[2]^{\adom}) &\defeq \trPDLREwLAptoREwLAp(\term[2])^{\adom}, \\
        \trPDLREwLAptoREwLAp(\fml[2]?) &\defeq \trPDLREwLAptoREwLAp(\fml[2]), \\
        \trPDLREwLAptoREwLAp(\term[2]^{\capid}) &\defeq \trPDLREwLAptoREwLAp(\term[2])^{\capid}, \\
        \trPDLREwLAptoREwLAp(\term[2]^{\capcomid}) &\defeq \trPDLREwLAptoREwLAp(\term[2])^{\capcomid}.
    \end{align*}
    \end{minipage}

    \lipicsEnd
\end{definition}

\begin{lemma}\label{lemma: PDLREwLA to REwLA}
    For all "expressions@@PDLREwLAp" $\expr \in \exprclassPDLREwLAp{\psig, \vsig}$
    and all $\struc \in \GRELpreorder{}$ "st" 
    $\struc \modelsfml f(\afml)^{\dom} = \afml?$ for all $\afml \in \psig$,
    we have:
    \begin{itemize}
        \item if $\expr = \fml$ is a "formula", then
        $\struc \modelsfml \trPDLREwLAptoREwLAp(\fml) = \fml?$,
        \item if $\expr = \term$ is a "term", then
        $\struc \modelsfml \trPDLREwLAptoREwLAp(\term) = \term$.
    \end{itemize}
\end{lemma}
\begin{proof}
    By easy induction on $\expr$.
    We only write for the following selected rules; the other cases are immediately shown by IH.

    \proofcase{Case $\fml = \fml[2] \to \fml[3]$}
    $\struc \modelsfml~
        \trPDLREwLAptoREwLAp(\fml[2] \to \fml[3])
        = \trPDLREwLAptoREwLAp(\fml[2])^{\adom} \union \trPDLREwLAptoREwLAp(\fml[3])
        =_{\text{IH}} (\fml[2]?)^{\adom} \union \fml[3]?
        = (\fml[2] \to \fml[3])?$.

    \proofcase{Case $\fml = \bo{\term[2]} \fml[2]$}
    $\struc \modelsfml~
        \trPDLREwLAptoREwLAp(\bo{\term[2]} \fml[2])
        = (\trPDLREwLAptoREwLAp(\term[2]) \compo \trPDLREwLAptoREwLAp(\fml[2])^{\adom})^{\adom}
        =_{\text{IH}} (\term[2] \compo \fml[2]?^{\adom})^{\adom}
        = \bo{\term[2]}\fml[2]?$. 
\end{proof}

\begin{lemma}\label{lemma: PDLREwLAp to REwLAp 2}
    For all $\fml \in \fmlclassPDLREwLAp{\psig, \vsig}$, we have:
    \[\GRELfinlin{} \modelsfml \trPDLREwLAptoREwLAp(\fml) = \id
    \quad\iff\quad \GRELfinlin{} \modelsfml \fml.\]
\end{lemma}
\begin{proof}
    We have:
    \begin{align*}
        &\GRELfinlin{} \modelsfml \trPDLREwLAptoREwLAp(\fml) = \id\\
        &\iff \struc \modelsfml \trPDLREwLAptoREwLAp(\fml) = \id
        \mbox{ for all $\struc \in \GRELfinlin{}$ "st" $\struc \modelsfml f(\afml)^{\dom} = \afml?$ for $\afml \in \psig$}\\
        \tag*{\footnotesize (For $\Longleftarrow$, since each $\afml \in \psig$ does not occur in $\trPDLREwLAptoREwLAp(\fml)$,}\\
        \tag*{\footnotesize we can reset $\afml^{\struc}$ so that $\struc \modelsfml f(\afml)^{\dom} = \afml?$ holds.)}\\
        &\iff \struc \modelsfml \fml? = \id
        \mbox{ for all $\struc \in \GRELfinlin{}$ "st" $\struc \modelsfml f(\afml)^{\dom} = \afml?$ for $\afml \in \psig$}
        \tag{\Cref{lemma: PDLREwLA to REwLA}} \\
        &\iff \GRELfinlin{} \modelsfml \fml? = \id
        \tag*{\footnotesize (For $\Longrightarrow$, since each $f(\afml) \in \vsig'$ do not occur in $\fml$,}\\
        \tag*{\footnotesize we can reset $f(\afml)^{\struc}$ so that $\struc \modelsfml f(\afml)^{\dom} = \afml?$ holds.)}\\
        &\iff \GRELfinlin{} \modelsfml \fml. \tag*{\qedhere}
    \end{align*}
\end{proof}

\begin{proposition*}[Restatement of \Cref{proposition: PDLREwLAp to REwLAp}]%
\PDLREwLAptoREwLAp
\end{proposition*}
\begin{proof}
    By \Cref{lemma: PDLREwLAp to REwLAp 2}.
\end{proof}

\subsection{Soundness of \Cref{theorem: PDL REwLA completeness}}\label{section: soundness}
In this section, we show the soundness of \Cref{theorem: PDL REwLA completeness}.

\begin{proof}[Proof of the soundness of \Cref{theorem: PDL REwLA completeness}]
We show that each "axiom" in \Cref{figure: PDLREwLAp axioms} is "valid" on $\GRELfinlin{}$.

\paragraph{Axioms "valid" on $\GRELpreorder{}$}
Below are "valid" on $\GRELpreorder{}$ ($\supseteq \GRELfinlin{}$).

\proofcase{For \text{\nameref{equation: PDL}}}
Every "valid" "\PDL" "formula@@PDL" $\fml$ satisfies $\REL{} \modelsfml \fml$.
As $\strucuniv$ does not affect to the evaluation of $\fml$, we have $\GRELpreorder{} \modelsfml \fml$.

As $\GRELpreorder{}$ corresponds a class of algebras of binary relations,
we observe that the following "substitution" "rule" holds on $\GRELpreorder{}$:
\begin{gather}
    \begin{prooftree}
    \hypo{\fml[1]}
    \infer1{\fml[1]\exprsubst{\Theta}}
    \end{prooftree}    
    \tag{Subst} \label{rule: subst}
\end{gather}
\AP
Here, for a "formula@@general" $\fml[1]$ and a "substitution" $\Theta$ mapping
each "term variable" to a "term@@general" and each "formula variable" to a "formula@@general",
we write $\fml[1]\intro*\exprsubst{\Theta}$ for
the $\fml[1]$ in which each "term variable" $\aterm$ has been replaced with $\Theta(\aterm)$
and each "formula variable" $\afml$ has been replaced with $\Theta(\afml)$.

\subproofcase{Proof of \eqref{rule: subst}}
For an $\struc \in \GRELpreorder{}$,
let $\struc\intro*\strucsubst{\Theta}$ denote the $\struc$ in which
$\aterm^{\struc\strucsubst{\Theta}} \defeq \semPDL{\Theta(\aterm)}{\struc}$ and 
$\afml^{\struc\strucsubst{\Theta}} \defeq \semPDL{\Theta(\afml)}{\struc}$.
By easy induction on $\expr$,
we have $\semPDL{\expr[1]\strucsubst{\Theta}}{\struc} = \semPDL{\expr[1]}{\struc\strucsubst{\Theta}}$.
Since $\struc\strucsubst{\Theta} \in \GRELpreorder{}$,
we have $\GRELpreorder{} \modelsfml \fml[1]\strucsubst{\Theta}$.
Hence, \text{\nameref{equation: PDL}} is "valid" on $\GRELpreorder{}$.

\proofcase{For \text{\nameref{equation: adom}}}
Because $\semPDLREwLAp{\term^{\adom}}{\struc} = \semPDLREwLAp{\bo{\term}\falsec?}{\struc}$ holds for all $\struc \in \GRELpreorder{}$.

\proofcase{For \eqref{equation: T capid}}
For all $\struc \in \GRELpreorder{}$,
we have:
$\tuple{x, x} \in \semPDLREwLAp{\term^{\capid}}{\struc}$ "iff"
$\tuple{x, x} \in \semPDLREwLAp{\dia{\term^{\capid}}\truec?}{\struc}$.
Thus, $\semPDLREwLAp{\term^{\capid}}{\struc} = \semPDLREwLAp{\dia{\term^{\capid}}\truec?}{\struc}$.

\proofcase{For \eqref{equation: capcomid-union-capid}}
By $\semPDLREwLAp{\term}{\struc} = \semPDLREwLAp{\term^{\capid} \union \term^{\capcomid}}{\struc}$ for all $\struc \in \GRELpreorder{}$.

\proofcase{For \eqref{equation: capid-union}}
By $\semPDLREwLAp{(\term[1] \union \term[2])^{\capid}}{\struc} = \semPDLREwLAp{\term[1]^{\capid} \union \term[2]^{\capid}}{\struc}$ for all $\struc \in \GRELpreorder{}$.

\proofcase{For \eqref{equation: capcomid-union}}
By $\semPDLREwLAp{(\term[1] \union \term[2])^{\capcomid}}{\struc} = \semPDLREwLAp{\term[1]^{\capcomid} \union \term[2]^{\capcomid}}{\struc}$ for all $\struc \in \GRELpreorder{}$.

\proofcase{For \eqref{equation: capid-adom}}
By $\semPDLREwLAp{(\term[1]^{\adom})^{\capid}}{\struc} = \semPDLREwLAp{\term[1]^{\adom}}{\struc}$ for all $\struc \in \GRELpreorder{}$.

\proofcase{For \eqref{equation: capcomid-adom}}
We have $\semPDLREwLAp{(\term[1]^{\adom})^{\capcomid}}{\struc} = \semPDLREwLAp{\falsec?}{\struc}$ for all $\struc \in \GRELpreorder{}$.
We thus have $\semPDLREwLAp{\bo{(\fml[2]?)^{\capcomid}}\fml}{\struc} = \semPDLREwLAp{\truec}{\struc}$ for all $\struc \in \GRELpreorder{}$.

\proofcase{For \eqref{equation: capid-test}}
By $\semPDLREwLAp{(\fml[2]?)^{\capid}}{\struc} = \semPDLREwLAp{\fml[2]?}{\struc}$ for all $\struc \in \GRELpreorder{}$.

\proofcase{For \eqref{equation: capcomid-test}}
We have $\semPDLREwLAp{(\fml[2]?)^{\capcomid}}{\struc} = \semPDLREwLAp{\falsec?}{\struc}$ for all $\struc \in \GRELpreorder{}$.
We thus have $\semPDLREwLAp{\bo{(\fml[2]?)^{\capcomid}}\fml}{\struc} = \semPDLREwLAp{\truec}{\struc}$ for all $\struc \in \GRELpreorder{}$.

\proofcase{For \eqref{equation: capid-capid}}
By $\semPDLREwLAp{(\term[1]^{\capid})^{\capid}}{\struc} = \semPDLREwLAp{\term[1]^{\capid}}{\struc}$ for all $\struc \in \GRELpreorder{}$.

\proofcase{For \eqref{equation: capcomid-capid}}
We have $\semPDLREwLAp{(\term[1]^{\capcomid})^{\capid}}{\struc} = \semPDLREwLAp{\falsec?}{\struc}$ for all $\struc \in \GRELpreorder{}$.
We thus have $\semPDLREwLAp{\bo{(\term[1]^{\capcomid})^{\capid}}\fml}{\struc} = \semPDLREwLAp{\truec}{\struc}$ for all $\struc \in \GRELpreorder{}$.

\proofcase{For \eqref{equation: capid-capcomid}}
We have $\semPDLREwLAp{(\term[1]^{\capid})^{\capcomid}}{\struc} = \semPDLREwLAp{\falsec?}{\struc}$ for all $\struc \in \GRELpreorder{}$.
We thus have $\semPDLREwLAp{\bo{(\term[1]^{\capid})^{\capcomid}}\fml}{\struc} = \semPDLREwLAp{\truec}{\struc}$ for all $\struc \in \GRELpreorder{}$.

\proofcase{For \eqref{equation: capcomid-capcomid}}
By $\semPDLREwLAp{(\term[1]^{\capcomid})^{\capcomid}}{\struc} = \semPDLREwLAp{\term[1]^{\capcomid}}{\struc}$ for all $\struc \in \GRELpreorder{}$.

\paragraph{Axioms "valid" on $\GRELpartialorder{}$}
We write $\AP\intro*\GRELpartialorder{}$ for the class of all $\struc \in \GREL {}$ "st" $\strucuniv^{\struc}$ is a "partial order".
Below "axioms" are "valid" on $\GRELpartialorder{}$ ($\supseteq \GRELfinlin{}$).
We show them using antisymmetry:
$\tuple{x, y}, \tuple{y, x} \in \strucuniv^{\struc} \Longrightarrow x = y$.

\proofcase{For \eqref{equation: capid-compo}}
We prove $\semPDLREwLAp{(\term[1] \term[2])^{\capid}}{\struc} = \semPDLREwLAp{\term[1]^{\capid} \term[2]^{\capid}}{\struc}$.

\subproofcase{($\supseteq$)}
Suppose $\tuple{x, x} \in \semPDLREwLAp{\term[1]^{\capid} \term[2]^{\capid}}{\struc}$.
We then have $\tuple{x, x} \in \semPDLREwLAp{\term[1]}{\struc}$ and
$\tuple{x, x} \in \semPDLREwLAp{\term[2]}{\struc}$.
Hence, we have $\tuple{x, x} \in \semPDLREwLAp{(\term[1] \term[2])^{\capid}}{\struc}$.

\subproofcase{($\subseteq$)}
Suppose $\tuple{x, x} \in \semPDLREwLAp{(\term[1] \term[2])^{\capid}}{\struc}$.
Then there is some $y$ such that
$\tuple{x, y} \in \semPDLREwLAp{\term[1]}{\struc}$ and
$\tuple{y, x} \in \semPDLREwLAp{\term[2]}{\struc}$.
As $\struc \in \GRELpartialorder{}$, we have $x = y$. 
Thus, $\tuple{x, x} \in \semPDLREwLAp{\term[1]^{\capid} \term[2]^{\capid}}{\struc}$.

\proofcase{For \eqref{equation: capcomid-compo}}
We prove $\semPDLREwLAp{(\term[1] \term[2])^{\capcomid}}{\struc} = \semPDLREwLAp{\term[1]^{\capcomid} \term[2]^{\capid}
\union \term[1]^{\capid} \term[2]^{\capcomid}
\union \term[1]^{\capcomid} \term[2]^{\capcomid}}{\struc}$.

\subproofcase{($\supseteq$)}
Suppose $\tuple{x, y} \in \semPDLREwLAp{\term[1]^{\capcomid} \term[2]^{\capid}}{\struc}$.
We then have $\tuple{x, y} \in \semPDLREwLAp{\term[1]}{\struc}$,
$\tuple{y, y} \in \semPDLREwLAp{\term[2]}{\struc}$, and $x \neq y$.
Hence, we have $\tuple{x, y} \in \semPDLREwLAp{(\term[1] \term[2])^{\capcomid}}{\struc}$.
Similarly for the case when $\tuple{x, y} \in \semPDLREwLAp{\term[1]^{\capid} \term[2]^{\capcomid}}{\struc}$.
Suppose $\tuple{x, y} \in \semPDLREwLAp{\term[1]^{\capcomid} \term[2]^{\capcomid}}{\struc}$.
Then there is some $z$ such that
$\tuple{x, z} \in \semPDLREwLAp{\term[1]}{\struc}$,
$\tuple{z, y} \in \semPDLREwLAp{\term[2]}{\struc}$,
$x \neq z$, and $z \neq y$.
We then have $x \neq y$ (if not, as $\struc \in \GRELpartialorder{}$, $x = z = y$, thus reaching a contradiction).
Hence, we have $\tuple{x, y} \in \semPDLREwLAp{(\term[1] \term[2])^{\capcomid}}{\struc}$.

\subproofcase{($\subseteq$)}
Suppose $\tuple{x, y} \in \semPDLREwLAp{(\term[1] \term[2])^{\capcomid}}{\struc}$.
Then $x \neq y$ and there is some $z$ such that
$\tuple{x, z} \in \semPDLREwLAp{\term[1]}{\struc}$ and
$\tuple{z, y} \in \semPDLREwLAp{\term[2]}{\struc}$.
As $\struc \in \GRELpartialorder{}$, 
we have either (i) $x = z$ and $z \neq y$, (ii) $x \neq z$ and $z = y$, or (iii) $x \neq z$ and $z \neq y$. 
Hence, we have $\tuple{x, y} \in \semPDLREwLAp{\term[1]^{\capcomid} \term[2]^{\capid}
\union \term[1]^{\capid} \term[2]^{\capcomid}
\union \term[1]^{\capcomid} \term[2]^{\capcomid}}{\struc}$.

\proofcase{For \eqref{equation: capid-*}}
We prove $\semPDLREwLAp{(\term[1]^{+})^{\capid}}{\struc} = \semPDLREwLAp{\term[1]^{\capid}}{\struc}$.

\subproofcase{($\supseteq$)}
Suppose $\tuple{x, x} \in \semPDLREwLAp{\term[1]^{\capid}}{\struc}$.
Then $\tuple{x, x} \in \semPDLREwLAp{(\term[1]^{+})^{\capid}}{\struc}$ clearly holds.

\subproofcase{($\subseteq$)}
Suppose $\tuple{x, x} \in \semPDLREwLAp{(\term[1]^{+})^{\capid}}{\struc}$.
Then there are $n \ge 0$ and $x_0, \dots, x_{n}$ with $x_0 = x_n = x$ such that
$\tuple{x_{i-1}, x_i} \in \semPDLREwLAp{\term[1]}{\struc}$ for all $i \in \rangeone{n}$.
As $\struc \in \GRELpartialorder{}$, we have $x = x_i$ for all $i \in \range{0}{n}$.
Hence, we have $\tuple{x, x} \in \semPDLREwLAp{\term[1]^{\capid}}{\struc}$.

\proofcase{For \eqref{equation: capcomid-*}}
We prove $\semPDLREwLAp{(\term[1]^{+})^{\capcomid}}{\struc} = \semPDLREwLAp{(\term[1]^{\capcomid})^{+}}{\struc}$.

\subproofcase{($\supseteq$)}
Suppose $\tuple{x, y} \in \semPDLREwLAp{(\term[1]^{\capcomid})^{+}}{\struc}$.
Then there are $n \ge 1$ and $x_0, \dots, x_{n}$ with $x_0 = x$ and $x_n = y$ such that
$\tuple{x_{i-1}, x_i} \in \semPDLREwLAp{\term[1]}{\struc}$ and $x_{i-1} \neq x_{i}$ for all $i \in \rangeone{n}$.
As $\struc \in \GRELpartialorder{}$, we have $x \neq y$.
Hence, we have $\tuple{x, y} \in \semPDLREwLAp{(\term[1]^{+})^{\capcomid}}{\struc}$.

\subproofcase{($\subseteq$)}
Suppose $\tuple{x, y} \in \semPDLREwLAp{(\term[1]^{+})^{\capcomid}}{\struc}$.
Then $x \neq y$ and there are $n \ge 1$ and $x_0, \dots, x_{n}$ with $x_0 = x$ and $x_n = y$ such that
$\tuple{x_{i-1}, x_i} \in \semPDLREwLAp{\term[1]}{\struc}$ for all $i \in \rangeone{n}$.
By removing each $i$ such that $x_{i-1} = x_i$ from the sequence $x_0, \dots, x_{n}$,
"wlog", we can assume that $x_{i-1} \neq x_{i}$ for all $i \in \rangeone{n}$ (then, we still have $n \ge 1$ by $x \neq y$).
Hence, we have $\tuple{x, y} \in \semPDLREwLAp{(\term[1]^{\capcomid})^{+}}{\struc}$.

\paragraph{Axioms "valid" on $\GRELfinpartialorder{}$}
We write $\AP\intro*\GRELfinpartialorder{}$ for the class of all $\struc \in \GREL {}$ "st" $\strucuniv^{\struc}$ is a finite "partial order".
Below "axiom" is the "valid" on $\GRELfinpartialorder{}$ ($\supseteq \GRELfinlin{}$).
We show it using well-foundedness:
there is no infinite sequence $x_0, x_1, \dots$ such that
for all $i \ge 0$, $\tuple{x_i, x_{i+1}} \in \strucuniv^{\struc} \setminus \diagonal_{\univ{\struc}}$.

\proofcase{For \text{\nameref{equation: Lob'}}}
It suffices to show the following contrapositive:
$\GRELfinlin{} \modelsfml \dia{(\term^{\capcomid})^{+}}\fml \to \dia{(\term^{\capcomid})^{+}}(\fml \land \lnot \dia{(\term^{\capcomid})^{+}}\fml)$.
Suppose $x \in \semPDLREwLAp{\dia{(\term^{\capcomid})^{+}}\fml}{\struc}$.
Then there is some $y$ such that
$\tuple{x, y} \in \semPDLREwLAp{(\term^{\capcomid})^{+}}{\struc}$ and $y \in \semPDLREwLAp{\fml}{\struc}$.
As $\struc \in \GRELfinpartialorder{}$,
there is some maximal element $y_M$ on the "partial order" $\strucuniv^{\struc}$
among elements $y$ satisfying the above.
We then have $y_M \in \semPDLREwLAp{\lnot \dia{(\term^{\capcomid})^{+}} \fml}{\struc}$.
Hence, $x \in \semPDLREwLAp{\dia{(\term^{\capcomid})^{+}}(\fml \land \lnot \dia{(\term^{\capcomid})^{+}}\fml)}{\struc}$.
\end{proof}

\begin{remark}
    For "\PDLREwLAp",
    we can observe that 
    the "theory" of $\GRELfinlin{}$ coincides with "that@equational theory" of $\GRELfinpartialorder{}$,
    via an argument like "tree unwinding" ("cf", \Cref{section: theorem: complexity PDLREwLAp substitution-closed: removing identity,section: theorem: complexity PDLREwLAp substitution-closed: to tree}, later).
    Hence, we do not need extra axioms for $\GRELfinlin{}$.
    \lipicsEnd
\end{remark}

\end{scope}
 
\clearpage
\begin{scope}\knowledgeimport{PDL-}
\section{Appendix to \Cref{section: PDL-}~``\nameref*{section: PDL-}''}
\begin{scope}\knowledgeimport{PDL-}
\subsection{Supplement of \Cref{footnote: PDL- and PDL}} \label{section: footnote: PDL- and PDL}
(This section is technically separated for the main body.)

We define the following transformation.
Intuitively, we decompose terms into the identity-part and the non-identity part in the level of ``traces''.
\begin{definition}\label{definition: normal form to PDL-}
    The function $\AP\intro*\PDLtonorm{\bl} \colon (\fmlclassPDL{} \to \fmlclassifreePDL{}) \dcup (\termclassPDL{} \to (\fmlclassifreePDL{} \times \termclassifreePDL{}))$,
    where $\tuple{\intro*\PDLtonormone{\term}, \intro*\PDLtonormtwo{\term}} \defeq \PDLtonorm{\term}$ for $\term \in \termclassPDL{}$,
    is inductively defined as follows:
    \begin{align*}
        \PDLtonorm{\afml} &\defeq \afml, &
        \PDLtonorm{(\fml[2] \to \fml[3])} &\defeq \PDLtonorm{\fml[2]} \to \PDLtonorm{\fml[3]}, \\
        \PDLtonorm{\falsec} &\defeq \falsec, &
        \PDLtonorm{(\bo{\term} \fml[2])} &\defeq (\PDLtonormone{\term} \to \PDLtonorm{\fml[2]}) \land \bo{\PDLtonormtwo{\term}} \PDLtonorm{\fml[2]},\\[.5ex]
        \PDLtonormone{\aterm} &\defeq \falsec, &
        \PDLtonormtwo{\aterm} &\defeq \aterm, \\
        \PDLtonormone{(\term[2] \compo \term[3])} &\defeq \PDLtonormone{\term[2]} \land \PDLtonormone{\term[3]}, &
        \PDLtonormtwo{(\term[2] \compo \term[3])} &\defeq
        \PDLtonormtwo{\term[2]} \PDLtonormone{\term[3]}? \union
        \PDLtonormone{\term[2]}? \PDLtonormtwo{\term[3]} \union
        \PDLtonormtwo{\term[2]} \PDLtonormtwo{\term[3]}, \\
        \PDLtonormone{(\term[2] \union \term[3])} &\defeq \PDLtonormone{\term[2]} \lor \PDLtonormone{\term[3]}, &
        \PDLtonormtwo{(\term[2] \union \term[3])} &\defeq \PDLtonormtwo{\term[2]} \union \PDLtonormtwo{\term[3]}, \\
        \PDLtonormone{(\term[2]^{+})} &\defeq \PDLtonormone{\term[2]}, &
        \PDLtonormtwo{(\term[2]^{+})} &\defeq (\PDLtonormtwo{\term[2]})^{+}, \\
        \PDLtonormone{(\fml[1]?)} &\defeq \PDLtonorm{\fml[1]}, &
        \PDLtonormtwo{(\fml[1]?)} &\defeq \emp. \tag*{\lipicsEnd}
    \end{align*}
\end{definition}

\begin{proposition}\label{proposition: normal form to PDL-}
    For every "expression" $\expr \in \exprclassPDL{}$,
    we have the following:
    \begin{enumerate}
        \item \label{proposition: normal form to PDL- 1} If $\expr = \fml$ is a "formula",
        $\GRELpreorder{} \modelsfml \PDLtonorm{\fml} \leftrightarrow \fml$.
        \item \label{proposition: normal form to PDL- 2} If $\expr = \term$ is a \kl{term},
        $\GRELpreorder{} \modelsfml \bo{\PDLtonormone{\term}? \union \PDLtonormtwo{\term}} \fml[3] \leftrightarrow \bo{\term} \fml[3]$, where $\fml[3]$ is any "formula".
    \end{enumerate}
\end{proposition}
\begin{proof}
    By induction on $\expr$.
    \proofcase{For \ref{proposition: normal form to PDL- 1}}
    Each case is immediate from IH.
    \proofcase{For \ref{proposition: normal form to PDL- 2}}
    We distinguish the following cases.
    Below, we only write for selected cases.
    The other cases are immediate from IH.
    
    \subproofcase{Case $\term = \term[2] \compo \term[3]$}
    We have:
    \begin{align*}
        \bo{\term[2] \compo \term[3]}\afml
        &~\leftrightarrow~ \bo{\term[2]} \bo{\term[3]}\afml
        ~\leftrightarrow_{\text{IH}}~ \bo{\term[2]} \bo{\PDLtonormone{\term[3]}? \union \PDLtonormtwo{\term[3]}} \afml
        ~\leftrightarrow_{\text{IH}}~ \bo{\PDLtonormone{\term[2]}? \union \PDLtonormtwo{\term[2]}} \bo{\PDLtonormone{\term[3]}? \union \PDLtonormtwo{\term[3]}} \afml\\
        &~\leftrightarrow~ \bo{\PDLtonormone{\term[2]}? \PDLtonormone{\term[3]}? \union \PDLtonormtwo{\term[2]}\PDLtonormone{\term[3]}? \union \PDLtonormone{\term[2]}?\PDLtonormtwo{\term[3]} \union \PDLtonormtwo{\term[2]} \PDLtonormtwo{\term[3]}} \afml\\
        &~\leftrightarrow~ \bo{\PDLtonormone{(\term[2] \compo \term[3])}? \union \PDLtonormtwo{(\term[2] \compo \term[3])}} \afml.
    \end{align*}

    \subproofcase{Case $\term = \term[2] \union \term[3]$}
    We have:
    \begin{align*}
        \bo{\term[2] \union \term[3]}\afml
        &~\leftrightarrow~ \bo{\term[2]}\afml \land \bo{\term[3]}\afml
        ~\leftrightarrow_{\text{IH}}~ \bo{\PDLtonormone{\term[2]}? \union \PDLtonormtwo{\term[2]}} \afml \land \bo{\PDLtonormone{\term[3]}? \union \PDLtonormtwo{\term[3]}} \afml\\
        &~\leftrightarrow~ \bo{\PDLtonormone{\term[2]}? \union \PDLtonormone{\term[3]}? \union \PDLtonormtwo{\term[2]} \union \PDLtonormtwo{\term[3]}} \afml
        ~\leftrightarrow~ \bo{\PDLtonormone{(\term[2] \union \term[3])}? \union \PDLtonormtwo{(\term[2] \union \term[3])}} \afml.
    \end{align*}
    Hence, this completes the proof.
\end{proof}

\begin{proposition}[{\Cref{footnote: PDL- and PDL}}]\label{proposition: PDL- and PDL}
    "\ifreePDL" and "\PDL" have the same \AP""expressive power"" on $\GRELpreorder{}$, "ie",
    for every "\PDL" "formula@@PDL" $\fml[1]$, there is a "\ifreePDL" "formula@@PDL-" $\fml[2]$ such that
    $\GRELpreorder{} \modelsfml \fml[1] \leftrightarrow \fml[2]$, and vice versa.
\end{proposition}
\begin{proof}
    \proofcase{{\ifreePDL} to {\PDL}}
    This is because {\ifreePDL} is a syntactic fragment of {\PDL}.
    \proofcase{{\PDL} to {\ifreePDL}}
    By \Cref{proposition: normal form to PDL-}.
\end{proof}

\end{scope}

\end{scope}
 
\clearpage
\section{Appendix to \Cref{section: reduction to identity-free}~``\nameref*{section: reduction to identity-free}''}

\subsection{Proof of \Cref{proposition: normal form}}\label{section: proposition: normal form}
\begin{proposition*}[Restatement of \Cref{proposition: normal form}]
    \propositionnormalform
\end{proposition*}
\begin{proof}
    (Below, we use \eqref{rule: cong} in derivations using the notations of $\leftrightarrow_{(\bl)}$.)

    \proofcase{For \ref{proposition: normal form 1}}
    Each case is straightforward.
    For the case of $\bo{\term} \fml[2]$, we have:
    \begin{align*}
        \tonorm{(\bo{\term} \fml[2])}
        &\;=\; (\tonormone{\term} \to \tonorm{\fml[2]}) \land \bo{\tonormtwo{\term}} \tonorm{\fml[2]}
        \;\leftrightarrow_{\text{\nameref{equation: PDL}}}\;
        \bo{\tonormone{\term}?} \tonorm{\fml[2]} \land \bo{\tonormtwo{\term}} \tonorm{\fml[2]}\\
        &\;\leftrightarrow_{\text{IH}}\; \bo{\term^{\capid}}
        \fml[2] \land \bo{\term^{\capcomid}} \fml[2]
        \;\leftrightarrow_{\text{\nameref{equation: PDL}}}\; \bo{\term^{\capid} \union \term^{\capcomid}} \fml[2]
        \;\leftrightarrow_{\eqref{equation: capcomid-union-capid}}\; \bo{\term} \fml[2].
    \end{align*}

    \proofcase{For \ref{proposition: normal form 2}}
    Each case is shown by the rules of \Cref{figure: PDLREwLAp axioms}.
    Below, we prove some selected cases.

    \subproofcase{Case $\term = \aterm$}
    We have:
    \begin{align*}
    \bo{\tonormone{\aterm}?}\fml[3]
    &~=~ \bo{\dia{\aterm^{\capid}}\truec?} \fml[3]
    ~\leftrightarrow_{\eqref{equation: T capid}}~ \bo{\aterm^{\capid}}\fml[3]
    \end{align*}

    \subproofcase{Case $\term = \term[2] \compo \term[3]$}
    We have:
    \begin{align*}
    &\bo{\tonormone{(\term[2] \compo \term[3])}?}\fml[3]
    ~=~ \bo{(\tonormone{\term[2]} \land \tonormone{\term[3]})?} \fml[3]
    ~\leftrightarrow_{\text{\nameref{equation: PDL}}}~ (\bo{\tonormone{\term[2]}?}\bo{\tonormone{\term[3]}?}\fml[3])\\
    &~\leftrightarrow_{\text{IH}}~
    (\bo{\term[2]^{\capid}}\bo{\term[3]^{\capid}}\fml[3])
    ~\leftrightarrow_{\text{\nameref{equation: PDL}}}~
    \bo{\term[2]^{\capid} \compo \term[3]^{\capid}}\fml[3]
    ~\leftrightarrow_{\eqref{equation: capid-compo}}~
    \bo{(\term[2] \compo \term[3])^{\capid}} \fml[3]. 
    \end{align*}

    \subproofcase{Case $\term = \term[2]^{\adom}$}
    We have:
    \begin{align*}
    &\bo{\tonormone{(\term[2]^{\adom})}?}\fml[3]
    ~=~ \bo{(\lnot \tonormone{\term[2]} \land \bo{\tonormtwo{\term[2]}}\falsec)?}\fml[3]
    ~\leftrightarrow_{\text{\nameref{equation: PDL}}}~
    (\bo{\tonormone{\term[2]}?}\falsec \land \bo{\tonormtwo{\term[2]}}\falsec) \to \fml[3]\\
    &~\leftrightarrow_{\text{IH}}~
    (\bo{\term[2]^{\capid}}\falsec \land \bo{\term[2]^{\capcomid}}\falsec) \to \fml[3] 
    ~\leftrightarrow_{\text{\nameref{equation: PDL}}}~ \bo{\term[2]^{\capid} \union \term[2]^{\capcomid}} \falsec \to \fml[3]\\
    &~\leftrightarrow_{\eqref{equation: capcomid-union-capid}}~ \bo{\term[2]} \falsec \to \fml[3]
    ~\leftrightarrow_{\text{\nameref{equation: PDL}}}~ \bo{\bo{\term[2]}\falsec?} \fml[3]
    ~\leftrightarrow_{\text{\nameref{equation: adom}}}~ \bo{\term[2]^{\adom}} \fml[3]
    ~\leftrightarrow_{\text{\eqref{equation: capid-adom}}}~ \bo{(\term[2]^{\adom})^{\capcomid}} \fml[3]
    \end{align*}
   
    \proofcase{For \ref{proposition: normal form 3}}
    Each case is shown by the rules of \Cref{figure: PDLREwLAp axioms}.
    We prove some selected cases below.
    
    \subproofcase{Case $\tonormtwo{(\term[2] \compo \term[3])}$}
    We have:
    \begin{align*}
        \bo{\tonormtwo{(\term[2] \compo \term[3])}} \fml[3]
        &~=~
        \bo{\tonormtwo{\term[2]} \compo \tonormone{\term[3]}?} \fml[3] \land
        \bo{\tonormone{\term[2]}? \compo \tonormtwo{\term[3]}} \fml[3] \land
        \bo{\tonormtwo{\term[2]} \compo \tonormtwo{\term[3]}} \fml[3]\\
        &~\leftrightarrow_{\text{IH with \eqref{rule: cong}}}~
        \bo{\term[2]^{\capcomid}\term[3]^{\capid} \union
        \term[2]^{\capid} \term[3]^{\capcomid} \union
        \term[2]^{\capcomid} \term[3]^{\capcomid}} \fml[3]
        ~\leftrightarrow_{\eqref{equation: capcomid-compo}}~ \bo{(\term[2] \compo \term[3])^{\capcomid}} \fml[3].
    \end{align*}

    \subproofcase{Case $\tonormtwo{(\term[2]^{+})}$}
    We have:
    \begin{align*}
        \bo{\tonormtwo{(\term[2]^{+})}} \fml[3]
        &~=~ \bo{(\tonormtwo{\term[2]})^{+}} \fml[3]
        ~\leftrightarrow_{\text{IH with \eqref{rule: cong}}}~ \bo{(\tonormtwo{(\term[2]^{\capcomid})})^{+}} \fml[3]
        ~\leftrightarrow_{\eqref{equation: capcomid-compo}}~ \bo{(\term[2]^{\capcomid})^{+}} \fml[3]. \tag*{\qedhere}
    \end{align*}
\end{proof}
 
\clearpage
\begin {scope}\knowledgeimport {PDL-}
\section{Appendix to \Cref{section: PDL- completeness}~``\nameref*{section: PDL- completeness}''}
We first introduce the following "derivable" "rules" in $\vdashifreePDLsfinlin$ ("cf", \Cref{section: footnote: PDL completeness}):
\begin{gather*}
    \begin{prooftree}
        \hypo{\fml_1}
        \hypo{\dots}
        \hypo{\fml_n}
        \infer3{\fml[2]}
    \end{prooftree} \qquad
    \text{\footnotesize where}\quad 
    \begin{aligned}
    &\text{\footnotesize $\fml_1 \to \fml_2 \to \dots \to \fml_n \to \fml[2]$}\\
    &\text{\footnotesize is obtained from \text{\nameref{equation: PROP}}}
    \end{aligned}
    \tag{MP*PROP} \label{rule: PDL- MP PROP}\\
    \begin{prooftree}
        \hypo{\fml_1 \to \dots \to \fml_n \to \fml[2]}
        \infer1{\bo{\term} \fml_1 \to \dots  \to \bo{\term}\fml_n \to \bo{\term} \fml[2]}
    \end{prooftree}
    \tag{Mon} \label{rule: PDL- Mon}\\
    \begin{prooftree}[small]
        \hypo{\fml[1]_1 \leftrightarrow \fml[2]_1}
        \hypo{\fml[1]_2 \leftrightarrow \fml[2]_2}
        \infer2{(\fml[1]_1 \to \fml[1]_2) \leftrightarrow (\fml[2]_1 \to \fml[2]_2)}
    \end{prooftree}
    \qquad\qquad
    \begin{prooftree}[small]
        \hypo{\fml[1] \leftrightarrow \fml[2]}
        \infer1{\bo{\term} \fml[1] \leftrightarrow \bo{\term} \fml[2]}
    \end{prooftree}
    \tag{Cong-f} \label{rule: PDL- cong-f}
\end{gather*}

\subsection{On upper bound of $\cl(\fml)$ (\Cref{definition: FL-closure})}\label{section: definition: FL-closure}
(This section follows from the standard arguments in FL-closure; see, "eg", \cite{harelDynamicLogic2000}.
Below, we explicitly write the proof,
for the case $\bo{\term[1]^{+}} \fml[2]$ instead of $\bo{\term[1]^{*}} \fml[2]$
and for the cases $\bo{\fml[3]? \compo \term[1]} \fml[2]$ and $\bo{\term[1] \compo \fml[3]?} \fml[2]$.)

\AP
We define $\intro*\clone(\fml)$ as the smallest set closed under the following rules:
\begin{align*}
    \clone(\fml) &~\ni~ \fml,&
    \clone(\fml[2] \to \fml[3]) &~\supseteq~ \clone(\fml[2]) \cup \clone(\fml[3]),\\
    \clone(\bo{\term[1]} \fml[2]) &~\supseteq~ \clone(\fml[2]),&
    \clone(\bo{\term[1] \union \term[2]} \fml[2]) &~\supseteq~ \clone(\bo{\term[1]}\fml[2]) \cup \clone(\bo{\term[2]}\fml[2]),\\
    \clone(\bo{\term[1]^{+}} \fml[2]) &~\supseteq~ \clone(\bo{\term[1]}\bo{\term[1]^{+}}\fml[2]),&
    \clone(\bo{\term[1] \compo \term[2]} \fml[2]) &~\supseteq~ \clone(\bo{\term[1]}\bo{\term[2]}\fml[2]),\\
    \clone(\bo{\fml[3]? \compo \term[1]} \fml[2]) &~\supseteq~ \clone(\fml[3] \to \bo{\term[1]} \fml[2]),&
    \clone(\bo{\term[1] \compo \fml[3]?} \fml[2]) &~\supseteq~ \clone(\bo{\term[1]} (\fml[3] \to \fml[2])).
\end{align*}

\begin{proposition}\label{proposition: clone and cl}
For every "\ifreePDL" "formula@@PDL-" $\fml$, we have $\clone(\fml) = \cl(\fml)$.
\end{proposition}
\begin{proof}
\proofcase{($\subseteq$)}
We show that $\cl(\fml)$ (\Cref{definition: FL-closure}) satisfies the rule set above.
We observe that if $\fml[2]_0 \in \cl(\fml[2]) \Longrightarrow \dots \Longrightarrow \fml[2]_n \in \cl(\fml[2])$ is a ""derivation path"" ("ie", a sequence that each $\fml[2]_{i} \in \cl(\fml[2]) \Longrightarrow \fml[2]_{i+1} \in \cl(\fml[2])$ matches some rule in \Cref{definition: FL-closure}),
then $\fml[2]_0 \in \cl(\fml[1]) \Longrightarrow \dots \Longrightarrow \fml[2]_n \in \cl(\fml[1])$ is also a "derivation path",
because both $\cl(\fml[1])$ and $\cl(\fml[2])$ have the same rule set (except the initial "rules" $\fml[1] \in \cl(\fml[1])$ and $\fml[2] \in \cl(\fml[2])$).
Thus, if $\fml[2] \in \cl(\fml[1])$, then $\fml[3] \in \cl(\fml[2]) \Longrightarrow \fml[3] \in \cl(\fml[1])$ holds by the same "derivation path" for all $\fml[3]$.
We thus have:
\[\fml[2] \in \cl(\fml[1]) ~\Longrightarrow~ \cl(\fml[2]) \subseteq \cl(\fml[1]).\]
For instance, by letting $\fml[1] = \bo{\term} \fml[2]$, we have the "rule" $\cl(\fml[2]) \subseteq \cl(\bo{\term} \fml[2])$.
For the other "rules", the proofs are similar.

\proofcase{($\supseteq$)}
If $\clone(\fml[1]) \ni \fml[2]$,
then there are $\fml[2]_n, \fml[2]_{n-1}, \dots, \fml[2]_0$ with $\fml[2]_n = \fml[1]$ and $\fml[2]_0 = \fml[2]$
such that $\clone(\fml[2]_n) \supseteq \dots \supseteq \clone(\fml[2]_0) \ni \fml[2]_0$ is a "derivation path" ("ie", each $\clone(\fml_{i+1}) \supseteq \clone(\fml_{i})$ matches some "rule" for $\clone$).
Thus, $\clone(\fml[1]) \ni \fml[2] \Longrightarrow \clone(\fml[1]) \supseteq \clone(\fml[2])$ (by the "derivation path").
For instance, by letting $\fml[2]$ be $\bo{\term} \fml[2]$
with $\clone(\bo{\term} \fml[2]) \supseteq \clone(\fml[2]) \ni \fml[2]$,
we have $\fml[2] \in \clone(\fml)$.
Hence, we have the "rule" $\bo{\term} \fml[2] \in \clone(\fml) \Longrightarrow \fml[2] \in \clone(\fml)$.
\end{proof}

\AP
We now inductively define $\intro*\cltwo(\fml)$ and $\intro*\clbo(\term, \fml)$ as follows:
\begin{align*}
    \cltwo(\afml) &\defeq \set{\afml},&
    \cltwo(\fml[2] \to \fml[3]) &\defeq \set{\fml[2] \to \fml[3]} \cup \cltwo(\fml[2]) \cup \cltwo(\fml[3]),\\
    \cltwo(\falsec) &\defeq \set{\falsec},&
    \cltwo(\bo{\term[1]} \fml[2]) &\defeq \cltwo(\fml[2]) \cup \clbo(\term[1], \fml[2]),\\
    \clbo(\aterm[1], \fml[2]) &\defeq \set{\bo{\aterm[1]} \fml[2]}, \span\span\\
    \clbo(\term[1] \union \term[2], \fml[2]) &\defeq \set{\bo{\term[1] \union \term[2]} \fml[2]} \cup \clbo(\term[1], \fml[2]) \cup \clbo(\term[2], \fml[2]), \span\span\\
    \clbo(\term[1]^{+}, \fml[2]) &\defeq \set{\bo{\term[1]^{+}} \fml[2]} \cup \clbo(\term[1], \bo{\term[1]^+} \fml[2]), \span\span\\
    \clbo(\term[1] \compo \term[2], \fml[2]) &\defeq \set{\bo{\term[1] \compo \term[2]} \fml[2] } \cup \clbo(\term[1], \bo{\term[2]} \fml[2]) \cup \clbo(\term[2], \fml[2]), \span\span\\
    \clbo(\fml[3]? \compo \term[1], \fml[2]) &\defeq \set{\bo{\fml[3]? \compo \term[1]}, \fml[3] \to \bo{\term[1]}\fml[2]} \cup \cltwo(\fml[3]) \cup \clbo(\term[1], \fml[2]), \span\span\\
    \clbo(\term[1] \compo \fml[3]?, \fml[2]) &\defeq \set{\bo{\term[1] \compo \fml[3]?}, \fml[3] \to \fml[2]} \cup \clbo(\term[1], \fml[3] \to \fml[2]) \cup \cltwo(\fml[3]). \span\span
\end{align*}

\begin{proposition}\label{proposition: clbo}
    For every "\ifreePDL" "expression@@PDL-" $\expr$,
    we have the following:
    \begin{itemize}
        \item If $\expr = \term$ is a "term@@PDL-", then $\clone(\bo{\term} \fml[2]) = \clone(\fml[2]) \cup \clbo(\term, \fml[2])$ for any $\fml[2]$.
        \item If $\expr = \fml$ is a "formula@@PDL-", then $\cltwo(\fml) = \clone(\fml)$.
    \end{itemize}
\end{proposition}
\begin{proof}
By induction on $\expr$.
Below, we write only some interesting cases.

\proofcase{Case $\term = \term[2]^{+}$}
$\clone(\bo{\term[2]^{+}}\fml[2])$ is the smallest set satisfying the following:
\begin{align*}
\clone(\bo{\term[2]^{+}}\fml[2])
&\supseteq \set{\bo{\term[2]^{+}}\fml[2]} \cup \clone(\fml[2]) \cup \uline{\clone(\bo{\term[2]}\bo{\term[2]^+}\fml[2])} \tag{By the rules of $\clone(\bo{\term[2]^{+}}\fml[2])$}\\
&= \set{\bo{\term[2]^{+}}\fml[2]} \cup \clone(\fml[2]) \cup \clone(\bo{\term[2]^+}\fml[2]) \cup \clbo(\term[2], \bo{\term[2]^+}\fml[2]). \tag{By IH}
\end{align*}
Thus,
$\clone(\bo{\term[2]^{+}}\fml[2]) = \set{\bo{\term[2]^{+}}\fml[2]} \cup \clone(\fml[2]) \cup \clbo(\term[2], \bo{\term[2]^+}\fml[2])$.
By $\clbo(\term[2]^{+}, \fml[2]) = \set{\bo{\term[2]^{+}}\fml[2]} \cup \clbo(\term[2], \bo{\term[2]^+}\fml[2])$ (by the rules of $\clbo(\term[2]^{+}, \fml[2])$),
this case holds.

\proofcase{Case $\term = \term[2] \compo \fml[3]?$}
We have:
\begin{align*}
\clone(\bo{\term[2] \compo \fml[3]?}\fml[2])
&= \set{\bo{\term[2] \compo \fml[3]?}\fml[2]} \cup \uline{\clone(\bo{\term[2]}(\fml[3] \to \fml[2]))} \cup \clone(\fml[2]) \tag{By the rules of $\clone(\bo{\term[2] \compo \fml[3]?}\fml[2])$}\\
&= \set{\bo{\term[2] \compo \fml[3]?}\fml[2]} \cup \uline{\clone(\fml[3] \to \fml[2])} \cup \clbo(\term[2], \fml[3] \to \fml[2]) \cup \clone(\fml[2]) \tag{By IH}\\
&= \set{\bo{\term[2] \compo \fml[3]?}\fml[2]} \cup \set{\fml[3] \to \fml[2]} \cup \uline{\clone(\fml[3])} \cup \clbo(\term[2], \fml[3] \to \fml[2]) \cup \clone(\fml[2]) \tag{By the rules of $\clone(\fml[3] \to \fml[2])$}\\
&= \set{\bo{\term[2] \compo \fml[3]?}\fml[2]} \cup \set{\fml[3] \to \fml[2]} \cup \cltwo(\fml[3]) \cup \clbo(\term[2], \fml[3] \to \fml[2]) \cup \clone(\fml[2]) \tag{By IH}\\
&= \clbo(\term[2] \compo \fml[3]?, \fml[2]) \cup \clone(\fml[2]). \tag*{(By the rules of $\clbo(\term[2] \compo \fml[3]?, \fml[2])$)\qedhere}
\end{align*}
\end{proof}
Thus, $\cl(\fml) = \cltwo(\fml)$.
Using this, it is easy to see $\card{\cl(\fml)} \le 2 \len{\fml}$.
\begin{proposition}\label{proposition: size of cl}
    For every "\ifreePDL" "expression@@PDL-" $\expr$,
    we have the following:
    \begin{itemize}
        \item If $\expr = \term$ is a "term@@PDL-", then $\card{\clbo(\term, \fml[2])} \leq 2 \cdot \len{\term}$ for any $\fml[2]$.
        \item If $\expr = \fml$ is a "formula@@PDL-", then $\card{\cltwo(\fml)} \leq 2 \cdot \len{\fml}$.
    \end{itemize}
\end{proposition}
\begin{proof}
    By easy induction on $\expr$.
\end{proof}

\subsection{Proof of \Cref{proposition: atoms consistent saturates}}\label{section: proposition: atoms consistent saturates}
(This section follows from the standard argument of normal modal logic.)

\begin{proposition*}[Restatement of \Cref{proposition: atoms consistent saturates}]
\propositionatomsaturates
\end{proposition*}
\begin{proof}
For each of them, its contrapositive is shown by \eqref{rule: MP}\eqref{equation: Ksub}.
\end{proof}

\subsection{Proof of \Cref{proposition: atoms cond}}\label{section: proposition: atoms cond}
(The most part of this section follows from the proof of the completeness, "cf" \cite[Lemma 2]{kozenElementaryProofCompleteness1981}.
Below, we explicitly write the proof,
for the case $\bo{\term[2]^{+}} \fml[2]$ instead of $\bo{\term[2]^{*}} \fml[2]$
and for the cases $\bo{\fml[3]? \compo \term[2]} \fml[2]$ and $\bo{\term[2] \compo \fml[3]?} \fml[2]$.)

\begin{proposition*}[Restatement of \Cref{proposition: atoms cond}]
    \propositionatomscondprop
\end{proposition*}
\begin{proof}
(Below, we use \eqref{rule: PDL- cong-f} in derivations using the notations of $\leftrightarrow_{(\bl)}$.)

\proofcase{For \ref{proposition: atoms cond to}}
We have:
$(\atomtof{\atom} \land (\fml[2] \to \fml[3]))
\leftrightarrow_{\text{\nameref{equation: PROP}}}
(\atomtof{\atom} \land \lnot \fml[2]) \lor (\atomtof{\atom} \land \fml[3])$.
Thus by $\fml[2] \to \fml[3] \in \cl(\fml_0)$, 
$\fml[2] \to \fml[3] \in \atom$ "iff"
$\fml[2] \not\in \atom$ or $\fml[3] \in \atom$.

\proofcase{For \ref{proposition: atoms cond false}}
If $\falsec \in \atom$, then $\atom$ is "inconsistent", reaching a contradiction.

\proofcase{For \ref{proposition: atoms cond union}}
We have:
$(\atomtof{\atom} \land \bo{\term[2] \union \term[3]} \fml[2])
\leftrightarrow_{\eqref{equation: union}}
(\atomtof{\atom} \land \bo{\term[2]} \fml[2] \land \bo{\term[3]} \fml[2])$.
Thus, $\bo{\term[2] \union \term[3]}\fml[2] \in \atom$ "iff"
$\bo{\term[2]}\fml[2] \in \atom$ and $\bo{\term[3]}\fml[2] \in \atom$.

\proofcase{For \ref{proposition: atoms cond compo}}
We have:
$(\atomtof{\atom} \land \bo{\term[2] \compo \term[3]} \fml[2])
\leftrightarrow_{\eqref{equation: compo}}
(\atomtof{\atom} \land \bo{\term[2]} \bo{\term[3]} \fml[2])$.
Thus, $\bo{\term[2] \term[3]}\fml[2] \in \atom$ "iff"
$\bo{\term[2]} \bo{\term[3]}\fml[2] \in \atom$.

\proofcase{For \ref{proposition: atoms cond *}}
We have:
$(\atomtof{\atom} \land \bo{\term[2]^{+}} \fml[2])
\leftrightarrow_{\eqref{equation: *}}
(\atomtof{\atom} \land \bo{\term[2]} \fml[2] \land \bo{\term[2]}\bo{\term[2]^{+}} \fml[2])$.
Thus, $\bo{\term[2]^{+}}\fml[2] \in \atom$ "iff"
$\bo{\term[2]}\fml[2] \in \atom$ and $\bo{\term[2]} \bo{\term[2]^{+}}\fml[2] \in \atom$.

\proofcase{For \ref{proposition: atoms cond test L}}
We have:
$(\atomtof{\atom} \land \bo{\fml[3]? \compo \term[2]} \fml[2])
\leftrightarrow_{\eqref{equation: test L}}
(\atomtof{\atom} \land (\fml[3] \to \bo{\term[2]} \fml[2]))$.
Thus, $\bo{\fml[3]? \compo \term[2]}\fml[2] \in \atom$ "iff"
$\fml[3] \to \bo{\term[2]} \fml[2] \in \atom$.

\proofcase{For \ref{proposition: atoms cond test R}}
We have:
$(\atomtof{\atom} \land \bo{\term[2] \compo \fml[3]?} \fml[2])
\leftrightarrow_{\eqref{equation: test R}}
(\atomtof{\atom} \land \bo{\term[2]}(\fml[3] \to \fml[2]))$.
Thus, $\bo{\term[2] \compo \fml[3]?} \fml[2] \in \atom$ "iff"
$\bo{\term[2]} (\fml[3] \to \fml[2]) \in \atom$.
\end{proof}

\subsection{Proof of \Cref{proposition: atoms consistent}}\label{section: proposition: atoms consistent}
The most part of this section follows from the proof of the completeness, "cf" \cite[Lemma 1]{kozenElementaryProofCompleteness1981}.
Below, we explicitly write the proof,
for the case $\bo{\term[2]^{+}} \fml[2]$ instead of $\bo{\term[2]^{*}} \fml[2]$
and for the cases $\bo{\fml[3]? \compo \term[2]} \fml[2]$ and $\bo{\term[2] \compo \fml[3]?} \fml[2]$.

\begin{proposition*}[Restatement of \Cref{proposition: atoms consistent}]
    \propositionatoms
\end{proposition*}

\begin{proof}
    Below are shown almost in the same manner as \cite[Lemma 1]{kozenElementaryProofCompleteness1981}.
    
    \proofcase{For \ref{proposition: atoms consistent union}}
    The contrapositive is shown by
    $(\lnot (\atomtof{\atom} \land \dia{\term[2]} \atomtof{\atom}')
    \land \lnot (\atomtof{\atom} \land \dia{\term[3]} \atomtof{\atom}'))
    \leftrightarrow_{\text{\nameref{equation: PROP}}}
    (\lnot \atomtof{\atom} \lor (\bo{\term[2]} \lnot \atomtof{\atom}' \land \bo{\term[3]} \lnot \atomtof{\atom}'))
    \leftrightarrow_{\eqref{equation: union}}
    (\lnot \atomtof{\atom} \lor \bo{\term[2] \union \term[3]} \lnot \atomtof{\atom}')
    \leftrightarrow_{\text{\nameref{equation: PROP}}}
    \lnot (\atomtof{\atom} \land \dia{\term[2] \union \term[3]} \atomtof{\atom}')
    $.
    
    \proofcase{For \ref{proposition: atoms consistent compo}}
    By \eqref{equation: compo},
    the "formula" $\atomtof{\atom} \land \dia{\term[2]}\dia{\term[3]} \atomtof{\atom}'$ is "consistent".
    By applying \ref{proposition: atoms consistent saturates}.\ref{proposition: atoms consistent saturate dia} iteratively, there is some "atom" $\atom''$ such that
    the "formula" $\atomtof{\atom} \land \dia{\term[2]}(\atomtof{\atom}'' \land \dia{\term[3]} \atomtof{\atom}')$ is "consistent".
    Hence, both $\atomtof{\atom} \land \dia{\term[2]} \atomtof{\atom}''$ and $\atomtof{\atom}'' \land \dia{\term[3]} \atomtof{\atom}'$ are "consistent" by \eqref{rule: NEC}\eqref{equation: Ksub}.

    \proofcase{For \ref{proposition: atoms consistent *}}
    Let $\atomset$ be the smallest set of "atoms" satisfying the following:
    if $\atom[2] \in \atomset \cup \set{\atom}$ and $\atomtof{\atom[2]} \land \dia{\term[2]}\atomtof{\atom[2]}'$ is "consistent",
    then $\atom[2]' \in \atomset$.
    By construction of $\atomset$, for each $\atom[2] \in \atomset \cup \set{\atom}$,
    the "formula" $\atomtof{\atom[2]} \land \dia{\term[2]} \lnot \bigvee \atomset$ is not "consistent".
    Thus,
    we have both $\vdashifreePDLsfinlin \atomtof{\atom} \to \bo{\term[2]} \bigvee \atomset$ and $\vdashifreePDLsfinlin \bigvee \atomset \to \bo{\term[2]} \bigvee \atomset$.
    By \eqref{rule: TC},
    we have
    $\vdashifreePDLsfinlin \atomtof{\atom} \to \bo{\term[2]^{+}} \bigvee \atomset$.
    If $\atom' \not\in \atomset$, then $\vdashifreePDLsfinlin \atomtof{\atom} \to \bo{\term[2]^{+}} \lnot \atomtof{\atom}'$ 
    (as $\vdashifreePDLsfinlin \bigvee \atomset \to \lnot \atomtof{\atom}'$ holds by \nameref{equation: PROP}),
    but contradicts to that $\atomtof{\atom} \land \dia{\term[2]^+} \atomtof{\atom}'$ is "consistent".
    Thus, $\atom' \in \atomset$.
    By taking the sequence of "atoms" from $\atom$ to $\atom'$ according to the definition of $\atomset$,
    this case has been shown.

    \proofcase{For \ref{proposition: atoms consistent test L}}
    By \eqref{equation: test L},
    the "formula" $\atomtof{\atom} \land \fml[3] \land \dia{\term[2]} \atomtof{\atom}'$ is "consistent".
    As $\fml[3] \in \cl(\fml_0)$, we have $\fml[3] \in \atom$.

    \proofcase{For \ref{proposition: atoms consistent test R}}
    By \eqref{equation: test R},
    the "formula" $\atomtof{\atom} \land \dia{\term[2]} (\fml[3] \land \atomtof{\atom}')$ is "consistent".
    As $\fml[3] \in \cl(\fml_0)$, we have $\fml[3] \in \atom'$.
\end{proof}

\subsection{Supplement of "L{\"o}b's rule" (in \Cref{lemma: add root})} \label{section: proposition: Lob}
(Below is well-known in the context of provability logic; see, "eg", \cite[p.~59]{boolosLogicProvability1994}.)

\noindent
The "L{\"o}b's rule"
\begin{center}
    \begin{prooftree}[small]
        \hypo{\bo{\term^{+}} \fml \to \fml}
        \infer1{\fml}
    \end{prooftree} 
\end{center}
is "derivable" from the "axiom" \eqref{equation: Lob}, as follows:
\begin{center}
\begin{prooftree}[small]
    \hypo{\bo{\term^{+}} \fml \to \fml}
    \hypo{\mathstrut}
    \infer1[\eqref{equation: Lob}]{\bo{\term^{+}}(\bo{\term^{+}} \fml \to \fml) \to \bo{\term^{+}} \fml}
    \hypo{\bo{\term^{+}} \fml \to \fml}
    \infer1[\eqref{rule: NEC}]{\bo{\term^{+}}(\bo{\term^{+}} \fml \to \fml)}
    \infer2[\eqref{rule: MP}]{\bo{\term^{+}} \fml}
    \infer2[\eqref{rule: MP}]{\fml}
\end{prooftree}        
\end{center}

\subsection{Supplement of \Cref{lemma: ordering atoms}} \label{section: proposition: all trace}
For a finite set $A = \set{\aterm_1, \dots, \aterm_n}$,
we use $A$ to denote the "term@@PDL-" $\sum_{\aterm \in A} \aterm$, for short.
In this section, we show the following:
\begin{proposition}\label{proposition: all trace}%
\gdef\propositionalltrace{%
For every "\ifreePDL" "term@@PDL-" $\term$ and finite set $A \supseteq \exprvsig(\term)$, we have
$\vdashifreePDLsfinlin \bo{A^{+}} \fml \to \bo{\term} \fml$.
}
\propositionalltrace
\end{proposition}

\begin{lemma}\label{lemma: all trace}
    $\vdashifreePDLsfinlin \bo{\term^{+}} \fml \to \bo{\term^{+}}\bo{\term^{+}} \fml$.
\end{lemma}
\begin{proof}
    We have:
    \begin{center}\scalebox{.9}{
    \begin{prooftree}[small]
        \hypo{\eqref{equation: *}}
        \infer1[\eqref{rule: PDL- MP PROP}]{\bo{\term^{+}} \fml \to \bo{\term}\bo{\term^{+}} \fml}
        \infer1[\eqref{rule: LI}]{\bo{\term^{+}} \fml \to \bo{\term^{+}}\bo{\term^{+}} \fml}
    \end{prooftree}}
    \end{center}
    Hence, this completes the proof.
\end{proof}

\begin{proof}[Proof of \Cref{proposition: all trace}]
    By induction on $\term$.
    We distinguish the following cases:

    \proofcase{Case $\term = \aterm$}
    By \eqref{equation: *}\eqref{equation: union} with \nameref{equation: PROP}.

    \proofcase{Case $\term = \term[2] \compo \term[3]$}
    We have:
    \begin{center}\scalebox{.9}{
    \begin{prooftree}[small]
        \hypo{}
        \infer1[IH]{\bo{A^{+}}\fml \to \bo{\term[3]}\fml}
        \infer1[\eqref{rule: PDL- Mon}]{\bo{A^{+}}\bo{A^{+}}\fml \to \bo{A^{+}}\bo{\term[3]}\fml}
        \hypo{}
        \infer1[IH]{\bo{A^{+}}\bo{\term[3]} \fml \rightarrow \bo{\term[2]}\bo{\term[3]} \fml}
        \hypo{\mathstrut}
        \infer1[\eqref{equation: compo}]{\bo{\term[2] \compo \term[3]} \fml \leftrightarrow \bo{\term[2]}\bo{\term[3]} \fml}
        \infer3[\eqref{rule: PDL- MP PROP}]{\bo{A^{+}}\bo{A^{+}}\fml \to \bo{\term[2] \compo \term[3]} \fml}
    \end{prooftree}}
    \end{center}
    Combining with \Cref{lemma: all trace} using \eqref{rule: PDL- MP PROP} yields ${} \vdashifreePDLsfinlin \bo{A^{+}} \fml \to \bo{\term[2] \compo \term[3]} \fml$.

    \proofcase{Case $\term = \term[2] \union \term[3]$}
    By \eqref{equation: union} with IH.

    \proofcase{Case $\term = \term[2]^{+}$}
    We have:
    \begin{center}\scalebox{.8}{
    \begin{prooftree}[small]
        \hypo{\mathstrut}
        \infer1[\Cref{lemma: all trace}]{\bo{A^{+}} \fml \to \bo{A^{+}}\bo{A^{+}} \fml}
        \hypo{\mathstrut}
        \infer1[IH]{\bo{A^{+}}\bo{A^{+}} \fml \to \bo{\term[2]} \bo{A^{+}} \fml}
        \infer2[\eqref{rule: PDL- MP PROP}]{\bo{A^{+}} \fml \to \bo{\term[2]} \bo{A^{+}} \fml}
        \hypo{}
        \infer1[IH]{\bo{A^{+}} \fml \to \bo{\term[2]} \fml}
        \infer2[\eqref{rule: TC}]{\bo{A^{+}} \fml \to \bo{\term[2]^{+}} \fml}
    \end{prooftree}}
    \end{center}

    \proofcase{Case $\term = \fml[3]? \compo \term[2]$}
    We have:
    \begin{center}\scalebox{.9}{
    \begin{prooftree}[small]
        \hypo{\mathstrut}
        \infer1[IH]{\bo{A^{+}}\fml \to \bo{\term[2]}\fml}
        \hypo{\mathstrut}
        \infer1[\eqref{equation: test L}]{\bo{\fml[3]? \compo \term[2]} \fml \leftrightarrow
        (\fml[3] \to \bo{\term[2]} \fml)}
        \infer2[\eqref{rule: PDL- MP PROP}]{\bo{A^{+}} \fml \to \bo{\fml[3]? \compo \term[2]} \fml}
    \end{prooftree}}
    \end{center}

    \proofcase{Case $\term = \term[2] \compo \fml[3]?$}
    We have:
    \begin{center}\scalebox{.8}{
    \begin{prooftree}[small]
        \hypo{\mathstrut}
        \infer1[\text{\nameref{equation: PROP}}]{\fml \to (\fml[3] \to \fml)}
        \infer1[\eqref{rule: PDL- Mon}]{\bo{A^{+}} \fml \to \bo{A^{+}}(\fml[3] \to \fml)}
        \hypo{}
        \infer1[IH]{\bo{A^{+}}(\fml[3] \to \fml) \to \bo{\term[2]}(\fml[3] \to \fml)}
        \hypo{\mathstrut}
        \infer1[\eqref{equation: test R}]{\bo{\term[2] \compo \fml[3]?} \fml \leftrightarrow
        \bo{\term[2]} (\fml[3] \to \fml)}
        \infer3[\eqref{rule: PDL- MP PROP}]{\bo{A^{+}} \fml \to \bo{\term[2] \compo \fml[3]?} \fml}
    \end{prooftree}}
    \end{center}

    Hence, this completes the proof.
\end{proof}
\end{scope}

\begin{scope}\knowledgeimport{PDL-}\knowledgeimport{PDLwithoutLob}
\subsection{Proof of \Cref{theorem: PDL- completeness without Lob}}\label{section: theorem: PDL- completeness without Lob}
In this section, we prove:
\begin{corollary*}[Restatement of \Cref{theorem: PDL- completeness without Lob}]
    \theoremPDLcompletenesswithoutLob
\end{corollary*}
\AP%
\phantomintro(PDLwithoutLob){consistent}%
\phantomintro(PDLwithoutLob){atom}%
We redefine ``"consistent"'' and its related notations,
by replacing $\vdashifreePDLsfinlin$ with $\vdashifreePDL$.
We recall \Cref{proposition: atoms cond,proposition: atoms consistent saturates,proposition: atoms consistent}.
They hold also for $\vdashifreePDL$, by the same proof as in \Cref{section: PDL- completeness} (since \nameref{equation: Lob'} is not used in them).
\begin{proposition}\label{proposition: atoms without Lob consistent saturates}
All the conditions in \Cref{proposition: atoms consistent saturates} hold for $\vdashifreePDL$.
\end{proposition}
\begin{proposition}\label{proposition: atoms without Lob cond}
All the conditions in \Cref{proposition: atoms cond} hold for $\vdashifreePDL$.
\end{proposition}
\begin{proposition}\label{proposition: atoms without Lob consistent}
All the conditions in \Cref{proposition: atoms consistent} hold for $\vdashifreePDL$.
\end{proposition}

\subsubsection{Canonical Model} \label{section: canonical model without Lob}
Let $\fml_0$ be a "formula".
The ""canonical model@@PDLwithoutLob"" $\intro*\canonicalmodelwithoutLob^{\fml_0}$ is the "generalized structure" defined as follows:
\begin{align*}
    \univ{\canonicalmodelwithoutLob^{\fml_0}} &\defeq \at(\fml_0), &
    \hspace{-1em}\strucuniv^{\canonicalmodelwithoutLob^{\fml_0}} &\defeq \univ{\canonicalmodelwithoutLob^{\fml_0}}^2,\\
    \aterm^{\canonicalmodelwithoutLob^{\fml_0}} &\defeq \set{\tuple{\atom, \atom'} \in \strucuniv^{\canonicalmodelwithoutLob^{\fml_0}} \mid \atomtof{\atom} \land \dia{\aterm}\atomtof{\atom}' \mbox{ is \kl{consistent}}},& 
    \hspace{-1em}\afml^{\canonicalmodelwithoutLob^{\fml_0}} &\defeq \set{\atom \in \univ{\canonicalmodelwithoutLob^{\fml_0}} \mid \afml \in \atom}.
\end{align*}
By definition, we have $\canonicalmodelwithoutLob^{\fml_0} \in \REL{}$.
We have the following truth lemma.

\begin{lemma}[Truth lemma]\label{lemma: prunning' without Lob}
    Let $\fml_0$ be a "formula".
    For every "expression" $\expr$,
    we have the following:
    \begin{enumerate}
        \item \label{lemma: prunning' 1 without Lob}
        If $\expr = \fml \in \cl(\fml_0)$ is a "formula",
        then for all $\atom[2] \in \at(\fml_0)$,
        $\fml \in \atom[2]$ "iff"
        $\atom[2] \in \semifreePDL{\fml}{\canonicalmodelwithoutLob^{\fml_0}}$.

        \item \label{lemma: prunning' 2 without Lob}
        If $\expr = \term \in \tcl(\fml_0)$ is a "term", then
        for all $\tuple{\atom[2], \atom[2]'} \in \strucuniv^{\canonicalmodelwithoutLob^{\fml_0}}$,
        if $\atomtof{\atom[2]} \land \dia{\term}\atomtof{\atom[2]}'$ is "consistent", then $\tuple{\atom[2], \atom[2]'} \in \semifreePDL{\term}{\canonicalmodelwithoutLob^{\fml_0}}$.

        \item \label{lemma: prunning' 3 without Lob}
        If $\expr = \term$ is a "term"
        and $\fml[2]$ is a "formula" "st" $\bo{\term}\fml[2] \in \cl(\fml_0)$,
        then for all $\atom[2] \in \at(\fml_0)$,
        $\bo{\term}\fml[2] \not\in \atom[2]$ "iff"
        there is some $\atom[2]'$ such that
        $\tuple{\atom[2], \atom[2]'} \in \semifreePDL{\term}{\canonicalmodelwithoutLob^{\fml_0}}$ and $\fml[2] \not\in \atom[2]'$.
    \end{enumerate}
\end{lemma}
\begin{proof}
    By induction on $\expr$.

    \proofcase{For \ref{lemma: prunning' 1 without Lob}}
    By the same proof as \Cref{lemma: prunning'}.\ref{lemma: prunning' 1}, using \Cref{proposition: atoms without Lob cond}.

    \proofcase{For \ref{lemma: prunning' 2 without Lob}}
    By the same proof as \Cref{lemma: prunning'}.\ref{lemma: prunning' 2}, using \Cref{proposition: atoms without Lob consistent}.

    \proofcase{For $\Longrightarrow$ of \ref{lemma: prunning' 3}}
    Then $\atomtof{\atom[2]} \land \dia{\term}\lnot\fml[2]$ is "consistent".
    By applying \Cref{proposition: atoms without Lob consistent saturates} iteratively,
    there is some $\atom[2]'$ such that
    $\atomtof{\atom[2]} \land \dia{\term} (\lnot\fml[2] \land \atomtof{\atom[2]}')$ is \kl{consistent}.
    Thus $\atomtof{\atom[2]} \land \dia{\term} \atomtof{\atom[2]}'$ is \kl{consistent}
    and $\fml[2] \not\in \atom[2]'$.
    By IH of \ref{lemma: prunning' 2}, $\tuple{\atom[2], \atom[2]'} \in \semifreePDL{\term}{\canonicalmodelwithoutLob^{\fml_0}}$.

    \proofcase{For $\Longleftarrow$ of \ref{lemma: prunning' 3}}
    By the same proof as the direction $\Longleftarrow$ of \Cref{lemma: prunning'}.\ref{lemma: prunning' 3}, using \Cref{proposition: atoms without Lob cond}.
\end{proof}

By \Cref{lemma: prunning' without Lob}, we are now ready to prove the completeness theorem.
\begin{proof}[Proof of \Cref{theorem: PDL- completeness without Lob}]
    \proofcase{Soundness ($\Longleftarrow$)}
    Easy.
    \proofcase{Completeness ($\Longrightarrow$)}
    Similar to the proof of \Cref{theorem: PDL- completeness} using \Cref{lemma: prunning' without Lob}.
\end{proof}
\end{scope}
 
\clearpage

\begin{scope}\knowledgeimport{PDL-}\knowledgeimport{st}
    
\section{Appendix to \Cref{section: completeness match-language equivalence}~``\nameref*{section: completeness match-language equivalence}''}

\subsection{Proof of \Cref{theorem: PDL- completeness match-language equivalence}}\label{section: theorem: PDL- completeness match-language equivalence}
In this section, we prove the following.
\begin{theorem*}[Restatement of \Cref{theorem: PDL- completeness match-language equivalence}]
    \theoremPDLcompletenessmatchlanguageequivalence
\end{theorem*}
The soundness ($\Longleftarrow$) is clear,
because both \eqref{equation: Det 1} and \eqref{equation: Det 2} hold on $\wordstruc^{\word}$.
Below, we consider the completeness ($\Longrightarrow$).
The proof is similar to that of \Cref{theorem: PDL- completeness} in \Cref{section: PDL- completeness},

\AP%
\phantomintro(st){consistent}%
\phantomintro(st){atom}%
We redefine ``"consistent"'' and its related notations,
by replacing $\vdashifreePDLsfinlin$ with $\vdashstifreePDL$.
We recall \Cref{proposition: atoms cond,proposition: atoms consistent saturates,proposition: atoms consistent}.
They hold also for $\vdashstifreePDL$, by the same proof as in \Cref{section: PDL- completeness}.
\begin{proposition}\label{proposition: atoms st consistent saturates}
All the conditions in \Cref{proposition: atoms consistent saturates} hold for $\vdashstifreePDL$.
\end{proposition}
\begin{proposition}\label{proposition: atoms st cond}
All the conditions in \Cref{proposition: atoms cond} hold for $\vdashstifreePDL$.
\end{proposition}
\begin{proposition}\label{proposition: atoms st consistent}
All the conditions in \Cref{proposition: atoms consistent} hold for $\vdashstifreePDL$.
\end{proposition}

\begin{lemma}["Cf", \Cref{lemma: ordering atoms}]\label{lemma: st ordering atoms}
    Let $\fml_0$ be a "formula".
    Then there is a sequence $\atom_1, \dots, \atom_n$ of pairwise distinct \kl{atoms}
    with $\at(\fml_0) = \set{\atom_1, \dots, \atom_n}$ satisfying the following:
    for all $i \in \rangeone{n}$ and $\bo{\term} \fml[2] \in \cl(\fml_0)$,
    if $\bo{\term} \fml[2] \notin \atom_i$, then
    there is some $j > i$ such that
    $\atomtof{\atom}_i \land \dia{\term} \atomtof{\atom}_j$ is \kl{consistent} and $\fml[2] \not\in \atom_j$.
\end{lemma}
Moreover, using \eqref{equation: Det 1} and \eqref{equation: Det 2},
\Cref{lemma: st ordering atoms} is strengthened as follows.
\begin{lemma}\label{lemma: st ordering atoms strong}
    Let $\fml_0$ be a "formula" and let $\atom_1, \dots, \atom_n$ be the sequence obtained in \Cref{lemma: st ordering atoms}.
    Then there is a partial function $f \colon \rangeone{n} \pfun \rangeone{n}$
    such that 
    for each $i \in \rangeone{n}$,
    if there is some $\bo{\aterm} \fml[2] \in \cl(\fml_0) \setminus \atom_i$,
    then $f(i) > i$ holds and
    for all $\bo{\aterm[2]} \fml[3] \in \cl(\fml_0) \setminus \atom_i$,
    $\atomtof{\atom}_i \land \dia{\aterm[2]} \atomtof{\atom}_{f(i)}$ is \kl{consistent} and $\fml[3] \not\in \atom_{f(i)}$.
\end{lemma}
\begin{proof}
    We define $f$ as follows:
    let $f(i) > i$ be such that $\atomtof{\atom}_i \land \dia{\aterm} \atomtof{\atom}_{f(i)}$ is \kl{consistent} and $\fml[2] \not\in \atom_{f(i)}$ obtained from \Cref{lemma: st ordering atoms}
    if there is some $\bo{\aterm} \fml[2] \in \cl(\fml_0) \setminus \atom_i$,
    and let $f(i)$ be undefined otherwise.
    Let $\bo{\aterm[2]} \fml[3] \in \cl(\fml_0) \setminus \atom_i$ be any.
    If $\aterm[2] \neq \aterm[1]$, then $\dia{\aterm} \fml[2] \land \dia{\aterm[2]} \fml[3]$ is "consistent",
    which contradicts \eqref{equation: Det 2}.
    Thus, $\aterm[2] = \aterm[1]$.
    If $\fml[3] \in \atom_{f(i)}$,
    then $\lnot \bo{\aterm} \fml[3] \land \dia{\aterm} \fml[3]$ is "consistent",
    which contradicts \eqref{equation: Det 1}.
    Thus, $\fml[3] \not\in \atom_{f(i)}$.
    Hence, this $f$ satisfies the required conditions.
\end{proof}

\subsubsection{Canonical Model} \label{section: canonical model st}
Let $\fml_0$ be a "consistent" "formula".
Let $\atom_1, \dots, \atom_n \in \at(\fml_0)$ be linearly ordered \kl{atoms} obtained from \Cref{lemma: st ordering atoms}
and let $f \colon \rangeone{n} \pfun \rangeone{n}$ (with $i < f(i)$ for all $i \in \fdom(f)$).
Let $\mathfrak{i}_1$ be such that $\fml_0 \in \atom_{\mathfrak{i}_1}$
and let $\mathfrak{i}_2, \dots, \mathfrak{i}_{n'}$ be such that
$\mathfrak{i}_{j} = f(\mathfrak{i}_{j-1})$ for each $j \in \range{2}{n'}$ and
$f(\mathfrak{i}_{n'})$ is undefined.
We then let $\atom'_{j} \defeq \atom_{\mathfrak{i}_j}$ for $j \in \rangeone{n'}$.

The ""canonical model@@st"" $\intro*\canonicalmodelst^{\fml_0}$ is the "generalized structure" defined as follows:
\begin{align*}
    \univ{\canonicalmodelst^{\fml_0}} &\defeq
    \set{\atom'_{1}, \dots, \atom'_{n'}},&
    \strucuniv^{\canonicalmodelst^{\fml_0}} &\defeq \set{\tuple{\atom'_{i}, \atom'_{j}} \mid 1 \le i < j \le n'},\\
    \aterm^{\canonicalmodelst^{\fml_0}} &\defeq \set{\tuple{\atom'_{i-1}, \atom'_{i}} \mid i \in \range{2}{n'} \text{ and } \bo{a}\fml[2] \in \cl(\fml_0) \setminus \atom'_{{i-1}} \text{ for some } \fml[2]},\span\span\\
    \afml^{\canonicalmodelst^{\fml_0}} &\defeq \set{\atom \in \univ{\canonicalmodelst^{\fml_0}} \mid \afml \in \atom}.\span\span
\end{align*}
By construction, we have $\canonicalmodelst^{\fml_0} \in \GRELstsfinlin{}$.

We have the following truth lemma.
\begin{lemma}[Truth lemma]\label{lemma: prunning' st}
    Let $\fml_0$ be a "consistent" "formula".
    For every "expression" $\expr$,
    we have the following:
    \begin{enumerate}
        \item \label{lemma: prunning' st 1}
        If $\expr = \fml \in \cl(\fml_0)$ is a "formula",
        then for all $\atom[2] \in \at(\fml_0)$,
        $\fml \in \atom[2]$ "iff" $\atom[2] \in \semifreePDL{\fml}{\canonicalmodelst^{\fml_0}}$.

        \item \label{lemma: prunning' st 2}
        If $\expr = \term \in \tcl(\fml_0)$ is a "term",
        then for all $i \in \range{2}{n'}$,
        if $\atomtof{\atom}'_{i-1} \land \dia{\term}\atomtof{\atom}'_{i}$ is "consistent", then $\tuple{\atom'_{{i-1}}, \atom'_{i}} \in \semifreePDL{\term}{\canonicalmodelst^{\fml_0}}$.

        \item \label{lemma: prunning' st 3}
        If $\expr = \term$ is a "term" and $\fml[2]$ is a "formula" "st" $\bo{\term}\fml[2] \in \cl(\fml_0)$,
        then for all $\atom[2] \in \at(\fml_0)$,
        $\bo{\term}\fml[2] \not\in \atom[2]$ "iff" there is some $\atom[2]'$ such that
        $\tuple{\atom[2], \atom[2]'} \in \semifreePDL{\term}{\canonicalmodelst^{\fml_0}}$ and $\fml[2] \not\in \atom[2]'$.
    \end{enumerate}
\end{lemma}
\begin{proof}
    By induction on $\expr$.
    
    \proofcase{For \ref{lemma: prunning' st 1}}
    By the same proof as \Cref{lemma: prunning'}.\ref{lemma: prunning' 1},
    using \Cref{proposition: atoms st cond}.

    \proofcase{For \ref{lemma: prunning' st 2}}
    By the same proof as \Cref{lemma: prunning'}.\ref{lemma: prunning' 2},
    using \Cref{proposition: atoms st consistent}.
    Here, for Case $\term = \aterm$, we use \Cref{lemma: st ordering atoms strong}.
  
    \proofcase{For $\Longrightarrow$ of \ref{lemma: prunning' st 3}}
    We distinguish the following cases.

    \subproofcase{Case $\term = \aterm$}
    Let $\atom[2] = \atom_{i-1}'$.
    By definition of $\aterm^{\canonicalmodelst^{\fml_0}}$ with \Cref{lemma: st ordering atoms strong},
    $\atom[2]' = \atom_{i}'$ satisfies the condition.

    \subproofcase{Case $\term = \term[2] \union \term[3]$}
    By \Cref{proposition: atoms st cond}, either $\bo{\term[2]} \fml[2] \not\in \atom[2]$ or $\bo{\term[3]} \fml[2] \not\in \atom[2]$ holds.
    By IH, this case has been proved.

    \subproofcase{Case $\term = \term[2] \compo \term[3]$}
    By \Cref{proposition: atoms st cond}, $\bo{\term[2]}\bo{\term[3]} \fml[2] \not\in \atom[2]$.
    By IH (two times), this case has been proved.

    \subproofcase{Case $\term = \term[2]^+$}
    Let $\atom[2] = \atom'_i$.
    We prove this case by induction on $i$ from $n'$ to $1$.
    By \Cref{proposition: atoms st cond},
    either $\bo{\term[2]} \fml[2] \not\in \atom'_i$ or $\bo{\term[2]}\bo{\term[2]^{+}} \fml[2] \not\in \atom'_i$ holds.

    \subsubproofcase{Case $\bo{\term[2]} \fml[2] \not\in \atom'_i$}
    By IH, there is some $\atom[2]'$ such that $\tuple{\atom'_i, \atom[2]'} \in \semifreePDL{\term[2]}{\canonicalmodelst^{\fml_0}}$ and $\fml[2] \not\in \atom[2]'$.
    Hence, there is some $\atom[2]'$ such that $\tuple{\atom'_i, \atom[2]'} \in \semifreePDL{\term[2]^+}{\canonicalmodelst^{\fml_0}}$ and $\fml[2] \not\in \atom[2]'$.

    \subsubproofcase{Case $\bo{\term[2]}\bo{\term[2]^{+}} \fml[2] \not\in \atom'_i$}
    By IH, there is some $\atom[2]'$ such that $\tuple{\atom'_i, \atom[2]'} \in \semifreePDL{\term[2]}{\canonicalmodelst^{\fml_0}}$ and $\bo{\term[2]^{+}} \fml[2] \not\in \atom[2]'$.
    As $\tuple{\atom'_i, \atom[2]'} \in \semifreePDL{\term[2]}{\canonicalmodelst^{\fml_0}} \subseteq \strucuniv^{\canonicalmodelst^{\fml_0}}$,
    there is some $j > i$ such that $\atom[2]' = \atom'_j$.
    Thus by IH, there is some $\atom[2]''$ such that $\tuple{\atom'_j, \atom[2]''} \in \semifreePDL{\term[2]^+}{\canonicalmodelst^{\fml_0}}$ and $\fml[2] \not\in \atom[2]''$.
    Hence, by $\tuple{\atom'_i, \atom[2]''} \in \semifreePDL{\term[2]^{+}}{\canonicalmodelst^{\fml_0}}$,
    this case has been proved.

    \subproofcase{Case $\term = \fml[3]? \compo \term[2]$}
    By \Cref{proposition: atoms st cond},
    we have $\fml[3] \in \atom[2]$ and $\bo{\term[2]} \fml[2] \not\in \atom[2]$.
    By IH "wrt" $\term[2]$,
    there is some $\atom[2]'$ such that $\tuple{\atom[2], \atom[2]'} \in \semifreePDL{\term[2]}{\canonicalmodelst^{\fml_0}}$ and $\fml[2] \not\in \atom[2]'$.
    By IH "wrt" $\fml[3]$,
    we have $\atom[2] \in \semifreePDL{\fml[3]}{\canonicalmodelst^{\fml_0}}$.
    Thus, $\tuple{\atom[2], \atom[2]'} \in \semifreePDL{\fml[3]? \compo \term[2]}{\canonicalmodelst^{\fml_0}}$.

    \subproofcase{Case $\term = \term[2] \compo \fml[3]?$}
    By \Cref{proposition: atoms st cond},
    we have $\bo{\term[2]} (\fml[3] \to \fml[2]) \not\in \atom[2]$.
    By IH "wrt" $\term[2]$,
    there is some $\atom[2]'$ such that $\tuple{\atom[2], \atom[2]'} \in \semifreePDL{\term[2]}{\canonicalmodelst^{\fml_0}}$ and $\fml[3] \to \fml[2] \not\in \atom[2]'$.
    Thus, $\fml[3] \in \atom[2]'$ and $\fml[2] \not\in \atom[2]'$.
    By IH "wrt" $\fml[3]$,
    we have $\atom[2]' \in \semifreePDL{\fml[3]}{\canonicalmodelst^{\fml_0}}$.
    We thus have $\tuple{\atom[2], \atom[2]'} \in \semifreePDL{\term[2] \compo \fml[3]?}{\canonicalmodelst^{\fml_0}}$.

    \proofcase{For $\Longleftarrow$ of \ref{lemma: prunning' st 3}}
    By the same proof as \Cref{lemma: prunning'}.\ref{lemma: prunning' 3},
    using \Cref{proposition: atoms st cond}.
\end{proof}

\begin{proof}[Proof of \Cref{theorem: PDL- completeness match-language equivalence}]
    \proofcase{Soundness ($\Longleftarrow$)}
    Easy.
    \proofcase{Completeness ($\Longrightarrow$)}
    Similar to the proof of \Cref{theorem: PDL- completeness} using \Cref{lemma: prunning' st}.
\end{proof}
\end{scope}
 
\clearpage
\begin{scope} \knowledgeimport{PDLREwLAp}\knowledgeimport{complexity}
\section{Appendix to \Cref{section: complexity}~``\nameref*{section: complexity}''}

\subsection{Proof of the upper bound of \Cref{theorem: complexity PDLREwLAp substitution-closed}}\label{section: theorem: complexity PDLREwLAp substitution-closed}
In this section, we prove the upper bound of the following theorem.
\begin{theorem*}[Restatement of \Cref{theorem: complexity PDLREwLAp substitution-closed}]
\theoremcomplexityPDLREwLApsubstitutionclosed
\end{theorem*}
For the upper bound,
it suffices to show for "\PDLREwLAp" (the others are its syntactic fragments).
The proof proceeds by an analog of the standard automata construction for modal logics and PDLs, "eg", \cite{vardiAutomatatheoreticTechniquesModal1986,calvaneseAutomataTheoreticApproachRegular2009,degiacomoLinearTemporalLogic2013}.
Here, we first remove the identity-part of "term variables" using "formula variables" (\Cref{section: theorem: complexity PDLREwLAp substitution-closed: removing identity}).
We then apply the following well-known techniques:
"tree unwinding" (\Cref{section: theorem: complexity PDLREwLAp substitution-closed: to tree}),
fixing the number of "term variables" and "formula variables" (\Cref{section: theorem: complexity PDLREwLAp substitution-closed: to one term variable,section: theorem: complexity PDLREwLAp substitution-closed: to one formula variable}),
and binary tree encoding like the reduction from S$\omega$S to S2S \cite{weyerDecidabilityS1SS2S2002} (\Cref{section: theorem: complexity PDLREwLAp substitution-closed: to bin}).
We then give a polynomial-time reduction to the "emptiness problem" of "alternating tree automata" (\Cref{section: theorem: complexity PDLREwLAp substitution-closed: to tree automata}), which is in "EXPTIME".

\subsubsection{Removing the identity-part}\label{section: theorem: complexity PDLREwLAp substitution-closed: removing identity}
In this section, we remove the identity-part "wrt" "term variables",
similar to \Cref{section: reduction to identity-free}.
Below, to present a polynomial-time reduction,
we give another reduction without translating to "identity-free \PDL".

We write $\intro*\GRELfinlinprime{}$ for the class of "generalized structures" $\struc \in \GRELsfinlin{}$,
where $\strucuniv^{\struc}$ has been replaced with $\strucuniv^{\struc} \cup \diagonal_{\univ{\struc}}$.
Note that $\bigcup_{\aterm \in \vsig} \aterm^{\struc}$ is still a subset of a "strict linear order",
whereas $\strucuniv^{\struc}$ is a "linear order".

For an "expression" $\expr \in \exprclassPDLREwLAp{\psig, \vsig}$,
we write $\intro*\tonormprime{\expr}$ for the "expression" $\expr$ in which
each $\aterm \in \vsig$ has been replaced with the "term" $\dia{\aterm^{\capid}} \truec? \union \aterm^{\capcomid}$.

We recall the reduction of \Cref{section: reduction to identity-free}.
Let $\psig'$ and $\vsig'$ be sets disjoint from $\psig$ and $\vsig$ having the same "cardinality" as $\vsig$,
let $\pi_1$ be a bijection from $\psig'$ to $\vsig$, and
let $\pi_2$ be a bijection from $\vsig'$ to $\vsig$.
We recall the "substitution" $\reintro*\fromifreePDLsubst$ mapping
each $\afml \in \psig$ to itself,
each $\afml \in \psig'$ to $\dia{\pi_1(\afml)^{\capid}} \truec$, and
each $\aterm \in \vsig'$ to $\pi_2(\aterm)^{\capcomid}$.
For each "expression" $\expr \in \exprclassPDLREwLAp{\psig \dcup \psig', \vsig'}$,
let $\fromifreePDL{\expr} \in \exprclassPDLREwLAp{\psig, \vsig}$ be the "\PDLREwLAp" "expression@@PDLREwLAp" obtained from $\expr$ by applying $\Theta_0$.
\AP
We observe that for every $\expr \in \exprclassPDLREwLAp{\psig, \vsig}$,
there is some $\expr[2] \in \exprclassifreePDL{\psig \dcup \psig', \vsig'}$ such that $\tonormprime{\expr} = \fromifreePDL{\expr[2]}$.
By construction, we have the following lemmas.
\begin{lemma}\label{lemma: reduction to identity-free complexity}
    For every "expression@@PDLREwLAp" $\expr \in \exprclassPDLREwLAp{\psig, \vsig}$,
    $\semPDLREwLAp{\expr}{\struc} = \semPDLREwLAp{\tonormprime{\expr}}{\struc}$.
\end{lemma}
\begin{proof}
    By easy induction on $\expr$
    using $\semPDLREwLAp{\aterm}{\struc} = \semPDLREwLAp{\dia{\aterm^{\capid}} \truec? \union \aterm^{\capcomid}}{\struc}$ for each $\aterm \in \vsig$.
\end{proof}

\begin{lemma}\label{lemma: finite linear orders to finite strict linear orders prime}
    For every "formula@@PDLREwLAp" $\fml \in \fmlclassPDLREwLAp{\psig, \vsig}$,
    we have:
    \[\GRELfinlin{} \modelsfml \fml \quad\iff\quad \GRELfinlinprime{} \modelsfml \fml[2],\]
    where $\fml[2]$ is the "formula@@PDLREwLAp" such that $\tonormprime{\fml} = \fromifreePDL{\fml[2]}$.
\end{lemma}
\begin{proof}
    Let $f$ be the bijection
    \begin{align*}
        f \colon \GRELfinlin{\psig, \vsig} &\to \GRELfinlinprime{\psig \dcup \psig', \vsig'}\\
        \struc[1] &\mapsto \struc[2],
    \end{align*}
    defined as follows:
    \begin{itemize}
        \item
        $\univ{\struc[2]} \defeq \univ{\struc[1]}$ and $\strucuniv^{\struc[2]} \defeq \strucuniv^{\struc} \setminus \diagonal_{\univ{\struc}}$;

        \item 
        $\afml^{\struc[2]} \defeq \afml^{\struc}$ for $\afml \in \psig$ and
        $\afml^{\struc[2]} \defeq \semPDLREwLAp{\dia{\pi_1(\afml)^{\capid}}\truec}{\struc}$ for $\afml \in \psig'$; and

        \item
        $\aterm^{\struc[2]} \defeq \semPDLREwLAp{\pi_2(\aterm)^{\capcomid}}{\struc}$ for $\aterm \in \vsig'$. 
    \end{itemize}
    By easy induction on $\expr$,
    for all $\expr \in \exprclassPDLREwLAp{\psig \dcup \psig', \vsig'}$ and $\struc[1] \in \GRELfinlinprime{}$,
    we have
    $\semPDLREwLAp{\fromifreePDL{\expr}}{\struc[1]} = \semPDLREwLAp{\expr}{f(\struc[1])}$.
    Thus,
    \begin{align*}
    \GRELfinlin{\psig, \vsig} \modelsfml \fml
    &\iff \GRELfinlin{\psig, \vsig} \modelsfml \tonormprime{\fml} \tag{By \Cref{lemma: reduction to identity-free complexity}}\\
    &\iff \GRELfinlinprime{\psig \dcup \psig', \vsig'} \modelsfml \fml[2] \tag{By above and that $f$ is a bijection}.
    \end{align*}
    Hence, this completes the proof.
\end{proof}

\subsubsection{Reduction to the "theory" on $\GRELstfintree{}$}\label{section: theorem: complexity PDLREwLAp substitution-closed: to tree}
We write $\intro*\GRELstfintree{\psig, \vsig}$ for the class of "generalized structures" $\struc \in \GRELfinpartialorder{\psig, \vsig}$ such that
\begin{itemize}
    \item $\struc$ is a labelled rooted "tree", that is,
    the relation $\bigcup_{\aterm \in \vsig} \aterm^{\struc}$ is well-founded,
    there is a root $r \in \univ{\struc}$ such that $\tuple{r, \vertex} \in (\bigcup_{\aterm \in \vsig} \aterm^{\struc})^*$ for every $\vertex \in \univ{\struc}$, and
    $\aterm[1]^{\struc} \cap \aterm[2]^{\struc} = \emptyset$ for distinct $\aterm[1], \aterm[2] \in \vsig$;
    \item $\strucuniv^{\struc} = (\bigcup_{\aterm \in \vsig} \aterm^{\struc})^*$.
\end{itemize}

In this section, by the "tree unwinding" argument,
we give a reduction to the "theory" on $\GRELstfintree{}$.
We recall the standard notion of "bisimulation" \cite{vanbenthemModalCorrespondenceTheory1976}.
\begin{definition}["Bisimulation@bisimulation"]\label{definition: bisimulation}
Let $\struc[1], \struc[2] \in \GREL{}$.
Given two sets $P$ and $A$,
a binary relation $\rel \subseteq \univ{\struc[1]} \times \univ{\struc[2]}$ between $\struc[1]$ and $\struc[2]$ is a \AP""bisimulation"" ("wrt" $P$ and $A$)
if for all $\tuple{\vertex[1], \vertex[2]} \in \rel$, the following holds.
\begin{description}
    \item[(atom)\label{definition: bisimulation atom}]
    For each $\afml \in P$,
    $\vertex[1] \in \semPDLREwLAp{\afml}{\struc[1]}$ "iff" $\vertex[2] \in \semPDLREwLAp{\afml}{\struc[2]}$.

    \item[(forth)\label{definition: bisimulation forth}]
    For each $\aterm \in A$,
    if $\tuple{\vertex[1], \vertex[1]'} \in \semPDLREwLAp{\aterm}{\struc[1]}$, there exists some $\vertex[2]'$ such that $\tuple{\vertex[2], \vertex[2]'} \in \semPDLREwLAp{\aterm}{\struc[2]}$ and $\tuple{\vertex[1]',\vertex[2]'} \in \rel$.

    \item[(back)\label{definition: bisimulation back}]
    For each $\aterm \in A$,
    if $\tuple{\vertex[2], \vertex[2]'} \in \semPDLREwLAp{\aterm}{\struc[2]}$, there exists some $\vertex[1]'$ such that $\tuple{\vertex[1], \vertex[1]'} \in \semPDLREwLAp{\aterm}{\struc[1]}$ and $\tuple{\vertex[1]', \vertex[2]'} \in \rel$.
\end{description}
For $\vertex[1] \in \univ{\struc[1]}$ and $\vertex[2] \in \univ{\struc[2]}$,
we say that $\tuple{\struc[1], \vertex[1]}$ and $\tuple{\struc[2], \vertex[2]}$ are \AP""bisimilar""
if there is a "bisimulation" $\rel$ between $\struc[1]$ and $\struc[2]$ such that $\tuple{\vertex[1], \vertex[2]} \in \rel$.
\lipicsEnd\end{definition}

We write $\AP\intro*\GRELspartialorder{}$ for the class of $\struc \in \GRELpartialorder{}$ "st" $\strucuniv^{\struc}$ is a "strict partial order".
We write $\AP\intro*\GRELpartialorderprime{}$ for the class of "generalized structures" $\struc \in \GRELspartialorder{}$,
where $\strucuniv^{\struc}$ has been replaced with $\strucuniv^{\struc} \cup \diagonal_{\univ{\struc}}$.
\begin{proposition}\label{proposition: bisimulation}
Let $\struc[1], \struc[2] \in \GRELpartialorderprime{}$.
Suppose that $\rel$ is a "bisimulation" between $\struc[1]$ and $\struc[2]$.
For all $\expr \in \exprclassPDLREwLAp{\psig,\vsig}$ and all $\tuple{x, y} \in \rel$, we have:
\begin{enumerate}
    \item\label{proposition: bisimulation 1} If $\expr = \fml$ is a "formula", then we have:
    $x \in \semPDLREwLAp{\fml}{\struc[1]}$ "iff" $y \in \semPDLREwLAp{\fml}{\struc[2]}$.

    \item\label{proposition: bisimulation 2} If $\expr = \term$ is a "term", then we have:
    if $\tuple{x, x'} \in \semPDLREwLAp{\term}{\struc[1]}$ then there exists $y'$ such that $\tuple{y, y'} \in \semPDLREwLAp{\term}{\struc[2]}$ and $\tuple{x',y'} \in \rel$.

    \item\label{proposition: bisimulation 3} If $\expr = \term$ is a "term", then we have:
    if $\tuple{y, y'} \in \semPDLREwLAp{\term}{\struc[2]}$ then there exists $x'$ such that $\tuple{x, x'} \in \semPDLREwLAp{\term}{\struc[1]}$ and $\tuple{x',y'} \in \rel$.
\end{enumerate}
\end{proposition}
\begin{proof}
    Via the transformation of \Cref{definition: normal form},
    it suffices to show for the form in which $\bl^{\capid}$ and $\bl^{\capcomid}$ only applied to "term variables".
    We show this proposition by induction on the size of $\expr$.

    \proofcase{For \ref{proposition: bisimulation 1}}
    We distinguish the following cases:
    \proofcase{Case $\fml = \afml$}
    By the definition of \ref{definition: bisimulation atom}.
    \proofcase{Case $\fml = \fml[2] \to \fml[3]$, Case $\fml = \falsec$}
    Easy by IH.
    \proofcase{Case $\fml = \bo{\term}\fml[2]$}
    By IH with \ref{definition: bisimulation forth} and \ref{definition: bisimulation back}.

    \proofcase{For \ref{proposition: bisimulation 2}}
    We distinguish the following cases:
    \proofcase{Case $\term = \aterm^{\capid}$}
    Because $\semPDLREwLAp{\aterm}{\struc[2]} \subseteq \univ{\struc[2]}^2 \setminus \triangle_{\univ{\struc[2]}}$ by $\struc \in \GRELpartialorderprime{}$,
    we have
    $\semPDLREwLAp{\aterm^{\capid}}{\struc[2]} = \emptyset$.
    \proofcase{Case $\term = \aterm^{\capcomid}$}
    By \ref{definition: bisimulation forth},
    there is some $y'$ such that $\tuple{y, y'} \in \semPDLREwLAp{\aterm}{\struc[2]}$.
    Because $\semPDLREwLAp{\aterm}{\struc[2]} \subseteq \univ{\struc[2]}^2 \setminus \triangle_{\univ{\struc[2]}}$,
    we have $\tuple{y, y'} \in \semPDLREwLAp{\aterm^{\capcomid}}{\struc[2]}$.
    \proofcase{Case $\term = \term[2] \union \term[3]$, Case $\term = \term[2] \compo \term[3]$, Case $\term = \term[2]^{+}$}
    Easy by applying IH (multiple times).
    \proofcase{Case $\term = \term[2]^{\adom}$}
    Then $x' = x$ and there is no $x''$ such that $\tuple{x, x''} \in \semPDLREwLAp{\term[2]}{\struc[1]}$.
    By the contrapositive of \ref{definition: bisimulation back}, we have $\tuple{y, y''} \not\in \semPDLREwLAp{\term[2]}{\struc[2]}$ for all $y''$.
    Thus, $\tuple{y, y} \in \semPDLREwLAp{\term[2]^{\adom}}{\struc[2]}$.
    \proofcase{Case $\term = \fml[3]?$}
    By IH of \ref{proposition: bisimulation 1}.
    \proofcase{For \ref{proposition: bisimulation 3}}
    Similar to \ref{proposition: bisimulation 2}.
\end{proof}
\begin{remark}
    In \Cref{proposition: bisimulation},
    removing the identity-part (\Cref{section: theorem: complexity PDLREwLAp substitution-closed: removing identity}) is crucial.
    If not, for instance, consider the following two "bisimilar" "structures" $\struc[1]$ and $\struc[2]$:
    \[\struc = \begin{tikzpicture}[baseline = -.5ex]
    \tikzstyle{vert}=[draw = black, circle, fill = gray!10, inner sep = 2pt, minimum size = 1pt, font = \scriptsize]
    \graph[grow right = .8cm, branch down = 4.5ex]{
    {s1/{$r$}[vert]}
    };
    \graph[use existing nodes, edges={color=black, pos = .5, earrow}, edge quotes={fill=white, inner sep=1pt,font= \scriptsize}]{
        s1 ->["$a$", out = 45, in = 135, looseness = 8] s1;
    };
\end{tikzpicture} \quad\leadsto\quad
\struc[2] = \begin{tikzpicture}[baseline = -.5ex]
    \tikzstyle{vert}=[draw = black, circle, fill = gray!10, inner sep = 2pt, minimum size = 1pt, font = \scriptsize]
    \graph[grow right = .8cm, branch down = 4.5ex]{
    {s1/{$r$}[vert]} -!- {s2/{}[vert]} -!- {s3/{}[vert]} -!- {s4/{$\dots$}}
    };
    \graph[use existing nodes, edges={color=black, pos = .5, earrow}, edge quotes={fill=white, inner sep=1pt,font= \scriptsize}]{
        s1 ->["$a$"] s2 ->["$a$"] s3 -> s4;
    };
\end{tikzpicture} 
\]
Then $r \in \semPDLREwLAp{a^{\capid}}{\struc[1]}$ but $r \not\in \semPDLREwLAp{a^{\capid}}{\struc[2]}$.
\end{remark}

\begin{lemma}\label{lemma: reduction to fintree}
    For all $\fml \in \fmlclassPDLREwLAp{\psig,\vsig}$,
    $\GRELfinlinprime{} \modelsfml \fml$ "iff" $\GRELstfintree{} \modelsfml \fml$.
\end{lemma}
\begin{proof}
\proofcase{($\Longrightarrow$)}
Let $\struc \in \GRELstfintree{}$.
Let $\struc' \in \GRELfinlinprime{}$ be a "generalized structure" obtained from $\struc$ by extending a "partial order" $\strucuniv^{\struc}$ with a "linear order".
As the "universal relation" is not used in the evaluation of $\fml$,
$\struc \modelsfml \fml$ "iff" $\struc' \modelsfml \fml$.

\proofcase{($\Longleftarrow$)}
Let $\struc \in \GRELfinlinprime{}$ and $r \in \univ{\struc}$ be any.
Let $\struc[2] \in \GRELstfintree{}$ be the ``\AP""tree unwinding""'' of $\struc$ from $r$, defined as follows:
\begin{itemize}
    \item The universe $\univ{\struc[2]} \subseteq (\univ{\struc} \vsig)^* \univ{\struc}$ is the smallest set (of traces)
    with $r \in \univ{\struc[2]}$
    such that if $\word y \in \univ{\struc[2]}$ and $\tuple{y, y'} \in \aterm^{\struc}$,
    then $\word y \aterm y' \in \univ{\struc[2]}$,
    \item $\aterm^{\struc[2]} \defeq \set{\tuple{\word y, \word y \aterm y'} \mid \word \in (\univ{\struc} \vsig)^*, y \in \univ{\struc}, \aterm \in \vsig, y' \in \univ{\struc}} \cap \univ{\struc[2]}^2$,
    \item $\strucuniv^{\struc[2]}$ is the smallest "partial order" subsuming $\bigcup_{\aterm \in \vsig} \aterm^{\struc[2]}$.
\end{itemize}
For instance, we transform $\struc$ to $\struc[2]$ as follows:
\[\struc = \begin{tikzpicture}[baseline = -5.ex]
    \tikzstyle{vert}=[draw = black, circle, fill = gray!10, inner sep = 2pt, minimum size = 1pt, font = \scriptsize]
    \graph[grow right = .8cm, branch down = 4.5ex]{
    {/, s2/{}[vert]} -!- {s1/{$r$}[vert], /, s4/{}[vert]} -!- {/, s3/{}[vert], s5/{}[vert, yshift= 1ex]}
    };
    \graph[use existing nodes, edges={color=black, pos = .5, earrow}, edge quotes={fill=white, inner sep=1pt,font= \scriptsize}]{
        s1 ->["$a$"] s2; s1 ->["$b$"] s3;
        s2 ->["$c$"] s4; s3 ->["$d$"] s4;
    };
\end{tikzpicture} \quad\leadsto\quad
\struc[2] = \begin{tikzpicture}[baseline = -5.ex]
    \tikzstyle{vert}=[draw = black, circle, fill = gray!10, inner sep = 2pt, minimum size = 1pt, font = \scriptsize]
    \graph[grow right = .8cm, branch down = 4.5ex]{
    {/, s2/{}[vert], s4/{}[vert]} -!- {s1/{$r$}[vert], /} -!- {/, s3/{}[vert], s5/{}[vert]}
    };
    \graph[use existing nodes, edges={color=black, pos = .5, earrow}, edge quotes={fill=white, inner sep=1pt,font= \scriptsize}]{
        s1 ->["$a$"] s2; s1 ->["$b$"] s3;
        s2 ->["$c$", pos =.4] s4; s3 ->["$d$", pos =.4] s5;
    };
\end{tikzpicture} 
\]
Because $\tuple{\struc, r}$ and $\tuple{\struc[2], r}$ are "bisimilar" and $r \in \semPDLREwLAp{\fml}{\struc[2]}$ by $\struc[2] \in \GRELstfintree{}$,
we have $r \in \semPDLREwLAp{\fml}{\struc}$ by \Cref{proposition: bisimulation}.
Hence, $\GRELfinlinprime{} \modelsfml \fml$.
\end{proof}

\subsubsection{Reduction to the "theory" on $\GRELstfintree{\psig, \set{\const{S}}}$}\label{section: theorem: complexity PDLREwLAp substitution-closed: to one term variable}
Let $\const{S}$ (Successor) be a fresh "term variable".
In this section,
we give a reduction from the "theory" on $\GRELstfintree{\psig, \vsig}$ to the "theory" on $\GRELstfintree{\psig \dcup \vsig, \set{\const{S}}}$.
For a class $\strucclass$, we say that a "formula" $\fml$ is \AP""satisfiable"" if $\semPDLREwLAp{\fml}{\struc} \neq \emptyset$ for some $\struc \in \strucclass$.

For an $\expr \in \exprclassPDLREwLAp{\psig, \vsig}$,
let $\intro*\tonormonetermvariable{\expr} \in \exprclassPDLREwLAp{\psig \dcup \vsig, \set{\const{S}}}$ be the "expression" $\expr$ in which each $\aterm \in \vsig$ has been replaced with the "term" $\const{S} \aterm?$ ($\aterm$ is used as a "formula variable").
We then have:
\begin{lemma}\label{lemma: reduction to one term variable}
    Let $\fml[2]_0 \defeq \bo{\const{S}^+}((\bigvee_{\aterm \in \vsig} \aterm) \land (\bigwedge_{\substack{\aterm[1], \aterm[2] \in \vsig\\ "st" \aterm[1] \neq \aterm[2]}} \lnot \aterm[1] \lor \lnot \aterm[2]))$.
    For every $\fml \in \fmlclassPDLREwLAp{\psig, \vsig}$, we have:

    $\fml$ is "satisfiable" on $\GRELstfintree{\psig, \vsig}$ ~"iff"~
    $\tonormonetermvariable{\fml} \land \fml[2]_0$ is "satisfiable" on $\GRELstfintree{\psig \dcup \vsig, \set{\const{S}}}$.
\end{lemma}
\begin{proof}
    \proofcase{$(\Longrightarrow)$}
    Let $\struc \in \GRELstfintree{\psig, \vsig}$.
    Let $\struc[2] \in \GRELstfintree{\psig \dcup \vsig, \set{\const{S}}}$ be the "generalized structure" defined as follows:
    \begin{itemize}
        \item
        $\univ{\struc[2]} \defeq \univ{\struc}$;

        \item
        $\const{S}^{\struc[2]} \defeq \bigcup_{\aterm \in \vsig} \aterm^{\struc}$;

        \item
        $\afml^{\struc[2]} \defeq \afml^{\struc[1]}$ for $\afml \in \psig$ and
        $\aterm^{\struc[2]} \defeq \set{y \mid \tuple{x, y} \in \aterm^{\struc[1]}}$ for $\aterm \in \vsig$;

        \item
        $\strucuniv^{\struc[2]}$ is the smallest "partial order" subsuming the binary relation $\const{S}^{\struc[2]}$.
    \end{itemize}
    For instance, we transform $\struc$ to $\struc[2]$ as follows:
    \[\struc = 
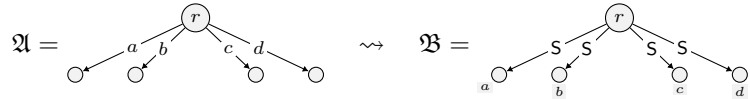
\begin{tikzpicture}[baseline = -3.ex]
    \tikzstyle{vert}=[draw = black, circle, fill = gray!10, inner sep = 2pt, minimum size = 1pt, font = \scriptsize]
    \graph[grow right = .8cm, branch down = 5.ex]{
    {/, s1/{}[vert]} -!- {/, s2/{}[vert]} -!- {s/{$r$}[vert]} -!- {/, s3/{}[vert]} -!- {/, s4/{}[vert]}
    };
    \graph[use existing nodes, edges={color=black, pos = .5, earrow}, edge quotes={fill=white, inner sep=1pt,font= \scriptsize}]{
        s ->["$a$"] s1; s ->["$b$"] s2; s ->["$c$"] s3; s ->["$d$"] s4;
    };
\end{tikzpicture} \quad\leadsto\quad
\struc[2] = \begin{tikzpicture}[baseline = -3.ex]
    \tikzstyle{vert}=[draw = black, circle, fill = gray!10, inner sep = 2pt, minimum size = 1pt, font = \scriptsize]
    \graph[grow right = .8cm, branch down = 5.ex]{
    {/, s1/{}[vert]} -!- {/, s2/{}[vert]} -!- {s/{$r$}[vert]} -!- {/, s3/{}[vert]} -!- {/, s4/{}[vert]}
    };
    \node[below left = .1ex and 0.cm of s1, fill = gray!10, font = \tiny, inner sep = 1pt]{${a}$};
    \node[below = .1ex of s2, fill = gray!10, font = \tiny, inner sep = 1pt]{${b}$};
    \node[below = .1ex of s3, fill = gray!10, font = \tiny, inner sep = 1pt]{${c}$};
    \node[below = .1ex of s4, fill = gray!10, font = \tiny, inner sep = 1pt]{${d}$};
    \graph[use existing nodes, edges={color=black, pos = .5, earrow}, edge quotes={fill=white, inner sep=1pt,font= \scriptsize}]{
        s ->["$\const{S}$"] {s1,s2,s3,s4};
    };
\end{tikzpicture}\]
    We then have $\semPDLREwLAp{\expr}{\struc} = \semPDLREwLAp{\tonormonetermvariable{\expr}}{\struc[2]}$ by easy induction on $\expr$.
    We also have $\semPDLREwLAp{\fml[2]_0}{\struc[2]} = \univ{\struc[2]}$ by construction.
    Hence, $\tonormonetermvariable{\fml} \land \fml[2]_0$ is "satisfiable" at $\struc[2]$.

    \proofcase{$(\Longleftarrow)$}
    By using the "bisimulation" (\Cref{proposition: bisimulation}),
    "wlog", we can assume that $\fml$ is true at the root $r$ of some $\struc \in \GRELstfintree{\psig \dcup \vsig, \set{\const{S}}}$.
    Let $\struc[2] \in \GRELstfintree{\psig, \vsig}$ be the "generalized structure" defined as follows:
    \begin{itemize}
        \item
        $\univ{\struc[2]} \defeq \univ{\struc}$;

        \item
        $\aterm^{\struc[2]} \defeq \semPDLREwLAp{\const{S} \aterm?}{\struc}$ for $\aterm \in \vsig$;

        \item
        $\afml^{\struc[2]} \defeq \afml^{\struc}$ for $\afml \in \psig$;

        \item $\strucuniv^{\struc[2]}$ is the smallest "partial order" extending the binary relation $\bigcup_{\aterm \in \vsig} \aterm^{\struc[2]}$.
    \end{itemize}
    We then have $\semPDLREwLAp{\tonormonetermvariable{\expr}}{\struc} = \semPDLREwLAp{\expr}{\struc[2]}$ by easy induction on $\expr$.
    By using the condition of $r \in \semPDLREwLAp{\fml[2]_0}{\struc[2]}$,
    we have $\struc[2] \in \GRELstfintree{\psig, \vsig}$.
\end{proof}

\subsubsection{Reduction to the "theory" on $\GRELstfinbintree{\psig, \set{\const{S}}}$}\label{section: theorem: complexity PDLREwLAp substitution-closed: to bin}
Let $\AP\intro*\GRELstfinbintree{\psig, \vsig}$ denote the class of $\struc \in \GRELstfintree{\psig, \vsig}$ such that
\begin{itemize}
    \item $\#\set{\vertex[2] \mid \tuple{\vertex, \vertex[2]} \in \bigcup_{\aterm \in \vsig} \aterm^{\struc}} \le 2$ for every $\vertex \in \univ{\struc}$.
\end{itemize}
Let $\const{D}$ (Down) be a fresh "formula variable".
In this section, by a binary tree encoding, we give a reduction from
the "theory" on $\GRELstfintree{\psig, \set{\const{S}}}$ to
the "theory" on $\GRELstfinbintree{\psig \dcup \set{\const{D}}, \set{\const{S}}}$. 

For an $\expr \in \exprclassPDLREwLAp{\psig, \set{\const{S}}}$,
let $\intro*\tonormbin{\expr} \in \exprclassPDLREwLAp{\psig \dcup \set{\const{D}}, \set{\const{S}}}$ be the "expression" $\expr$ in which
each $\const{S}$ has been replaced with the "term" $\const{S} \const{D}? (\const{S} \lnot \const{D}?)^*$.
We then have:
\begin{lemma}\label{lemma: reduction to finbintree}
    For every $\fml \in \fmlclassPDLREwLAp{\psig, \set{\const{S}}}$, we have:

    $\fml$ is "satisfiable" on $\GRELstfintree{\psig, \set{\const{S}}}$
    ~"iff"~
    $\tonormbin{\fml} \land \bo{\const{S}} \const{D}$ is "satisfiable" on $\GRELstfinbintree{\psig \dcup \set{\const{D}}, \set{\const{S}}}$.
\end{lemma}
\begin{proof}
    \proofcase{$(\Longrightarrow)$}
    By the "bisimulation" (of \Cref{proposition: bisimulation}),
    "wlog", we can assume that $\fml$ is true at the root $r$ of some $\struc \in \GRELstfintree{\psig, \set{\const{S}}}$.
    Let $\le$ be a "linear order" extending $\strucuniv^{\struc}$.
    Let $\struc[2] \in \GRELstfinbintree{\psig \dcup \set{\const{D}}, \set{\const{S}}}$ be the "generalized structure" defined as follows:
    \begin{itemize}
        \item
        $\univ{\struc[2]} \defeq \univ{\struc}$;

        \item
        $\const{S}^{\struc[2]}$ and $\const{D}^{\struc[2]}$ are defined the smallest sets satisfying the following:
        \begin{itemize}
            \item For each $y_0 \in \univ{\struc[2]}$,
            let $y_1 \le y_2 \le \dots \le y_n$ be "st" $\set{y_i \mid i \in \rangeone{n}} = \set{y \mid \tuple{y_0, y} \in \bigcup_{\aterm \in \vsig} \aterm^{\struc}}$.
            Then,
            $y_1 \in \const{D}^{\struc[2]}$ and $\tuple{y_{i-1}, y_i} \in \const{S}^{\struc[2]}$ for all $i \in \rangeone{n}$.
        \end{itemize}

        \item
        $\afml^{\struc[2]} \defeq \afml^{\struc[1]}$ for each $\afml \in \psig$.
 
        \item
        $\strucuniv^{\struc[2]}$ is the smallest "partial order" subsuming the binary relation $\const{S}^{\struc[2]}$.
    \end{itemize}
    For instance, we transform $\struc$ to $\struc[2]$ as follows:
    \[\struc = \begin{tikzpicture}[baseline = -5.5ex]
    \tikzstyle{vert}=[draw = black, circle, fill = gray!10, inner sep = 2pt, minimum size = 1pt, font = \scriptsize]
    \graph[grow right = .8cm, branch down = 5.ex]{
    {/, s1/{}[vert], s11/{}[vert]} -!- {/, s2/{}[vert], s12/{}[vert]} -!- {s/{$r$}[vert]} -!- {/, s3/{}[vert]} -!- {/, s4/{}[vert]}
    };
    \node[below left = .1ex and 0 of s2, fill = gray!10, font = \tiny, inner sep = 1pt]{${\const{B}}$};
    \node[below left = .1ex and 0 of s4, fill = gray!10, font = \tiny, inner sep = 1pt]{${\const{B}}$};
    \graph[use existing nodes, edges={color=black, pos = .5, earrow}, edge quotes={fill=white, inner sep=1pt,font= \scriptsize}]{
        s ->["$\const{S}$"] s1; s ->["$\const{S}$"] s2; s ->["$\const{S}$"] s3; s ->["$\const{S}$"] s4;
        s1 ->["$\const{S}$"] {s11, s12};
    };
\end{tikzpicture} \quad\leadsto\quad
\struc[2] = \begin{tikzpicture}[baseline = -5.5ex]
    \tikzstyle{vert}=[draw = black, circle, fill = gray!10, inner sep = 2pt, minimum size = 1pt, font = \scriptsize]
    \graph[grow right = .8cm, branch down = 5.ex]{
    {/, s1/{}[vert], s11/{}[vert]} -!- {/, s2/{}[vert], s12/{}[vert]} -!- {s/{$r$}[vert]} -!- {/, s3/{}[vert]} -!- {/, s4/{}[vert]}
    };
    \node[below left = .1ex and 0 of s1, fill = gray!10, font = \tiny, inner sep = 1pt]{${\const{D}}$};
    \node[below left = .1ex and 0 of s11, fill = gray!10, font = \tiny, inner sep = 1pt]{${\const{D}}$};
    \node[below left = .1ex and 0 of s2, fill = gray!10, font = \tiny, inner sep = 1pt]{${\const{B}}$};
    \node[below left = .1ex and 0 of s4, fill = gray!10, font = \tiny, inner sep = 1pt]{${\const{B}}$};
    \graph[use existing nodes, edges={color=black, pos = .5, earrow}, edge quotes={fill=white, inner sep=1pt,font= \scriptsize}]{
        s ->["$\const{S}$"] s1 ->["$\const{S}$"] s2 ->["$\const{S}$"] s3 ->["$\const{S}$"] s4;
        s1 ->["$\const{S}$"] s11 ->["$\const{S}$"] s12;
    };
\end{tikzpicture}\]
    We then have $\semPDLREwLAp{\expr}{\struc} = \semPDLREwLAp{\tonormbin{\expr}}{\struc[2]}$ by easy induction on $\expr$.
    We also have $r \in \semPDLREwLAp{\fml[2]_0}{\struc[2]}$.
    Hence, $\tonormbin{\fml} \land \fml[2]_0$ is true at the root $r$ of $\struc[2]$.

    \proofcase{$(\Longleftarrow)$}
     By using the "bisimulation" (\Cref{proposition: bisimulation}),
    "wlog", we can assume that $\fml$ is true at the root $r$ of some $\struc \in \GRELstfinbintree{\psig \dcup \set{\const{D}}, \set{\const{S}}}$.
    Let $\struc[2]$ be the "generalized structure" defined as follows:
    \begin{itemize}
        \item
        $\univ{\struc[2]} \defeq \univ{\struc}$;

        \item
        $\const{S}^{\struc[2]} \defeq \semPDLREwLAp{\const{S} \const{D}? (\const{S} \lnot \const{D}?)^*}{\struc}$;

        \item
        $\afml^{\struc[2]} \defeq \afml^{\struc}$ for each $\afml \in \psig$;

        \item $\strucuniv^{\struc[2]}$ is the smallest "partial order" subsuming the binary relation $\const{S}^{\struc[2]}$.
    \end{itemize}
    For instance, we transform $\struc$ to $\struc[2]$ as follows:
    \[\struc = \begin{tikzpicture}[baseline = -5.5ex]
    \tikzstyle{vert}=[draw = black, circle, fill = gray!10, inner sep = 2pt, minimum size = 1pt, font = \scriptsize]
    \graph[grow right = .8cm, branch down = 5.ex]{
    {/, s1/{}[vert], s11/{}[vert]} -!- {/, s2/{}[vert], s12/{}[vert]} -!- {s/{$r$}[vert], /, s13/{}[vert]} -!- {/, s4/{}[vert]} -!- {/, s5/{}[vert]}
    };
    \node[below left = .1ex and 0 of s1, fill = gray!10, font = \tiny, inner sep = 1pt]{${\const{D}}$};
    \node[below left = .1ex and 0 of s4, fill = gray!10, font = \tiny, inner sep = 1pt]{${\const{D}}$};
    \node[below left = .1ex and 0 of s11, fill = gray!10, font = \tiny, inner sep = 1pt]{${\const{D}}$};
    \graph[use existing nodes, edges={color=black, pos = .5, earrow}, edge quotes={fill=white, inner sep=1pt,font= \scriptsize}]{
        s ->["$\const{S}$"] s1 ->["$\const{S}$"] s2; s ->["$\const{S}$"] s4 ->["$\const{S}$"] s5;
        s1 ->["$\const{S}$"] s11 ->["$\const{S}$"] s12 ->["$\const{S}$"] s13;
    };
    \end{tikzpicture} \quad\leadsto\quad
\struc[2] = \begin{tikzpicture}[baseline = -5.5ex]
    \tikzstyle{vert}=[draw = black, circle, fill = gray!10, inner sep = 2pt, minimum size = 1pt, font = \scriptsize]
     \graph[grow right = .8cm, branch down = 5.ex]{
    {/, s1/{}[vert], s11/{}[vert]} -!- {/, s2/{}[vert], s12/{}[vert]} -!- {s/{$r$}[vert], /, s13/{}[vert]} -!- {/, s4/{}[vert]} -!- {/, s5/{}[vert]}
    };
    \graph[use existing nodes, edges={color=black, pos = .5, earrow}, edge quotes={fill=white, inner sep=1pt,font= \scriptsize}]{
        s ->["$\const{S}$"] {s1,s2,s4,s5};
        s1 ->["$\const{S}$"] {s11, s12, s13};
    };
\end{tikzpicture}\]
    We then have $\semPDLREwLAp{\tonormbin{\expr}}{\struc} = \semPDLREwLAp{\expr}{\struc[2]}$ by easy induction on $\expr$.
    Hence, $\fml$ is true at the root $r$ of $\struc[2]$.
    By the condition of $r \in \semPDLREwLAp{\bo{\const{S}} \const{D}}{\struc[2]}$,
    every vertex $\vertex \in \univ{\struc}$ satisfies $\tuple{r, \vertex} \in (\const{S}^{\struc[2]})^*$.
    Hence, $\struc[2] \in \GRELstfintree{\psig, \set{\const{S}}}$.
\end{proof}

\subsubsection{Reduction to the "theory" on $\GRELstfinbintree{\set{\const{B}}, \set{\const{S}}}$}\label{section: theorem: complexity PDLREwLAp substitution-closed: to one formula variable}
Let $\const{B}$ (Black) be a fresh "formula variable" different from $\const{S}$.
In this section, by a unary encoding of "formula variables",
we give a reduction from the "theory" on $\GRELstfinbintree{\psig, \set{\const{S}}}$ to the "theory" on $\GRELstfinbintree{\set{\const{B}}, \set{\const{S}}}$.

For an $\expr \in \exprclassPDLREwLAp{\psig, \set{\const{S}}}$,
let $\exprpsig(\expr) = \set{\afml_0, \afml_1, \dots, \afml_{n-1}}$ where $\afml_0, \afml_1, \dots, \afml_{n-1}$ are pairwise distinct, and
let $\AP\intro*\tonormuni{\expr} \in \exprclassPDLREwLAp{\set{\const{B}}, \set{\const{S}}}$ be the $\expr$ in which
each occurrence $\afml_i \in \psig$ has been replaced with the "formula" $\dia{\const{S}^{i}} \const{B}$
and each occurrence $\const{S}$ has been replaced with the "term" $\const{S}^{n}$.
We then have:
\begin{lemma}\label{lemma: reduction to one formula variable}
    For every $\fml \in \fmlclassPDLREwLAp{\psig, \set{\const{S}}}$, we have:
    
    $\fml$ is "satisfiable" on $\GRELstfinbintree{\psig, \set{\const{S}}}$
    ~"iff"~
    $\tonormuni{\fml}$ is "satisfiable" on $\GRELstfinbintree{\set{\const{B}}, \set{\const{S}}}$.
\end{lemma}
\begin{proof}
    \proofcase{$(\Longrightarrow)$}
    Let $\struc \in \GRELstfinbintree{\psig, \set{\const{S}}}$.
    Let $\struc[2] \in \GRELstfinbintree{\set{\const{B}}, \set{\const{S}}}$ be the "generalized structure" defined as follows:
    \begin{itemize}
        \item
        $\univ{\struc[2]} \defeq \univ{\struc} \times \range{0}{n-1}$ (we write $\tuple{x, i} \in \univ{\struc[2]}$ for $(x)_i$, for short);

        \item
        $\const{S}^{\struc[2]} \defeq \set{\tuple{(x)_{i-1}, (x)_{i}} \mid x \in \univ{\struc}, i \in \rangeone{n-1}}
        \cup \set{\tuple{(x)_{n-1}, (y)_0} \mid \tuple{x, y} \in \const{S}^{\struc[1]}}$;

        \item
        $\const{B}^{\struc[2]} \defeq \set{(x)_{i} \mid x \in \afml_{i}^{\struc}, i \in \range{0}{n-1}}$;

        \item
        $\strucuniv^{\struc[2]}$ is the smallest "partial order" extending the binary relation $\const{S}^{\struc[2]}$.
    \end{itemize}
    For instance, when $n = 4$, we transform $\struc$ to $\struc[2]$ as follows:
    \[\struc = 
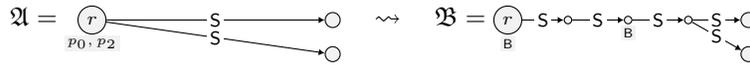
\begin{tikzpicture}[baseline = -.5ex]
    \tikzstyle{vert}=[draw = black, circle, fill = gray!10, inner sep = 2pt, minimum size = 1pt, font = \scriptsize]
    \graph[grow right = .8cm, branch down = 3ex]{
    {s1/{$r$}[vert]} -!- {s2/{}} -!- {s3/{}} -!- {s4/{}} -!- {s5/{}[vert], s5'/{}[vert]}
    };
    \node[below = .1ex of s1, fill = gray!10, font = \tiny, inner sep = 1pt]{$\afml_0, \afml_2$};
    \graph[use existing nodes, edges={color=black, pos = .5, earrow}, edge quotes={fill=white, inner sep=1pt,font= \scriptsize}]{
    s1 ->["$\const{S}$"] {s5,s5'};
    };
\end{tikzpicture} \quad\leadsto\quad
\struc[2] = \begin{tikzpicture}[baseline = -.5ex]
    \tikzstyle{vert}=[draw = black, circle, fill = gray!10, inner sep = 2pt, minimum size = 1pt, font = \scriptsize]
    \graph[grow right = .8cm, branch down = 3ex]{
    {s1/{$r$}[vert]} -!- {s2/{}[vert, inner sep = 1pt]} -!- {s3/{}[vert, inner sep = 1pt]} -!- {s4/{}[vert, inner sep = 1pt]} -!- {s5/{}[vert], s5'/{}[vert]}
    };
    \node[below = .1ex of s1, fill = gray!10, font = \tiny, inner sep = 1pt]{$\const{B}$};
    \node[below = .1ex of s3, fill = gray!10, font = \tiny, inner sep = 1pt]{$\const{B}$};
    \graph[use existing nodes, edges={color=black, pos = .5, earrow}, edge quotes={fill=white, inner sep=1pt,font= \scriptsize}]{
    s1 ->["$\const{S}$"] s2 ->["$\const{S}$"] s3 ->["$\const{S}$"] s4 ->["$\const{S}$"] {s5, s5'};
    };
\end{tikzpicture}\]
    First, we have that $\tuple{(x)_i, (y)_j} \in \semPDLREwLAp{\tonormuni{\term}}{\struc[2]}$ implies $i = j$,
    by easy induction on $\term$ (note that $\const{S}$ always occurs as $\const{S}^{n}$ in $\tonormuni{\term}$).
    Using this (for the cases of "composition" and "Kleene plus"),
    by easy induction on $\expr$,
    for every "\PDLREwLAp" "expression@@PDLREwLAp" $\expr$, we have:
    \begin{itemize}
        \item if $\expr = \term$ is a "term@@PDLREwLAp", then
        $\tuple{x, y} \in \semPDLREwLAp{\term}{\struc}$ "iff" $\tuple{(x)_0, (y)_0} \in \semPDLREwLAp{\tonormuni{\term}}{\struc[2]}$;
        \item if $\expr = \fml$ is a "formula@@PDLREwLAp", then
        $x \in \semPDLREwLAp{\fml}{\struc}$ "iff" $(x)_0 \in \semPDLREwLAp{\tonormuni{\fml}}{\struc[2]}$.
    \end{itemize}
    Hence,
    $\tonormuni{\fml}$ is also "satisfiable" on $\GRELstfinbintree{\set{\const{B}}, \set{\const{S}}}$.
    
    \proofcase{$(\Longleftarrow)$}
    By using the "bisimulation" (\Cref{proposition: bisimulation}),
    "wlog", we can assume that $\tonormuni{\term}$ is true at the root $r$ of some $\struc \in \GRELstfinbintree{\set{\const{B}}, \set{\const{S}}}$.
    Let $\struc[2]$ be the "generalized structure" defined as follows:
    \begin{itemize}
        \item
        $\univ{\struc[2]} \defeq \set{x \mid \tuple{r, x} \in \semPDLREwLAp{(\const{S}^{n})^*}{\struc}}$;

        \item
        $\afml_i^{\struc[2]} \defeq \semPDLREwLAp{\dia{\const{S}^{i}} \const{B}}{\struc} \cap \univ{\struc[2]}$ for $i \in \range{0}{n-1}$;

        \item
        $\const{S}^{\struc[2]} \defeq \semPDLREwLAp{\const{S}^{n}}{\struc} \cap \univ{\struc[2]}^2$;
        
        \item $\strucuniv^{\struc[2]}$ is the smallest "partial order" extending the binary relation $\const{S}^{\struc}$.
    \end{itemize}
    First, we have $\semPDLREwLAp{\tonormuni{\term}}{\struc} \subseteq \semPDLREwLAp{(\const{S}^n)^*}{\struc}$,
    by easy induction on $\term$.
    We then have
    both $\semPDLREwLAp{\tonormuni{\term}}{\struc} \cap \univ{\struc[2]}^2 = \semPDLREwLAp{\term}{\struc[2]}$
    and $\semPDLREwLAp{\tonormuni{\fml}}{\struc} \cap \univ{\struc[2]} = \semPDLREwLAp{\fml}{\struc[2]}$
    by easy mutual induction on $\term$ and $\fml$.
    Hence,
    $\fml$ is also "satisfiable" on $\GRELstfinbintree{\psig, \set{\const{S}}}$.
\end{proof}

\nointro{positive boolean formula}
\subsubsection{Reduction to the "emptiness problem" of "tree automata@alternating tree automata"}\label{section: theorem: complexity PDLREwLAp substitution-closed: to tree automata}
\begin{scope}\knowledgeimport{ATA}%
\AP
An ""alternating tree automaton"" (\reintro*\kl{ATA}) is a tuple
$\automaton = \tuple{Q, \Sigma, \delta, q_0}$,
where
$Q$ is a finite set of \AP""states"", 
$\Sigma$ is a finite set of "characters",
$\delta$ is a transition function $\delta: Q \times \Sigma \to \mathbb{B}_{+}(\set{0, 1, 2} \times Q)$, and
$q_0 \in Q$ is the initial state.
Given an input (non-empty, finite, and binary) "tree" $T \colon \set{1, 2}^* \pfun \Sigma$,
a ""run"" of $\automaton$ on $T$ is a "tree" $\trace \colon \nat^* \to Q \times \fdom(T)$ such that
for all $\gamma \in \fdom(\trace)$ with $\trace(\gamma) = \tuple{q, \word}$,
the \kl{positive boolean formula} $\delta(q, T(\word))$ is semantically equivalent to $\truec$ if
the elements in
$\set{\tuple{q', d'} \in Q \times \set{0, 1, 2} \mid \trace(\gamma d) = \tuple{q', \word \series d'} \text{ for $d \in \nat$}}$
are set to $\truec$,
where $\word \AP\intro*\series d \defeq \begin{cases}
\word d & \text{if $d \in \set{1, 2}$}\\
\word & \text{if $d = 0$}
\end{cases}$.
The \AP""tree language"" $\intro*\treelangATA(\automaton)$ is defined as follows:
\[\treelangATA(\automaton) ~\defeq~ \set{T ~\text{an input "tree"} \mid \text{there is a "run" $\trace$ of $\automaton$ on $T$ "st" $\trace(\eps) = \tuple{q_0, \eps}$}}.\]
We consider the "emptiness problem" for "ATAs"---%
given an "ATA" $\automaton$, does $\treelangATA(\automaton) = \emptyset$ hold?
For the complexity, we have:
\begin{proposition}[{\cite[Theorem 7.5.1]{comonTreeAutomataTechniques2007}}]\label{proposition: ATA EXPTIME}
    The "emptiness problem" for "ATAs" is in "EXPTIME".
\end{proposition}

In the following,
we give a reduction
from the "theory" of "\PDLREwLAp" on $\GRELstfinbintree{\set{\const{D}}, \set{\const{S}}}$
to the "emptiness problem" for "ATAs".
For a technical reason, we use \textbf{"Kleene star" ($\bl^{*}$) as a primitive operator}.
\begin{definition}\label{definition: cl'}
\AP
The ""FL-closure"" $\clex(\fml)$ of a "\PDLREwLAp" "formula" $\fml$
is the smallest set of "\PDLREwLAp" "formulas" closed under the following rules:
\begin{align*}
    &\fml \in \clex(\fml),&
    &\fml[2] \to \fml[3] \in \clex(\fml) \Longrightarrow \fml[2], \fml[3] \in \clex(\fml),\\
    &\bo{\term[1]} \fml[2] \in \clex(\fml) \Longrightarrow \fml[2] \in \clex(\fml),&
    &\bo{\term[1] \union \term[2]} \fml[2] \in \clex(\fml) \Longrightarrow \bo{\term[1]}\fml[2], \bo{\term[2]}\fml[2] \in \clex(\fml),\\
    &\bo{\term[1]^{+}} \fml[2] \in \clex(\fml) \Longrightarrow \bo{\term[1]}\bo{\term[1]^{*}}\fml[2] \in \clex(\fml),&
    &\bo{\term[1] \compo \term[2]} \fml[2] \in \clex(\fml) \Longrightarrow \bo{\term[1]}\bo{\term[2]}\fml[2] \in \clex(\fml),\\
    &\bo{\term[1]^{*}} \fml[2] \in \clex(\fml) \Longrightarrow \bo{\term[1]}\bo{\term[1]^{*}}\fml[2] \in \clex(\fml),\\
    &\bo{\term[1]^{\adom}} \fml[2] \in \clex(\fml) \Longrightarrow \bo{\term[1]}\fml[2] \in \clex(\fml),&
    &\bo{\fml[3]?} \fml[2] \in \clex(\fml) \Longrightarrow \fml[3] \in \clex(\fml),\\
    &\bo{\term[1]^{\capid}} \fml[2] \in \clex(\fml) \Longrightarrow \bo{\term[1]}\fml[2] \in \clex(\fml),&
    &\bo{\term[1]^{\capcomid}} \fml[2] \in \clex(\fml) \Longrightarrow \bo{\term[1]}\fml[2] \in \clex(\fml). \tag*{\lipicsEnd}
\end{align*}
\end{definition}
Similar to \Cref{section: definition: FL-closure}, we can alternatively define $\clex(\fml)$ as follows:
\begin{align*}
    \clex(\afml)
    &\defeq \set{\afml},&
    \clex(\fml[2] \to \fml[3])
    &\defeq \set{\fml[2] \to \fml[3]} \cup \clex(\fml[2]) \cup \clex(\fml[3]),\\
    \clex(\falsec)
    &\defeq \set{\falsec},&
    \clex(\bo{\term} \fml[2])
    &\defeq \clexbo(\term, \fml[2]) \cup \clex(\fml[2]),\span\span\\[.5ex]
    \clexbo(\aterm, \fml[2])
    &\defeq \set{\bo{\aterm}\fml[2]}, \span\span\\
    \clexbo(\term[2] \union \term[3], \fml[2])
    &\defeq \set{\bo{\term[2] \union \term[3]}\fml[2]}
    \cup \clexbo(\term[2], \fml[2])
    \cup \clexbo(\term[3], \fml[2]),\span\span\\
    \clexbo(\term[2] \compo \term[3], \fml[2])
    &\defeq \set{\bo{\term[2] \compo \term[3]}\fml[2]}
    \cup \clexbo(\term[2], \bo{\term[3]}\fml[2])
    \cup \clexbo(\term[3], \fml[2]),\span\span\\
    \clexbo(\term[2]^{+}, \fml[2])
    &\defeq \set{\bo{\term[2]^{+}}\fml[2]}
    \cup \clexbo(\term[2], \bo{\term[2]^{*}}\fml[2]),\span\span\\
    \clexbo(\term[2]^{*}, \fml[2])
    &\defeq \set{\bo{\term[2]^{*}}\fml[2]}
    \cup \clexbo(\term[2], \bo{\term[2]^{*}}\fml[2]),\span\span\\
    \clexbo(\term[2]^{\adom}, \fml[2])
    &\defeq \set{\bo{\term[2]^{\adom}}\fml[2]}
    \cup \clexbo(\term[2], \fml[2]),\span\span\\
    \clexbo(\term[2]^{\capid}, \fml[2])
    &\defeq \set{\bo{\term[2]^{\capid}}\fml[2]}
    \cup \clexbo(\term[2], \fml[2]),\span\span\\
    \clexbo(\term[2]^{\capcomid}, \fml[2])
    &\defeq \set{\bo{\term[2]^{\capcomid}}\fml[2]}
    \cup \clexbo(\term[2], \fml[2]),\span\span\\
    \clexbo(\fml[3]?, \fml[2])
    &\defeq \set{\bo{\fml[3]?}\fml[2]}
    \cup \clex(\fml[3]).\span\span
\end{align*}
Thus, $\card\clex(\fml) \le \len{\fml}$ and $\card\clexbo(\term) \le \len{\term}$, by easy induction on "expressions".
We also define:
\[\AP\intro*\clexex(\fml) \defeq \clex(\fml) \cup \set{\bo{\term^{\capid}}\fml[2], \bo{\term^{\capcomid}}\fml[2], \bo{\term}\falsec, \bo{\term^{\capid}}\falsec, \bo{\term^{\capcomid}}\falsec \mid \bo{\term}\fml[2] \in \clex(\fml)}.\]
We then have $\card \clexex(\fml) \le 6 \len{\fml}$.

Below, we give an "ATA" construction.
We use the \AP""Iverson bracket"" notation:
$\intro*\Iverson{P} \defeq \begin{cases}
    \truec & \text{if $P$ holds}\\
    \falsec & \text{otherwise}.
\end{cases}$
For $\tuple{p, \fml} \in \cl(\fml_0) \times \set{0, 1}$, we write $(\fml)_{p}$.
Intuitively, $(\fml)_1$ is used to express ``$\fml$'' and $(\fml)_0$ is used to express ``$\lnot \fml$''.
For $i \in \set{1, 2}$, we use ``$i{\downarrow}$'' for indicating that the current vertex has the $i$-th child.
\begin{definition}\label{definition: to ATA}
    For a "formula" $\fml_0 \in \GRELstfinbintree{\psig, \set{\const{S}}}$,
    the "ATA" $\automaton_{\fml_0}$ is defined as
    $\tuple{\clexex(\fml_0) \times \set{0, 1}, \wp(\psig \dcup \set{1{\downarrow}, 2{\downarrow}}), \delta, (\fml_0)_1}$,
    where the transition function $\delta$ is defined as in \Cref{figure: transition function}.
    \begin{figure}[t]
    \begin{tcolorbox}[colback=black!3, top = .3ex, bottom = .3ex, left = .2em, right = .2em]
    \vspace{-2ex}
    \scalebox{.84}{\parbox{1.18\linewidth}{%
    \begin{align*}
        \delta((\afml)_1, P) &\defeq \Iverson{\afml \in P},&
        \delta((\fml[2] \to \fml[3])_1, P) &\defeq \tuple{(\fml[2])_0, 0} \lor \tuple{(\fml[3])_1, 0},\\
        \delta((\falsec)_1, P) &\defeq \falsec,\\
        \delta((\bo{\term} \fml[2])_1, P) &\defeq \tuple{(\bo{\term^{\capid}}\fml[2])_1, 0} \land \tuple{(\bo{\term^{\capcomid}}\fml[2])_1, 0} \mbox{ if $\term$ does not match $\term[2]^{\capid}$ or $\term[2]^{\capcomid}$},\span\span \\
        \delta((\bo{\const{S}^{\capid}}\fml[2])_1, P) &\defeq \truec,\span\span\\
        \delta((\bo{(\term[2]\term[3])^{\capid}}\fml[2])_1, P) &\defeq \tuple{(\bo{\term[2]^{\capid}}\fml[2])_1, 0} \land \tuple{(\bo{\term[3]^{\capid}} \fml[2])_1, 0},\span\span\\
        \delta((\bo{(\term[2] \union \term[3])^{\capid}}\fml[2])_1, P) &\defeq \tuple{(\bo{\term[2]^{\capid}}\fml[2])_1, 0} \land \tuple{(\bo{\term[3]^{\capid}} \fml[2])_1, 0},\span\span\\
        \delta((\bo{(\term[2]^{+})^{\capid}}\fml[2])_1, P) &\defeq \tuple{(\bo{\term[2]^{\capid}}\fml[2])_1, 0},&
        \delta((\bo{(\term[2]^{*})^{\capid}}\fml[2])_1, P) &\defeq \tuple{(\fml[2])_1, 0},\span\span\\
        \delta((\bo{(\term[2]^{\adom})^{\capid}}\fml[2])_1, P) &\defeq \tuple{(\bo{\term[2]}\falsec)_0, 0} \lor \tuple{(\fml[2])_1, 0},&
        \delta((\bo{(\fml[3]?)^{\capid}}\fml[2])_1, P) &\defeq \tuple{(\fml[3])_0, 0} \lor \tuple{(\fml[2])_1, 0},\span\span\\
        \delta((\bo{(\term[2]^{\capid})^{\capid}}\fml[2])_1, P) &\defeq \tuple{(\bo{\term[2]^{\capid}}\fml[2])_1, 0},& 
        ((\bo{(\term[2]^{\capcomid})^{\capid}}\fml[2])_1, P) &\defeq \truec,\span\span\\
        \delta((\bo{\const{S}^{\capcomid}}\fml[2])_1, P) &\defeq (\Iverson{1{\downarrow} \not\in P} \lor \tuple{(\fml[2])_1, 1}) \land (\Iverson{2{\downarrow} \not\in P} \lor \tuple{(\fml[2])_1, 2}),\span\span\\
        \delta((\bo{(\term[2]\term[3])^{\capcomid}}\fml[2])_1, P) &\defeq \tuple{(\bo{\term[2]^{\capcomid}}\bo{\term[3]}\fml[2])_1, 0} \land
        (\tuple{(\bo{\term[2]^{\capid}}\falsec)_1, 0} \lor \tuple{(\bo{\term[3]^{\capcomid}}\fml[2])_1, 0}),\span\span\\
        \delta((\bo{(\term[2] \union \term[3])^{\capcomid}}\fml[2])_1, P) &\defeq \tuple{(\bo{\term[2]^{\capcomid}}\fml[2])_1, 0} \land \tuple{(\bo{\term[3]^{\capcomid}}\fml[2])_1, 0},\span\span\\
        \delta((\bo{(\term[2]^{+})^{\capcomid}}\fml[2])_1, P) &\defeq \tuple{(\bo{\term[2]^{\capcomid}}\bo{\term[2]^{*}}\fml[2])_1, 0},&
        \delta((\bo{(\term[2]^{*})^{\capcomid}}\fml[2])_1, P) &\defeq \tuple{(\bo{\term[2]^{\capcomid}}\bo{\term[2]^{*}}\fml[2])_1, 0},\span\span\\
        \delta((\bo{(\term[2]^{\adom})^{\capcomid}}\fml[2])_1, P) &\defeq \truec,&
        \delta((\bo{(\fml[3]?)^{\capcomid}}\fml[2])_1, P) &\defeq \truec,\span\span\\
        \delta((\bo{(\term[2]^{\capid})^{\capcomid}}\fml[2])_1, P) &\defeq \truec,&
        \delta((\bo{(\term[2]^{\capcomid})^{\capcomid}}\fml[2])_1, P) &\defeq \tuple{(\bo{\term[2]^{\capcomid}}\fml[2])_1, 0}.\span\span
    \end{align*}
    Additionally, $\delta((\fml)_0, P)$ is defined as the ``dual'' of the $\delta((\fml)_1, P)$,
    that is, the \kl{formula} $\delta((\fml)_1, P)$ in which
    $\tuple{(\fml)_{p}, i}$,
    $\land$,
    $\lor$,
    $\falsec$, and
    $\truec$ have been replaced with
    $\tuple{(\fml)_{1-p}, i}$,
    $\lor$,
    $\land$,
    $\truec$, and
    $\falsec$, respectively.
    }}
    \end{tcolorbox}
    \caption{Definition of the transition function $\delta$ of the \kl{ATA} $\automaton_{\fml_0}$.}
    \label{figure: transition function}
    \end{figure}
    \lipicsEnd
\end{definition}

For a (non-empty, finite, and binary) "tree" $T \colon \set{1, 2}^* \pfun \wp(\psig \dcup \set{1{\downarrow}, 2{\downarrow}})$,
we define the "generalized structure" $\struc_{T} \in \GRELstfinbintree{\psig, \set{\const{S}}}$ as follows:
\begin{itemize}
    \item $\univ{\struc_{T}}$ is the smallest set with $\eps \in \univ{\struc_{T}}$ such that
    for all $\word \in \univ{\struc_{T}}$,
    if $d{\downarrow} \in T(\word)$ and $\word d \in \fdom(T)$, then $\word d \in \univ{\struc_{T}}$.

    \item $\const{S}^{\struc_{T}} \defeq \set{\tuple{\word, \word d} \mid \word \in \set{1, 2}^*, d \in \set{1, 2}} \cap \univ{\struc_{T}}^2$,

    \item $\afml^{\struc_{T}} \defeq \set{\word \in \univ{\struc_{T}} \mid \afml \in T(\word)}$ for $\afml \in \psig$, and

    \item $\strucuniv^{\struc_{T}}$ is the smallest "partial order" extending $\const{S}^{\struc_{T}}$.
\end{itemize}
We also define:
\[\AP\intro*\runlangATA(\automaton) ~\defeq~ \set{\tuple{T, q, \word} \mid \text{there is a "run" $\trace$ of $\automaton$ on $T$ "st" $\trace(\eps) = \tuple{q, \word}$}}.\]
We observe $\treelangATA(\automaton) = \set{T \mid \tuple{T, q_0, \eps} \in \runlangATA(\automaton)}$ where $q_0$ is the initial state of $\automaton$.
By construction of $\automaton_{\fml_0}$, we have:
\begin{lemma}\label{lemma: to ATA}
    Let $\fml_0 \in \fmlclassPDLREwLAp{\psig, \set{\const{S}}}$
    and let $T$ be any input "tree".
    For every "expression" $\expr$ and $\word \in \univ{\struc_{T}}$, we have:
    \begin{enumerate}
        \item \label{lemma: to ATA fml} if $\expr = \fml$ is a "formula" "st" $\fml \in \clexex(\fml_0)$, then 
        \[\tuple{T, (\fml)_{p}, \word} \in \runlangATA(\automaton_{\fml_0})
        \quad\iff\quad 
        \word \in_{p} \semPDLREwLAp{\fml}{\struc_{T}},\]

        \item \label{lemma: to ATA term} if $\expr = \term$ is a "term" "st" $\bo{\term^{\capcomid}}\fml[2] \in \clexex(\fml_0)$ for some $\fml[2]$, then 
        \[\tuple{T, (\bo{\term^{\capcomid}}\fml[2])_{p}, \word} \in \runlangATA(\automaton_{\fml_0})
        \quad\iff\quad 
        \word \in_{p} \semPDLREwLAp{\bo{\term^{\capcomid}}\fml[2]}{\struc_{T}}.\]
    \end{enumerate}
    Here, $(\in_{1}) = (\in)$ and $(\in_{0}) = (\not\in)$.
\end{lemma}
\begin{proof}
    By induction on the pair of the reversed ordering of $\strucuniv^{\struc_{T}} \setminus \diagonal_{\univ{\struc_{T}}}$ and the size of $\expr$.

    \proofcase{For \ref{lemma: to ATA fml}}
    \proofcase{Case $\fml = \afml, \fml[2] \to \fml[3]$}
    Easy.
    \proofcase{Case $\fml = \bo{\term^{\capid}} \fml[2]$}
    All the cases follow easily from the semantics and IH.
    \proofcase{Case $\fml = \bo{\term^{\capcomid}} \fml[2]$}
    By IH of \ref{lemma: to ATA term}.
    \proofcase{Case $\fml = \bo{\term} \fml[2]$}
    By the same argument as the cases of $\bo{\term^{\capid}} \fml[2]$ and $\bo{\term^{\capcomid}} \fml[2]$,
    we can show this case.

    \proofcase{For \ref{lemma: to ATA term}}
    Many cases follow easily from the semantics and IH.
    We write some selected cases.

    \proofcase{Case $\term = \const{S}$}
    When $p = 1$, $1{\downarrow} \in T(\word)$, and $2{\downarrow} \not\in T(\word)$ (similarly for the other cases), we have:
    \begin{align*}
    &\tuple{T, (\bo{\const{S}^{\capcomid}}\fml[2])_1, \word} \in \runlangATA(\automaton_{\fml_0})\\
    &\iff
    (\mbox{false} \mbox{ or } \tuple{T, (\bo{\const{S}^{\capcomid}}\fml[2])_1, \word 1} \in \runlangATA(\automaton_{\fml_0}))\\
    &\hspace{5em} \mbox{ and }
    (\mbox{true} \mbox{ or } \tuple{T, (\bo{\const{S}^{\capcomid}}\fml[2])_1, \word 2} \in \runlangATA(\automaton_{\fml_0}))
    \tag{By definition of $\delta$}\\
    &\iff
    \tuple{T, (\bo{\const{S}^{\capcomid}}\fml[2])_1, \word 1} \in \runlangATA(\automaton_{\fml_0})\\
    &\iff
    \word 1 \in \semPDLREwLAp{\fml[2]}{\struc_{T}} \tag{By IH "wrt" the first parameter}\\
    &\iff \word \in \semPDLREwLAp{\bo{\const{S}^{\capcomid}}\fml[2]}{\struc_{T}}. \tag{The only $\const{S}$-successor of $\word$ is $\word 1$}
    \end{align*}

    \proofcase{Case $\term = \term[2] \term[3]$}
    When $p = 1$ (similarly for the case of $p = 0$), we have:
    \begin{align*}
    &\tuple{T, (\bo{(\term[2]\term[3])^{\capcomid}}\fml[2])_1, \word} \in \runlangATA(\automaton_{\fml_0})\\
    &\iff
    \tuple{T, (\bo{\term[2]^{\capcomid}}\bo{\term[3]}\fml[2])_1, \word} \in \runlangATA(\automaton_{\fml_0})
    \mbox{ and }\\
    & \hspace{5em} (\tuple{T, (\bo{\term[2]^{\capid}}\falsec)_0, \word} \in \runlangATA(\automaton_{\fml_0})
    \mbox{ or }
    \tuple{T, (\bo{\term[3]^{\capcomid}}\fml[2])_1, \word} \in \runlangATA(\automaton_{\fml_0})) \tag{By definition of $\delta$}\\
    &\iff
    \word \in \semPDLREwLAp{\bo{\term[2]^{\capcomid}}\bo{\term[3]}\fml[2]}{\struc_{T}} \mbox{ and } (\word \not\in \semPDLREwLAp{\bo{\term[2]^{\capid}}\falsec}{\struc_{T}}
    \mbox{ or }
    \word \in \semPDLREwLAp{\bo{\term[3]^{\capcomid}}\fml[2]}{\struc_{T}}) \tag{By IH}\\
    &\iff
    \word \in \semPDLREwLAp{\bo{\term[2]^{\capcomid}}\bo{\term[3]}\fml[2]}{\struc_{T}} \mbox{ and } \word \not\in \semPDLREwLAp{\bo{\term[2]^{\capid}}\bo{\term[3]^{\capcomid}}\fml[2]}{\struc_{T}} \\
    &\iff
    \word \in \semPDLREwLAp{\bo{(\term[2]\term[3])^{\capcomid}}\fml[2]}{\struc_{T}}. 
    \end{align*}
    
    \proofcase{Case $\term = \term[2]^*$}
    When $p = 1$ (similarly for the case of $p = 0$), we have:
    \begin{align*}
    \tuple{T, (\bo{(\term[2]^*)^{\capcomid}}\fml[2])_1, \word} \in \runlangATA(\automaton_{\fml_0})
    &\iff \tuple{T, (\bo{\term[2]^{\capcomid}}\bo{\term[2]^*} \fml[2])_1, \word} \in \runlangATA(\automaton_{\fml_0}) \tag{By definition of $\delta$}\\
    &\iff \word \in \semPDLREwLAp{\bo{\term[2]^{\capcomid}}\bo{\term[2]^*} \fml[2]}{\struc_{T}} \tag{By IH}\\
    &\iff \word \in \semPDLREwLAp{\bo{(\term[2]^*)^{\capcomid}}\fml[2]}{\struc_{T}}.
    \end{align*}

    Hence, the proof is completed.
\end{proof}
Hence, we have:
$\treelangATA(\automaton_{\fml_0}) \neq \emptyset$
"iff"
$\fml_0$ is "satisfiable" on $\GRELstfinbintree{\psig, \set{\const{S}}}$.
\begin{proof}[Proof for the upper bound of \Cref{theorem: complexity PDLREwLAp substitution-closed}]
By combining the reductions in this section,
there is a polynomial-time reduction from the "theory" of "\PDLREwLAp" on $\GRELfinlin{}$
to "that@theory" on $\GRELstfinbintree{\set{\const{D}}, \set{\const{S}}}$.
By \Cref{lemma: to ATA},
we can give a polynomial-time this "theory" to
the "emptiness problem" of "ATAs".
Thus by \Cref{proposition: ATA EXPTIME}, we have obtained the "EXPTIME" upper bound.
\end{proof}
\end{scope}

\subsection{Proof of the upper bound of \Cref{theorem: complexity PDLREwLAp standard}}\label{section: theorem: complexity PDLREwLAp standard}
\begin{scope}\knowledgeimport{ASA}%
An \AP""alternating string automaton"" (\reintro*\kl{ASA}) is a tuple
$\automaton = \tuple{Q, \Sigma, \delta, q_0}$,
where
$Q$ is a finite set of "states", 
$\Sigma$ is a finite set of "characters",
$\delta$ is a transition function $\delta: Q \times \Sigma \to \mathbb{B}_{+}(\set{0, 1} \times Q)$,
$q_0 \in Q$ is the initial state.
Given an input (non-empty and finite) ``"string"'' $T \colon \set{1}^* \pfun \Sigma$,
a ""run"" of $\automaton$ on $T$ is a "tree" $\trace \colon \nat^* \to Q \times \fdom(T)$ such that
for all $\gamma \in \fdom(\trace)$ with $\trace(\gamma) = \tuple{q, \word}$,
the \kl{positive boolean formula} $\delta(q, T(\word))$ is semantically equivalent to $\truec$ if
the elements in
$\set{\tuple{q', d'} \in Q \times \set{0, 1} \mid \trace(\gamma d) = \tuple{q', \word \series d'} \text{ for $d \in \nat$}}$
are set to $\truec$,
where $\word \series d \defeq \begin{cases}
\word d & \text{if $d = 1$}\\
\word & \text{if $d = 0$}.
\end{cases}$
\AP
The ""string language"" $\intro*\wlangASA(\automaton)$ is defined as follows:
\[\wlangASA(\automaton) ~\defeq~ \set{T ~\text{an input "string"} \mid \text{there is a "run" $\trace$ of $\automaton$ on $T$ "st" $\trace(\eps) = \tuple{q_0, \eps}$}}.\]
We consider the "emptiness problem" for "ASAs"---%
given an "ASA" $\automaton$, does $\wlang(\automaton) = \emptyset$ hold?
For the complexity, we have:
\begin{proposition}[{\cite[Theorem 3.1]{jiangNoteSpaceComplexity1991}}]\label{proposition: ASA PSPACE}
    The "emptiness problem" for "ASAs" is in "PSPACE".
\end{proposition}

\begin{theorem*}[Restatement of \Cref{theorem: complexity PDLREwLAp standard}]
\theoremcomplexityPDLREwLApstandard
\end{theorem*}
\begin{proof}[Proof Sketch for the upper bound of \Cref{theorem: complexity PDLREwLAp standard}]
Using the same reductions from \Cref{lemma: reduction to one formula variable,lemma: reduction to one term variable},
we can give a polynomial-time reduction from the "theory" of "\PDLREwLAp" on $\GRELstfinlin{}$
to "that@theory" on $\GRELstfinlin{\set{\const{D}}, \set{\const{S}}}$.
By the same automata construction as \Cref{lemma: to ATA}
(where we redefine $\delta((\bo{\const{S}^{\capcomid}}\fml[2])_1, P)$
as $(\Iverson{1{\downarrow} \not\in P} \lor \delta((\fml[2])_1, P))$),
we can give a polynomial-time this "theory" to
the "emptiness problem" of "ASAs".
Thus by \Cref{proposition: ASA PSPACE}, we have obtained the "PSPACE" upper bound.
\end{proof}
\end{scope}

\end{scope}

\end{document}